\documentclass[journal]{IEEEtran}
\usepackage[dvips]{graphicx}
\usepackage{epsfig}
\usepackage{amsmath}
\usepackage{epsfig}
\usepackage{array}
\usepackage{amssymb}
\usepackage{color}
\newtheorem{Proposition}{\textbf{Proposition}}

\newtheorem{Lemma}{\textbf{Lemma}}

\usepackage{subcaption}
\def\BibTeX{{\rm B\kern-.05em{\sc i\kern-.025em b}\kern-.08em
    T\kern-.1667em\lower.7ex\hbox{E}\kern-.125emX}}
\begin{document}

\title{Investigation of STEEP for Secure Communications Over SIMO and MISO Channels Subject to Full-Duplex Jamming and Eavesdropping}

\author{Md Saydur Rahman,~\IEEEmembership{Member,~IEEE}, and Yingbo Hua,~\IEEEmembership{Life Fellow,~IEEE}\thanks{Department of Electrical and Computer Engineering, University of California at Riverside,
Riverside, CA, 92521, USA. Emails: mrahm054@ucr.edu,
 yhua@ece.ucr.edu (Corresponding Author).  This work was supported in part by the Department of Defense under W911NF-20-2-0267. The views and conclusions contained in this
document are those of the authors and should not be interpreted as representing the official policies, either
expressed or implied, of  the U.S. Government. The U.S. Government is
authorized to reproduce and distribute reprints for Government purposes notwithstanding any copyright
notation herein.}}

\maketitle

\begin{abstract}
Secure communications over SIMO and MISO channels between wireless nodes are commonly encountered in applications and widely considered in the literature. This paper investigates such a problem by considering a newly proposed scheme called secrecy-message transmission by echoing encrypted probes (STEEP). We focus on a type of heavily eavesdropped channels where an aggressive eavesdropper (Eve)  has multiple antennas and is capable of jamming and eavesdropping in the full-duplex mode. Assuming optimal jamming and eavesdropping by Eve, we present the secrecy rates achievable by STEEP in two forms: using SIMO channel in phase 1 and using MISO channel in phase 2, and vice versa. With any given channel condition between every pair of  nodes and any finite jamming power by Eve, the secrecy rate of STEEP in either form is shown to be positive subject to a positive power in phase 1 and a sufficiently large power in phase 2. The connections between STEEP in either form and the conventional SIMO and MISO schemes are provided, and their secrecy rates are compared comprehensively subject to random fading channels and finite transmit powers.
\end{abstract}

\begin{IEEEkeywords}
Secure communications, physical layer security, secret-message transmission, heavily eavesdropped channels.
\end{IEEEkeywords}

\section{Introduction}
In recent years, physical layer security has received tremendous interests in various wireless network settings such as those involving
intelligent reflecting surfaces (IRS) \cite{Thien2022, Zhang2024, Wang2021, Cao2023}, multi-antenna relays \cite{Lv2022, Huang2021, Cao2021, Qian2023}, LEO and GEO satellites \cite{Talgat2024, Jiang2023, ZishengYin2022, LeiXu2024}, and unmanned aerial vehicles (UAVs) \cite{Tingtingli2021,huiciwu2022, YalinZhang2025, huiciwu2020, DanyuDiao2022, GaofengPan2020, ChenxiLiu_2019,chaohan2022,tatarmamaghani2021,HongjiangLei2020,Abdalla2023}.

A primary objective of all of these works is to maintain or maximize a positive secrecy rate from one legitimate node (Alice) to another (Bob). In most of these cases, the eavesdropper (Eve) is assumed to be a single-antenna receiver while one of the legitimate nodes or their helper (such as IRS or relay) may offer a higher spatial diversity. In all of these cases, the overall channel from Alice to Bob is assumed to be stronger than that from Alice to Eve, or otherwise a positive secrecy rate is not achievable by any of those schemes.

However, a recently proposed scheme called secret-message transmission by echoing encrypted probes (STEEP) \cite{Hua_STEEP_2025}, \cite{Hua_Rahman_2024_ICC}, \cite{Hua_Rahman_Swami_2024}, \cite{Rahman_Hua_2024} can achieve a positive secrecy rate from Alice to Bob, or vice versa, without the need for user's receive channels to be stronger than Eve's receive channels. Such a positive secrecy rate is achieved even if the number of antennas on Eve is larger than that on any of the users and even if there is no additional helping node. A comprehensive analysis of STEEP for MIMO channels is available in \cite{Hua_STEEP_2025}, and some of the numerical results on STEEP for MISO channels have been reported in \cite{Hua_Rahman_2024_ICC}, \cite{Hua_Rahman_Swami_2024}, \cite{Rahman_Hua_2024}. As shown in \cite{Hua_STEEP_2025}, STEEP is an important extension from an encryption lemma \cite{Bloch2011}, evolved from \cite{Hua2023}, and is  much more general than a binary signaling (BS) scheme shown in \cite{Hayashi_2020}.
 The BS scheme is the most straightforward application of the encryption lemma, but corresponds to the worst case of a phase-shift-keying (PSK) based STEEP as shown in \cite{Rafique_Hua_2025}. For STEEP over Gaussian channels, Gaussian signaling is known to be optimal at asymmetric large powers \cite{Hua_STEEP_2025}.

The goal of this paper is to present further results on STEEP for the cases where the channel from one user to another is either SIMO or MISO. This is a common setting considered in many of the  prior works such as \cite{Thien2022}-\cite{Abdalla2023}. Indeed, most of the IoT devices are of a single antenna, but larger platforms such as UAVs may have multiple antennas. Specifically, we assume that Alice has multiple antennas while Bob has a single antenna. We also assume that Eve has multiple antennas and could perform jamming and eavesdropping in full-duplex. This is in contrast to most prior works where Eve has a single antenna and is completely passive.  We will compare the performance of STEEP with that of the conventional wiretap channel scheme subject to the same power consumptions from both users. Although the system model considered in this paper is somewhat similar to that in \cite{ChenxiLiu_2019}, the scheme considered here is much more robust against Eve's channel strength. This work is a substantially expanded version of our prior work in \cite{Rahman_Hua_2024}.

Although the theory and methods of STEEP for Gaussian MIMO channels shown in \cite{Hua_STEEP_2025} apply directly to SIMO and MISO channels, that work does not make any explicit use of artificial noise. Furthermore, the optimal jamming from Eve was not considered in \cite{Hua_STEEP_2025}. Consequently, most of the results in this paper cannot be inferred from those in \cite{Hua_STEEP_2025} or other related prior works \cite{Hua_Rahman_2024_ICC}-\cite{Rahman_Hua_2024}. Although a high level conclusion in regard to achieving a positive secrecy rate subject to a large phase-2 power is consistent, a separate analytical treatment has to be done to establish the claim.

More specifically, the novel contributions in this paper include the following:
\begin{enumerate}
  \item Detailed analytical results of STEEP for secure communications over SIMO and MISO channels subject to full-duplex jamming and eavesdropping from a multi-antenna Eve are shown for the first time (see sections \ref{sec:SIMO-MISO-STEEP} and \ref{sec:MISO-SIMO-STEEP});
  \item Optimal jamming from the multi-antenna Eve against both single-antenna user and  multi-antenna user  is incorporated for the first time into the analysis and results (see section \ref{sec:jamming} for jamming from Eve to Alice);
  \item It is shown that under virtually any channel conditions at users and Eve, including jamming from full-duplex Eve, the secrecy rate of either SIMO-MISO-STEEP or MISO-SIMO-STEEP can be made positive by  choosing a positive power in phase 1 and a sufficiently large power in phase 2. This is in contrast to the zero secrecy rate of the conventional SIMO and MISO schemes within a large space of channel conditions at users and Eve.
  \item Extensive simulation results  are for the first time provided to illustrate the secrecy-rate advantages of  the two forms of STEEP over the conventional SIMO and MISO schemes widely adopted in the literature (see section \ref{sec:Simulation}).
\end{enumerate}

In section \ref{sec:system_model}, the system model with a multi-antenna Alice, a single-antenna Bob and a multi-antenna Eve is described. Two forms of STEEP are defined at a high level, which are  SIMO-MISO-STEEP and MISO-SIMO-STEEP.

In section \ref{sec:SIMO_conv}, the conventional SIMO scheme is shown. In section \ref{sec:SIMO-MISO-STEEP}, the SIMO-MISO-STEEP scheme is introduced and analyzed. Together, these two sections provide a strong connection between the conventional SIMO scheme and the SIMO-MISO-STEEP scheme. But the former is meant to transmit a secret message from Bob to Alice while the latter is meant to transmit a secret message from Alice to Bob.

In section \ref{sec:MISO_conv}, the conventional MISO scheme is discussed. In section \ref{sec:MISO-SIMO-STEEP}, the MISO-SIMO-STEEP scheme is then introduced and analyzed. A strong connection between these two schemes becomes obvious there. Again, these two schemes are meant to transmit secret messages in opposite directions.

In section \ref{sec:complexity}, the major steps in the two forms of STEEP are highlighted so that their complexity becomes clear. Other issues of importance are also discussed in that section.

In section \ref{sec:Simulation}, the statistical behaviors of the secrecy rates  of SIMO-MISO-STEEP and MISO-SIMO-STEEP, subject to random fading channels and finite transmit powers,  are illustrated. They are also compared with the secrecy rates of the conventional SIMO and MISO schemes.

Section \ref{sec:conclusion} concludes the paper.

\section{System Model and Preliminaries}\label{sec:system_model}
We assume that Alice has $n_A\geq 1$ antennas, Bob has a single antenna (or $n_B=1$), and Eve has $n_E\geq 1$ antennas used in full-duplex (FD) mode for jamming and eavesdropping.

A practical FD radio suffers from a residual self-interference (SI) with power $p_{\texttt{SI}}$. Despite any methods for SI cancellation, the resulting $p_{\texttt{SI}}$ is proportional to its transmitted jamming power $P_E$, i.e., $p_{\texttt{SI}}=\rho_{\texttt{SI}} P_E$, with $\rho_{\texttt{SI}}$ being typically between $-30$dB and $-120$dB. An effective jamming power received by Alice or Bob is $p_E=\frac{P_E}{D^\alpha_p}$, with $D$ being the corresponding transmission distance and $\alpha_p\geq 2$ being a path loss exponent. For FD Eve to be able to eavesdrop on Alice or Bob, $p_{\texttt{SI}}$ must be limited and hence so must be $P_E$ and $p_E$. On the other hand, the amount of power transmitted by Alice or Bob does not suffer from the limitation of SI. In this paper, we will assume a finite $p_E$ and ignore $p_{\texttt{SI}}$, the latter of which, despite whether $p_{\texttt{SI}}$ is larger  or smaller than the noise variance at Eve, only leads to a conservative (or good) measure of secrecy rates against Eve.  Also note that  all power values to be referred to from now on are normalized by transmission distances.


When Alice transmits, the channel between the users is MISO with its channel vector denoted by $\mathbf{h}_{B,A}\in\mathbb{C}^{n_A\times 1}$, and the channel from Alice to Eve is MIMO with its channel matrix denoted by $\mathbf{H}_{E,A}\in\mathbb{C}^{n_E\times n_A}$. At the same time, the jamming channel from Eve to Bob is also MISO with its channel vector denoted by $\mathbf{h}_{B,E}\in\mathbb{C}^{n_E\times 1}$.

When Bob transmits, the channel between the users is SIMO with the channel vector denoted by $\mathbf{h}_{A,B}\in\mathbb{C}^{n_A\times 1}$, the channel from Bob to Eve is also SIMO with the channel vector denoted by $\mathbf{h}_{E,B}\in\mathbb{C}^{n_E\times 1}$, Furthermore, the jamming channel from Eve to Alice is MIMO with the channel matrix denoted by $\mathbf{H}_{A,E}\in\mathbb{C}^{n_A\times n_E}$.

We will assume that all signals and (additive) noises are circular Gaussian distributed. Without loss of generality, we let the variance of each complex noise component be normalized to one.

As shown in \cite{Hua_STEEP_2025}, STEEP is a round-trip scheme with phases 1 and 2. In phase 1, one of the users sends a sequence of (uncoded) random probing symbols. In phase 2, the other user sends a sequence of  secret-message symbols combined with estimated probing symbols. As discussed later, the two packets transmitted in the two phases can be overlapped in time using frequency division duplex (FDD) to reduce latency. When Alice transmits in phase 1 and Bob transmits in phase 2, we will call it MISO-SIMO-STEEP. If the transmission order is reversed, we will call it SIMO-MISO-STEEP.

It will become clear that MISO-SIMO-STEEP is meant for Bob (with single antenna) to send a secret message to Alice (with multiple antennas) while SIMO-MISO-STEEP is meant for Alice to send a secret message to Bob.

Unless mentioned or clarified otherwise, the transmit powers from Alice, Bob and Eve will be denoted by $p_A$, $p_B$ and $p_E$ respectively. We will also use $p_{E,1}$ and $p_{E,2}$ for the transmit powers by Eve in phases 1 and 2 respectively.

The optimal jamming scheme by Eve against Bob is straightforward due to the MISO jamming channel. But the optimal jamming scheme by Eve against Alice is less obvious due to the MIMO jamming channel, and hence will be treated separately. We also assume that Eve knows all channel parameters including users'.

When Alice transmits a signal, we assume that Alice also transmits an artificial noise. We will point out in section \ref{sec:MISO-SIMO-STEEP} how this signaling is related to that of STEEP for the MISO case in \cite{Hua_Rahman_2024_ICC}  where Alice transmits a sequence of random probing vectors (without explicit artificial noise) and Bob estimates an effective stream of probing symbols.

The following two facts will be used frequently without necessarily being mentioned.

\emph{MMSE (minimum mean squared error) Estimation}: Given zero mean jointly Gaussian random vectors $\mathbf{x}$ and $\mathbf{y}$. The MMSE estimate $\mathbf{\hat x}$ of $\mathbf{x}$ from $\mathbf{y}$ is $\mathbf{\hat x} =\mathbb{E}\{\mathbf{x}|\mathbf{y}\} =\mathbb{E}\{\mathbf{x}\mathbf{y}^H\}(\mathbb{E}\{\mathbf{y}\mathbf{y}^H\})^{-1}\mathbf{y}$. The MSE of $\mathbf{\hat x}$ is $\mathbb{E}\{\|\mathbf{\hat x}-\mathbf{x}\|^2\}
=\mathbb{E}\{\|\mathbf{x}\|^2\}-\mathbb{E}\{\|\mathbf{\hat x}\|^2\}=\mathbb{E}\{\|\mathbf{x}\|^2\}
-\texttt{Tr}\left (\mathbb{E}\{\mathbf{x}\mathbf{y}^H\}(\mathbb{E}\{\mathbf{y}\mathbf{y}^H\})^{-1}
(\mathbb{E}\{\mathbf{x}\mathbf{y}^H\})^H\right )$.

\emph{Matrix Inverse Lemma}: Let $\mathbf{A}$ and $\mathbf{D}$ be invertible. Then
$
  (\mathbf{A}+\mathbf{U}\mathbf{D}\mathbf{V})^{-1}
=\mathbf{A}^{-1}-\mathbf{A}^{-1}\mathbf{U}\mathbf{B}^{-1}\mathbf{V}\mathbf{A}^{-1}
$
with $\mathbf{B}=\mathbf{D}^{-1}+\mathbf{V}\mathbf{A}^{-1}\mathbf{U}$.

We use $\mathbb{E}$ for expectation, $\mathbb{I}$ for mutual information, $\mathbb{H}$ for entropy, $\mathcal{CN}$ for circular Gaussian, $\log=\log_2$, and $x^+\doteq\max(x,0)$.

In the next four sections, we discuss  the conventional SIMO, SIMO-MISO-STEEP, conventional MISO and MISO-SIMO-STEEP schemes in order.

\section{Conventional SIMO scheme}\label{sec:SIMO_conv}
In this section, we present an optimal conventional transmission scheme over a SIMO channel which is jammed at and eavesdropped on by a multi-antenna full-duplex Eve. We assume that Eve knows the channel response between users and applies the optimal jamming. The secrecy rate $R_{\texttt{s,SIMO,conv}}$ of this scheme is derived and analyzed. It will be shown that $R_{\texttt{s,SIMO,conv}}^+$ is zero at any transmission power from user whenever the eavesdropping channel is strong enough.

In the conventional SIMO scheme, Bob sends a sequence of random symbols $\sqrt{p_B}r_B(k)$, $k=1,\cdots,K$, which \emph{is encoded} with a secret message. Then the corresponding signals received by Alice and Eve are
\begin{align}\label{eq:yA1k}
  &\mathbf{y}_{A,1}(k) = \sqrt{p_B}\mathbf{h}_{A,B}r_B(k)+\mathbf{w}_{A,1}(k)\notag\\
  &\,\,+\sqrt{p_{E,1}}
  \mathbf{H}_{A,E}\mathbf{v}_{A,E}n_{E,1}(k),
\end{align}
\begin{equation}\label{eq:yE1k}
  \mathbf{y}_{E,1}(k)=\sqrt{p_B}\mathbf{h}_{E,B}r_B(k) +\mathbf{w}_{E,1}(k),
\end{equation}
where $\sqrt{p_{E,1}}\mathbf{v}_{A,E}n_{E,1}(k)$ is the jamming noise from Eve.

As mentioned earlier, all signals, channel noises and jamming noises are normalized white circular complex Gaussian. Here  $r_B(k)\sim \mathcal{CN}(0,1)$, $\mathbf{w}_{A,1}(k)\sim\mathcal{CN}(0,\mathbf{I})$,  $n_{E,1}(k)\sim \mathcal{CN}(0,1)$, and $\mathbf{w}_{E,1}(k)\sim \mathcal{CN}(0,\mathbf{I})$.

\subsection{MMSE estimation of $r_B(k)$ by Alice}
To receive the message from Bob, Alice needs to compute the MMSE estimate of $r_B(k)$, denoted by $\hat r_{B|A}(k)$, from $\mathbf{y}_{A,1}(k)$ in \eqref{eq:yA1k}. It follows that
\begin{align}\label{eq:hrBAk}
  &\hat r_{B|A}(k) =\sqrt{p_B}\mathbf{h}_{A,B}^H\left (p_B\mathbf{h}_{A,B}\mathbf{h}_{A,B}^H+\mathbf{I}\right .\notag\\
  &\,\,\left .+p_{E,1}\mathbf{H}_{A,E}\mathbf{v}_{A,E}\mathbf{v}_{A,E}^H\mathbf{H}_{A,E}^H\right )^{-1}\mathbf{y}_{A,1}(k),
\end{align}
the MSE of which is
\begin{align}\label{eq:sDrBA1}
  &\sigma_{\Delta r_{B|A}}^2 =1-p_B\mathbf{h}_{A,B}^H\left (p_B\mathbf{h}_{A,B}\mathbf{h}_{A,B}^H+\mathbf{I}\right .\notag\\
  &\,\,\left .+p_{E,1}\mathbf{H}_{A,E}\mathbf{v}_{A,E}\mathbf{v}_{A,E}^H\mathbf{H}_{A,E}^H\right )^{-1}\mathbf{h}_{A,B}.
\end{align}

To reveal a useful structure of $\sigma_{\Delta r_{B|A}}^2$, we now apply the matrix inverse lemma to obtain:
\begin{align}\label{eq:sDrBA2}
  &\sigma_{\Delta r_{B|A}}^2=1-p_B\mathbf{h}_{A,B}^H\left (\mathbf{A}_1^{-1}\right .\notag\\
  &\,\,\left. -p_{E,1}\mathbf{A}_1^{-1}\mathbf{H}_{A,E}\mathbf{v}_{A,E}b_1^{-1}
  \mathbf{v}_{A,E}^H\mathbf{H}_{A,E}^H\mathbf{A}_1^{-1}\right )\mathbf{h}_{A,B},
\end{align}
with $\mathbf{A}_1=p_B\mathbf{h}_{A,B}\mathbf{h}_{A,B}^H+\mathbf{I}$ and $b_1=1+
p_{E,1}\mathbf{v}_{A,E}^H\mathbf{H}_{A,E}^H\mathbf{A}_1^{-1}\mathbf{H}_{A,E}\mathbf{v}_{A,E}$. Since $\mathbf{h}_{A,B}^H\mathbf{A}_1^{-1}=\frac{1}{p_B\|\mathbf{h}_{A,B}\|^2+1}\mathbf{h}_{A,B}^H$ and $\mathbf{A}_1^{-1}=\mathbf{I}-\alpha_B
\mathbf{\bar h}_{A,B}\mathbf{\bar h}_{A,B}^H$ with $\alpha_B=\frac{p_B\|\mathbf{h}_{A,B}\|^2}{p_B\|\mathbf{h}_{A,B}\|^2+1}$ and  $\mathbf{\bar h}_{A,B}\doteq\frac{1}{\|\mathbf{h}_{A,B}\|}\mathbf{h}_{A,B}$, we have
\begin{equation}\label{}
  b_1=1+p_{E,1}\left (\|\mathbf{H}_{A,E}\mathbf{v}_{A,E}\|^2-\alpha_B|\mathbf{\bar h}_{A,B}^H\mathbf{H}_{A,E}\mathbf{v}_{A,E}|^2
  \right ).
\end{equation}
Hence,
\begin{align}\label{eq:sDrBA3}
  &\sigma_{\Delta r_{B|A}}^2=1-p_B\left (\frac{\|\mathbf{h}_{A,B}\|^2}{p_B\|\mathbf{h}_{A,B}\|^2+1}
  -t_1\frac{\|\mathbf{h}_{A,B}\|^2}{(p_B\|\mathbf{h}_{A,B}\|^2+1)^2}\right )\notag\\
  &=\frac{1}{p_B\|\mathbf{h}_{A,B}\|^2+1}\left (1+
  t_1\alpha_B\right ),
\end{align}
 where $t_1=\frac{p_{E,1}|\mathbf{\bar h}_{A,B}^H\mathbf{H}_{A,E}\mathbf{v}_{A,E}|^2}{b_1}$.

  We see that if $p_B>0$, then $\alpha_B>0$ and $\sigma_{\Delta r_{B|A}}^2<1$. In this case, only $t_1$ depends on the jamming power $p_{E,1}$, and $t_1=0$ if and only if $p_{E,1}=0$.

 \subsection{Optimal jamming by Eve}\label{sec:jamming}
 It follows from \eqref{eq:sDrBA3} that
 the optimal $\mathbf{v}_{A,E}$ for Eve is such that $t_1$ is maximized, i.e.,
 \begin{equation}\label{}
   \mathbf{v}_{A,E}=arg \max_{\mathbf{v},\|\mathbf{v}\|=1}t_1'
 \end{equation}
 with $t_1'=\frac{p_{E,1}|\mathbf{\bar h}_{A,B}^H\mathbf{H}_{A,E}\mathbf{v}|^2}
   {1+p_{E,1}\left (\|\mathbf{H}_{A,E}\mathbf{v}\|^2-\alpha_B|\mathbf{\bar h}_{A,B}^H\mathbf{H}_{A,E}\mathbf{v}|^2
  \right )}$. A property of the denominator of $t_1'$ is shown next.

\begin{Proposition}
  $\|\mathbf{H}_{A,E}\mathbf{v}\|^2-\alpha_B|\mathbf{\bar h}_{A,B}^H\mathbf{H}_{A,E}\mathbf{v}|^2>0$.
\end{Proposition}
  \begin{IEEEproof}
  Let $\mathbf{H}=\mathbf{H}_{A,E}$ and $\mathbf{u}=\mathbf{\bar h}_{A,B}$. Write the SVD of $\mathbf{H}$ as $\mathbf{H}=\sum_{i=1}^r\sigma_i\mathbf{u}_i\mathbf{v}_i^H$ where $\sigma_1>\cdots>\sigma_r>0$. For any $\mathbf{v}=\sum_{i=1}^r c_i\mathbf{v}_i$,  $|\mathbf{u}^H\mathbf{H}\mathbf{v}|^2
  =|\sum_{i=1}^r c_i \sigma_i \mathbf{u}^H\mathbf{u}_i|^2
  <\max_i |c_i|^2\sigma_i^2$, and $\|\mathbf{H}\mathbf{v}\|^2=\|\sum_i c_i\sigma_i\mathbf{u}_i\|^2=\sum_i|c_i|^2\sigma_i^2$. Hence, $\|\mathbf{H}\mathbf{v}\|^2-\alpha_B|\mathbf{u}^H\mathbf{H}_{A,E}\mathbf{v}|^2
  >\sum_i|c_i|^2\sigma_i^2-\alpha_B \max_i |c_i|^2\sigma_i^2>0$ where $\alpha_B<1$.
  \end{IEEEproof}

  Let
  \begin{equation}\label{}
    t_1''\doteq\lim_{p_{E,1}\to\infty}t_1' = \frac{|\mathbf{\bar h}_{A,B}^H\mathbf{H}_{A,E}\mathbf{v}|^2}
   {\|\mathbf{H}_{A,E}\mathbf{v}\|^2-\alpha_B|\mathbf{\bar h}_{A,B}^H\mathbf{H}_{A,E}\mathbf{v}|^2
  }.
  \end{equation}
  Here $t_1''$ is invariant to scaling on $\mathbf{v}$.

  \begin{Proposition}
  Let $\mathbf{\hat v}$ be the optimal $\mathbf{v}$ for large $p_{E,1}$, i.e., $\mathbf{\hat v}=arg\max_{\|\mathbf{v}\|=1} t_1''$. Then the following statements hold:
  \begin{itemize}
    \item  If $n_E\geq n_A$, then $\mathbf{\hat v}=\frac{1}{\|\mathbf{H}_{A,E}^\dag \mathbf{\bar h}_{A,B}\|}\mathbf{H}_{A,E}^\dag \mathbf{\bar h}_{A,B}$.
    \item If $n_A>n_E$, then $\mathbf{\hat v}$ is among the eigenvectors of $\mathbf{G}\doteq
(\mathbf{H}_{A,E}^H\mathbf{H}_{A,E})^{-1}\mathbf{e}\mathbf{e}^H$, which gives the largest $|\mathbf{e}^H\mathbf{v}|^2$. Here $\mathbf{e} = \mathbf{H}_{A,E}^H\mathbf{\bar h}_{A,B}$.
  \end{itemize}
  \end{Proposition}
  \begin{IEEEproof}
  If $n_E\geq n_A$, $\texttt{range}(\mathbf{H}_{A,E})=\mathcal{C}^{n_A\times 1}$, i.e., for any $\mathbf{g}\in \mathcal{C}^{n_A\times 1}$, there is a $\mathbf{v}$ such that $\mathbf{g}=\mathbf{H}_{A,E}\mathbf{v}$. Since $t_1''$ is also invariant to scaling on $\mathbf{H}_{A,E}\mathbf{v}$, the optimal $\mathbf{g}$ that maximizes $t_1''$ can be found by
  \begin{equation}\label{}
    \max_\mathbf{g} |\mathbf{\bar h}_{A,B}^H\mathbf{g}|^2
  \end{equation}
  subject to $\|\mathbf{g}\|=1$. This solution is $\mathbf{\hat g}=\mathbf{\bar h}_{A,B}$. Hence the corresponding optimal solution for $\mathbf{v}$ without norm constraint is
  from the solution of $\mathbf{\bar h}_{A,B}=\mathbf{H}_{A,E}\mathbf{v}$, which is $\mathbf{H}_{A,E}^\dag \mathbf{\bar h}_{A,B}$. With the unit norm constraint, the optimal $\mathbf{v}$ becomes $\mathbf{\hat v} = \frac{1}{\|\mathbf{H}_{A,E}^\dag \mathbf{\bar h}_{A,B}\|}\mathbf{H}_{A,E}^\dag \mathbf{\bar h}_{A,B}$.

  If $n_A>n_E$, the optimal $\mathbf{v}$ without norm constraint can be found by
    \begin{equation}\label{}
    \max_\mathbf{v} |\mathbf{e}^H\mathbf{v}|^2
  \end{equation}
subject to $\|\mathbf{H}_{A,E}\mathbf{v}\|=1$. The Lagrangian function of the problem is
\begin{equation}\label{}
  L=-|\mathbf{e}^H\mathbf{v}|^2+\eta (\|\mathbf{H}_{A,E}\mathbf{v}\|^2
  -1).
\end{equation}
It follows that
\begin{equation}\label{}
  \frac{1}{2}\frac{\partial L}{\partial \mathbf{v}}
  =-\mathbf{e}\mathbf{e}^H\mathbf{v}+\eta\mathbf{H}_{A,E}^H\mathbf{H}_{A,E}\mathbf{v}.
\end{equation}
So, by solving $\frac{\partial L}{\partial \mathbf{v}}=0$, the optimal $\mathbf{v}$ is a right-side eigenvector of $\mathbf{G}=
(\mathbf{H}_{A,E}^H\mathbf{H}_{A,E})^{-1}\mathbf{e}\mathbf{e}^H$. Among the $n_E$ eigenvectors of $\mathbf{G}$, the optimal $\mathbf{v}$ is the one which gives the largest $|\mathbf{e}^H\mathbf{v}|^2$. The final optimal $\mathbf{v}$ should be normalized to have the unit norm.
  \end{IEEEproof}

\subsection{MMSE estimation of $r_B(k)$ by Eve}

In order to determine the secrecy rate of the conventional SIMO scheme, we must assume that Eve performs the MMSE estimate $\hat r_{B|E}(k)$ of $r_B(k)$  from $\mathbf{y}_{E,1}(k)$ in \eqref{eq:yE1k}. The MSE of this estimate is
\begin{align}\label{eq:sDrBE}
  &\sigma_{\Delta r_{B|E}}^2 = 1-p_B\mathbf{h}_{E,B}^H\left (
  p_B\mathbf{h}_{E,B}\mathbf{h}_{E,B}^H+\mathbf{I}\right )^{-1}\mathbf{h}_{E,B}\notag\\
  &=\frac{1}{p_B\|\mathbf{h}_{E,B}\|^2+1}.
\end{align}

\subsection{Secrecy rate of the conventional SIMO scheme}
It follows that
the secrecy rate of the conventional SIMO scheme from Bob to Alice in bits per channel use is $R_{s,\texttt{SIMO,conv}}^+$ with
\begin{align}\label{eq:R_s_conv_BA}
  &R_{s,\texttt{SIMO,conv}}=\mathbb{I}(r_B(k);\mathbf{y}_{A,1}(k))-
  \mathbb{I}(r_B(k);\mathbf{y}_{E,1}(k))
  \notag\\
  &=\log\frac{1}{\sigma_{\Delta r_{B|A}}^2}-\log\frac{1}{\sigma_{\Delta r_{B|E}}^2}\notag\\
  &=\log
  \frac{1+p_B\|\mathbf{h}_{A,B}\|^2}{(1+p_B\|\mathbf{h}_{E,B}\|^2)(1+
  t_1\alpha_B)},
\end{align}
where we have applied  Lemma 1 in \cite{Hua_STEEP_2025}, \eqref{eq:sDrBA3} and \eqref{eq:sDrBE}.

We see that $R_{s,\texttt{SIMO,conv}}>0$  for some $p_B$ if and only if $\lim_{p_B\to\infty}R_{s,\texttt{SIMO,conv}}>0$, or equivalently if and only if
\begin{equation}\label{eq:space_SIMO}
\frac{\|\mathbf{h}_{A,B}\|^2}{\|\mathbf{h}_{E,B}\|^2(1+
  t_1''')}>1,
\end{equation}
where
\begin{align}\label{}
    &t_1'''\doteq\lim_{p_B\to\infty}t_1\notag\\
    &=\frac{p_{E,1}|\mathbf{\bar h}_{A,B}^H\mathbf{H}_{A,E}\mathbf{v}_{A,E}|^2}
   {1+p_{E,1}\left (\|\mathbf{H}_{A,E}\mathbf{v}_{A,E}\|^2-|\mathbf{\bar h}_{A,B}^H\mathbf{H}_{A,E}\mathbf{v}_{A,E}|^2
  \right )
  }.
  \end{align}
  We also see that the jamming power $p_{E,1}$ from Eve also causes a significant burden (i.e., $t_1'''$) for the conventional SIMO scheme to yield a positive secrecy rate.
  Equivalently, we have the following:
  \begin{Proposition}
    $R^+_{s,\texttt{SIMO,conv}}=0$ at any $p_B$ if and only if the negation of \eqref{eq:space_SIMO} holds.
  \end{Proposition}
The channel space defined by the negation (or converse) of \eqref{eq:space_SIMO} is clearly significant.

\section{SIMO-MISO-STEEP}\label{sec:SIMO-MISO-STEEP}
In this section, we present the SIMO-MISO-STEEP scheme. We will first make clear a connection between phase 1 of this scheme and the conventional SIMO scheme, and then provide the details of phase 2 of this scheme. We will also analyze the secrecy rate $R_{\texttt{s,S-M-STEEP}}$ of SIMO-MISO-STEEP and show that $R_{\texttt{s,S-M-STEEP}}$ is positive when the phase-2 power of SIMO-MISO-STEEP is large enough, despite Eve's finite jamming power and Eve's eavesdropping channel strength.

\subsection{Phase 1 of SIMO-MISO-STEEP}

In phase 1, Bob transmits a sequence of random probes to Alice and Eve, which is similar to the conventional SIMO scheme except that the transmitted random signal $r_B(k)$ from Bob is now \emph{uncoded} and consisting of completely independent symbols. Neither Alice nor Eve will in general be able to determine exactly the  random probing symbols  due to channel noises and the uncoded symbols.

\subsection{Phase 2 of SIMO-MISO-STEEP}

In phase 2, Alice transmits the following vector sequence:
\begin{align}\label{eq:xA2}
&\mathbf{x}_{A,2}(k)\doteq
  \sqrt{p_A\gamma_A}\mathbf{v}_A(s_A(k)+\hat r_{B|A}(k))\notag\\
  &\,\,+
\sqrt{\frac{p_A(1-\gamma_A)}{n_A-1}}\mathbf{V}_A\mathbf{n}_A(k).
\end{align}
Here $s_A(k)$ is a sequence of secret message symbols encoded with a message $\mathcal{M}_A$ to be reliably detected by Bob, $\hat r_{B|A}(k)$ in \eqref{eq:hrBAk} is used here to ``hide or encrypt'' $s_A(k)$, and $\mathbf{n}_A(k)$ is an artificial noise from Alice to jam Eve. Furthermore, to maximize the SNR at Bob, $\mathbf{v}_A$ is chosen to be $\frac{\mathbf{h}_{B,A}^*}{\|\mathbf{h}_{B,A}\|}
=\mathbf{\bar h}_{B,A}^*$, $\mathbf{V}_A$ is chosen to be orthogonal complement to $\mathbf{v}_A$, and hence $[\mathbf{v}_A,\mathbf{V}_A]$ is  an $n_A\times n_A$ unitary matrix. So, the power of $\mathbf{x}_{A,2}(k)$ is $p_A\gamma_A(1+\sigma_{\hat r_{B|A}}^2)+p_A(1-\gamma_A)=p_A+p_A\gamma_A\sigma_{\hat r_{B|A}}^2<(1+\gamma_A)p_A<2p_A$, where $0<\gamma_A<1$.

Note that if frequency division duplex (FDD) is used between Alice and Bob, $\mathbf{x}_{A,2}(k)$ can be transmitted as soon as after $\hat r_{B|A}(k)$ is obtained from $\mathbf{y}_{A,1}(k)$. In other words, the transmission in phase 2 could start as soon as the first sample in phase 1 has been received and processed.

The corresponding signal received by Bob is
\begin{align}\label{eq:yB2k}
  &y_{B,2}(k) = \sqrt{p_A\gamma_A}\|\mathbf{h}_{B,A}\|(s_A(k)+\hat r_{B|A}(k))
  +w_{B,2}(k) \notag\\
  &\,\,+\sqrt{p_{E,2}}\|\mathbf{h}_{B,E}\|n_{E,2}(k),
\end{align}
where we have assumed that Eve sends the optimized jamming noise $\sqrt{p_{E,2}}\mathbf{\bar h}_{B,E}^*n_{E,2}(k)$ against Bob. And the corresponding signal received by Eve is
\begin{align}\label{}
  &\mathbf{y}_{E,2}(k) = \sqrt{p_A\gamma_A}\mathbf{g}_{E,A}(s_A(k)+\hat r_{B|A}(k))\notag\\
  &\,\,+\sqrt{\frac{p_A(1-\gamma_A)}{n_A-1}}\mathbf{H}_{E,A}'\mathbf{n}_A(k)+\mathbf{w}_{E,2}(k),
\end{align}
with $\mathbf{g}_{E,A}=\mathbf{H}_{E,A}\mathbf{v}_A$ and $\mathbf{H}_{E,A}'=\mathbf{H}_{E,A}\mathbf{V}_A$.

To detect the message $\mathcal{M}_A$ from Alice, Bob can use the sequences $y_{B,2}(k)$ and $r_B(k)$, but Eve can only use $\mathbf{y}_{E,1}(k)$ and $\mathbf{y}_{E,2}(k)$. The effective capacity from Alice to Bob is now
\begin{equation}\label{eq:CAB}
  C_{A|B}=\mathbb{I}(s_A(k);y_{B,2}(k),r_B(k))=-\log\sigma_{\Delta s_{A|B}}^2
\end{equation}
with $\sigma_{\Delta s_{A|B}}^2$ being the MSE of the MMSE estimate of $s_A(k)$ from $y_{B,2}(k)$ and $r_B(k)$. And the effective capacity from Alice to Eve is now
\begin{equation}\label{eq:CAE}
  C_{A|E}=\mathbb{I}(s_A(k);\mathbf{y}_{E,1}(k),\mathbf{y}_{E,2}(k))=-\log\sigma_{\Delta s_{A|E}}^2
\end{equation}
with $\sigma_{\Delta s_{A|E}}^2$ being the MSE of the MMSE estimate of $s_A(k)$ from $\mathbf{y}_{E,1}(k)$ and $\mathbf{y}_{E,2}(k)$.

\subsection{MMSE estimation of $s_A(k)$ by Bob}

The MMSE estimate of $s_A(k)$ from $r_B(k)$ and $y_{B,2}(k)$ can be shown to be
\begin{equation}\label{}
  \hat s_{A|B}(k) =\mathbb{E}\{s_A(k)\Delta y_{B,2}^*(k)\}(\mathbb{E}\{|\Delta y_{B,2}(k)|^2\})^{-1}\Delta y_{B,2}(k)
\end{equation}
where
\begin{align}\label{eq:DyB2k}
  &\Delta y_{B,2}(k) \doteq y_{B,2}(k)-\mathbb{E}\{y_{B,2}(k)|r_B(k)\}\notag\\
  &=y_{B,2}(k) - \sqrt{p_A\gamma_A}\|\mathbf{h}_{B,A}\|\mathbb{E}\{\hat r_{B|A}(k)|r_B(k)\}.
\end{align}
Since $\mathbb{E}\{\hat r_{B|A}(k)|r_B(k)\}=\sigma_{\hat r_{B|A}}^2r_B(k)$ with $\sigma_{\hat r_{B|A}}^2
=1-\sigma_{\Delta r_{B|A}}^2$, then
\begin{align}\label{}
  &\Delta y_{B,2}(k) = \sqrt{p_A\gamma_A}\|\mathbf{h}_{B,A}\|(s_A(k)+\hat r_{B|A}(k)-\sigma_{\hat r_{B|A}}^2r_B(k))\notag\\
  &\,\,
  +w_{B,2}(k) +\sqrt{p_{E,2}}\|\mathbf{h}_{B,E}\|n_{E,2}(k).
\end{align}
It follows that the MSE of $\hat s_{A|B}(k)$ is
\begin{align}\label{eq:sDsAB}
  &\sigma_{\Delta s_{A|B}}^2
  =1-\mathbb{E}\{s_A(k)\Delta y_{B,2}^*(k)\}(\mathbb{E}\{|\Delta y_{B,2}(k)|^2\})^{-1}\notag\\
  &\,\,\cdot
  (\mathbb{E}\{s_A(k)\Delta y_{B,2}^*(k)\})^*\notag\\
  &=\frac{p_A\gamma_A\|\mathbf{h}_{B,A}\|^2
  \sigma_{\Delta r_{B|A}}^2\sigma_{\hat r_{B|A}}^2+1+p_{E,2}\|\mathbf{h}_{B,E}\|^2}
  {p_A\gamma_A\|\mathbf{h}_{B,A}\|^2
  (1+\sigma_{\Delta r_{B|A}}^2\sigma_{\hat r_{B|A}}^2)+1+p_{E,2}\|\mathbf{h}_{B,E}\|^2}.
\end{align}

Note that with $p_B>0$, we have $0<\sigma_{\Delta r_{B|A}}^2<1$, $0<\sigma_{\hat r_{B|A}}^2<1$ and $0<\sigma_{\Delta r_{B|A}}^2\sigma_{\hat r_{B|A}}^2<1$. It follows from \eqref{eq:sDsAB} that for $p_B>0$ (and hence $\sigma_{\Delta r_{B|A}}^2\sigma_{\hat r_{B|A}}^2$ is strictly positive),
\begin{equation}\label{eq:lim_AB}
  \lim_{p_A\to\infty}\sigma_{\Delta s_{A|B}}^2
  =\frac{
  \sigma_{\Delta r_{B|A}}^2\sigma_{\hat r_{B|A}}^2}
  {1+\sigma_{\Delta r_{B|A}}^2\sigma_{\hat r_{B|A}}^2}.
\end{equation}


\subsection{MMSE estimation of $s_A(k)$ at Eve}

Recall that the signals received by Eve in phases 1 and 2 of SIMO-MISO-STEEP are
\begin{equation}\label{}
  \left \{\begin{array}{c}
            \mathbf{y}_{E,1}(k)=\sqrt{p_B}\mathbf{h}_{E,B}r_B(k) +\mathbf{w}_{E,1}(k), \\
            \mathbf{y}_{E,2}(k) = \sqrt{p_A\gamma_A}\mathbf{g}_{E,A}(s_A(k)+\hat r_{B|A}(k))\\
  \,\,+\sqrt{\frac{p_A(1-\gamma_A)}{n_A-1}}\mathbf{H}_{E,A}'\mathbf{n}_A(k)+\mathbf{w}_{E,2}(k) .
          \end{array}
  \right .
\end{equation}
It follows that the MSE of the MMSE estimate of $s_A(k)$ by Eve from $\mathbf{y}_{E,1}(k)$ and $\mathbf{y}_{E,2}(k)$ is
\begin{equation}\label{}
  \sigma_{\Delta s_{A|E}}^2
  =1-\mathbf{r}^H\mathbf{R}^{-1}\mathbf{r}
\end{equation}
where $\mathbf{r}^H = \mathbb{E}\{s_A(k)[\mathbf{y}_{E,1}^H(k),\mathbf{y}_{E,2}^H(k)]\}
  =[0,\sqrt{p_A\gamma_A}\mathbf{g}_{E,A}^H]$ and
\begin{equation}\label{}
  \mathbf{R}=\mathbb{E}\left \{\left [\begin{array}{c}
                                        \mathbf{y}_{E,1} \\
                                        \mathbf{y}_{E,2}
                                      \end{array}
  \right ]
  \left [\begin{array}{c}
                                        \mathbf{y}_{E,1} \\
                                        \mathbf{y}_{E,2}
                                      \end{array}
  \right ]^H\right \}=\left [\begin{array}{cc}
                      \mathbf{R}_{1,1} & \mathbf{R}_{1,2} \\
                      \mathbf{R}_{1,2}^H & \mathbf{R}_{2,2}
                    \end{array}
  \right ],
\end{equation}
with
$
  \mathbf{R}_{1,1}
  =p_B\mathbf{h}_{E,B}\mathbf{h}_{E,B}^H +\mathbf{I}
$, $
  \mathbf{R}_{1,2}
  =\sqrt{p_A\gamma_A p_B}\sigma_{\hat r_{B|A}}^2\mathbf{h}_{E,B}\mathbf{g}_{E,A}^H$, and
\begin{align}\label{}
&\mathbf{R}_{2,2}
=p_A\gamma_A(1+\sigma_{\hat r_{B|A}}^2)\mathbf{g}_{E,A}\mathbf{g}_{E,A}^H\notag\\
&\,\,+\frac{p_A(1-\gamma_A)}{n_A-1}
\mathbf{H}_{E,A}'\mathbf{H}_{E,A}'^H+\mathbf{I}.
\end{align}
Furthermore, one can verify that
 \begin{equation}\label{eq:sDsAE}
  \sigma_{\Delta s_{A|E}}^2
  =1-p_A\gamma_A\mathbf{g}_{E,A}^H(\mathbf{R}_{2,2}-\mathbf{R}_{1,2}^H
                                \mathbf{R}_{1,1}^{-1}\mathbf{R}_{1,2})^{-1}
                                \mathbf{g}_{E,A}.
\end{equation}

\subsection{Properties of $\sigma_{\Delta s_{A|E}}^2$ in \eqref{eq:sDsAE}}
To show an important property of the secrecy rate of the SIMO-MISO-STEEP, we need to reveal some properties of $\sigma_{\Delta s_{A|E}}^2$. We start with the following:
\begin{align}\label{}
  &\mathbf{R}_{1,2}^H\mathbf{R}_{1,1}^{-1}\mathbf{R}_{1,2}\notag\\
  &=p_A\gamma_A p_B\sigma_{\hat r_B}^4\mathbf{g}_{E,A}\mathbf{h}_{E,B}^H
  (p_B\mathbf{h}_{E,B}\mathbf{h}_{E,B}^H +\mathbf{I})^{-1}\mathbf{h}_{E,B}\mathbf{g}_{E,A}^H
  \notag\\
  &=t_2p_A\gamma_A\mathbf{g}_{E,A}\mathbf{g}_{E,A}^H
\end{align}
with $t_2=\frac{ p_B\sigma_{\hat r_B}^4\|\mathbf{h}_{E,B}\|^2}{p_B\|\mathbf{h}_{E,B}\|^2+1}$. Hence
\begin{align}\label{eq:sDsAE_b}
  &\sigma_{\Delta s_{A|E}}^2
  =1-p_A\gamma_A\mathbf{g}_{E,A}^H\left (t_2'p_A\gamma_A\mathbf{g}_{E,A}\mathbf{g}_{E,A}^H+\mathbf{I}\right .\notag\\
  &\,\,\left .+\frac{p_A(1-\gamma_A)}{n_A-1}
\mathbf{H}_{E,A}'\mathbf{H}_{E,A}'^H\right )^{-1}
                                \mathbf{g}_{E,A}
\end{align}
with $t_2'=(1+\sigma_{\hat r_{B|A}}^2)-t_2
=1+\sigma_{\hat r_{B|A}}^2-\frac{ \sigma_{\hat r_{B|A}}^4p_B\|\mathbf{h}_{E,B}\|^2}{p_B\|\mathbf{h}_{E,B}\|^2+1}>1$.

Using the matrix inverse lemma, we have
\begin{align}\label{}
  &\sigma_{\Delta s_{A|E}}^2
  =1-p_A\gamma_A\mathbf{g}_{E,A}^H\left (\mathbf{A}_2\right .\notag\\
  &\,\,\left .+\frac{p_A(1-\gamma_A)}{n_A-1}
\mathbf{H}_{E,A}'\mathbf{H}_{E,A}'^H\right )^{-1}
                                \mathbf{g}_{E,A}\notag\\
&=1-p_A\gamma_A\mathbf{g}_{E,A}^H\left (
\mathbf{A}_2^{-1}\right .\notag\\
&\,\,\left .-\frac{p_A(1-\gamma_A)}{n_A-1}\mathbf{A}_2^{-1}
\mathbf{H}_{E,A}'\mathbf{B}_2^{-1}\mathbf{H}_{E,A}'^H\mathbf{A}_2^{-1}\right )\mathbf{g}_{E,A}
\end{align}
with $\mathbf{A}_2=t_2'p_A\gamma_A\mathbf{g}_{E,A}\mathbf{g}_{E,A}^H+\mathbf{I}_{n_E}$ and $\mathbf{B}_2
=\mathbf{I}_{n_A-1}+\frac{p_A(1-\gamma_A)}{n_A-1}
\mathbf{H}_{E,A}'^H\mathbf{A}_2^{-1}\mathbf{H}_{E,A}'$. It follows that
\begin{align}\label{eq:35}
  &\sigma_{\Delta s_{A|E}}^2
=1-p_A\gamma_A\left (\frac{\|\mathbf{g}_{E,A}\|^2}{t_2'p_A\gamma_A\|\mathbf{g}_{E,A}\|^2+1}
\right .\notag\\
&\,\,\left .-\frac{p_A(1-\gamma_A)}{n_A-1}\frac{\|\mathbf{g}_{E,A}\|^2}
{(t_2'p_A\gamma_A\|\mathbf{g}_{E,A}\|^2+1)^2}
t_2''\right )\notag\\
&=\frac{(t_2'-1)p_A\gamma_A\|\mathbf{g}_{E,A}\|^2+1}{t_2'p_A\gamma_A\|\mathbf{g}_{E,A}\|^2+1}\notag\\
&\,\,
+\frac{p_A^2\gamma_A(1-\gamma_A)}{n_A-1}\frac{\|\mathbf{g}_{E,A}\|^2}
{(t_2'p_A\gamma_A\|\mathbf{g}_{E,A}\|^2+1)^2}
t_2''
\end{align}
with $t_2''=\mathbf{\bar g}_{E,A}^H\mathbf{H}_{E,A}'\mathbf{B}_2^{-1}\mathbf{H}_{E,A}'^H\mathbf{\bar g}_{E,A}
=\mathbf{c}_{E,A}^H\mathbf{B}_2^{-1}\mathbf{c}_{E,A}$ and $\mathbf{c}_{E,A}=\mathbf{H}_{E,A}'^H\mathbf{\bar g}_{E,A}$. The above term associated with $t_2''$ is due to the artificial noise from Alice.

Since $\mathbf{A}_2^{-1}=\mathbf{I}_{n_E}-\frac{t_2'p_A\gamma_A}{1+t_2'p_A\gamma_A\|\mathbf{g}_{E,A}\|^2}
\mathbf{g}_{E,A}\mathbf{g}_{E,A}^H$, we have
\begin{align}\label{eq:B2}
  &\mathbf{B}_2
=\mathbf{I}_{n_A-1}+\frac{p_A(1-\gamma_A)}{n_A-1}
\mathbf{H}_{E,A}'^H\left (\mathbf{I}_{n_E}\right .\notag\\
&\,\,\left .-\frac{t_2'p_A\gamma_A}{1+t_2'p_A\gamma_A\|\mathbf{g}_{E,A}\|^2}
\mathbf{g}_{E,A}\mathbf{g}_{E,A}^H\right )\mathbf{H}_{E,A}'\notag\\
&\approx \mathbf{I}_{n_A-1}+\frac{p_A(1-\gamma_A)}{n_A-1}\mathbf{T}_1
\end{align}
where the approximation is due to $\frac{t_2'p_A\gamma_A}{1+t_2'p_A\gamma_A\|\mathbf{g}_{E,A}\|^2}
\mathbf{g}_{E,A}\mathbf{g}_{E,A}^H\approx \mathbf{\bar g}_{E,A}\mathbf{\bar g}_{E,A}^H$ subject to a large $p_A$, and $\mathbf{T}_1
=\mathbf{H}_{E,A}'^H(\mathbf{I}_{n_E}-\mathbf{\bar g}_{E,A}\mathbf{\bar g}_{E,A}^H)\mathbf{H}_{E,A}'\in\mathcal{C}^{(n_A-1)\times (n_A-1)}$. Also $rank(\mathbf{T}_1)=\min(n_E-1,n_A-1)$ since $\mathbf{I}_{n_E}-\mathbf{\bar g}_{E,A}\mathbf{\bar g}_{E,A}^H$ and $\mathbf{H}_{E,A}'$ have their ranks equal to $n_E-1$ and $n_A-1$ respectively.

\subsubsection{For $n_E\geq n_A$}
In this case, $\mathbf{T}_1$ has the full rank $n_A-1$, or more specifically $\mathbf{T}_1>0$ (positive definite) due to the following lemma.

\begin{Lemma}
   Let $n>m$, $\mathbf{H}\in\mathcal{C}^{n\times m}$, $\mathbf{c}\in\mathcal{C}^{n\times 1}$ and $\|\mathbf{c}\|=1$. Then $\mathbf{D}\doteq\mathbf{H}^H\mathbf{H}-\mathbf{H}^H\mathbf{c}\mathbf{c}^H\mathbf{H}>0$ (positive definite) subject to $\mathbf{c}\notin \texttt{range}(\mathbf{H})$ and $\texttt{rank}(\mathbf{H})=m$.
 \end{Lemma}
 \begin{IEEEproof}
 Let the SVD of $\mathbf{H}$ be $\mathbf{H}=\mathbf{U}\boldsymbol{\Sigma}\mathbf{V}^H$ where $\mathbf{U}\in\mathcal{C}^{n\times m}$ and $\mathbf{V}\in\mathbf{C}^{m\times m}$ each has $m$ orthonormal columns, and $\boldsymbol{\Sigma}$ is a $m\times m$ positive diagonal matrix since $\texttt{rank}(\mathbf{H})=m$. Then
 $\mathbf{D}=\mathbf{V}\boldsymbol{\Sigma}\mathbf{Q}\boldsymbol{\Sigma}\mathbf{V}^H$ with
$
   \mathbf{Q}=\mathbf{I}_m-\mathbf{c}'\mathbf{c}'^H
 $
 and $\mathbf{c}'=\mathbf{U}^H\mathbf{c}\in\mathcal{C}^{m\times 1}$. Since $\mathbf{c}\notin \texttt{range}(\mathbf{H})$, $\|\mathbf{c}'\|<\|\mathbf{c}\|=1$ and hence all eigenvalues of $\mathbf{Q}$ are positive (i.e., one of the eigenvalues of $\mathbf{Q}$ is $1-\|\mathbf{c}'\|^2$, and all others of them are ones). Therefore, $\mathbf{Q}>0$ and hence $\mathbf{D}>0$.
 \end{IEEEproof}

 Hence $\lim_{p_A\to\infty}t_2''=\lim_{p_A\to\infty}\mathbf{c}_{E,A}^H\mathbf{B}_2^{-1}\mathbf{c}_{E,A}= 0$. Therefore, it follows from \eqref{eq:35} that for $n_E\geq n_A$,
  \begin{align}\label{eq:lim_AE}
  &\lim_{p_A\to\infty}\sigma_{\Delta s_{A|E}}^2
=\frac{t_2'-1}{t_2'}
=\frac{\sigma_{\hat r_{B|A}}^2-\frac{ \sigma_{\hat r_{B|A}}^4p_B\|\mathbf{h}_{E,B}\|^2}{p_B\|\mathbf{h}_{E,B}\|^2+1}}{1+\sigma_{\hat r_{B|A}}^2-\frac{ \sigma_{\hat r_{B|A}}^4p_B\|\mathbf{h}_{E,B}\|^2}{p_B\|\mathbf{h}_{E,B}\|^2+1}}\notag\\
&>\frac{\sigma_{\hat r_{B|A}}^2-\sigma_{\hat r_{B|A}}^4}
{1+\sigma_{\hat r_{B|A}}^2-\sigma_{\hat r_{B|A}}^4}
=\lim_{p_A\to\infty}\sigma_{\Delta s_{A|B}}^2,
\end{align}
the last expression of which is \eqref{eq:lim_AB} subject to $p_B>0$.

\subsubsection{For $n_E<n_A$}
In this case, $\mathbf{T}_1$ has the rank $n_E-1<n_A-1$. Let $\mathbf{Q}_1\in \mathcal{C}^{(n_A-1)\times (n_E-1)}$ be the matrix consisting of the $n_E-1$ eigenvectors corresponding to the $n_E-1$ nonzero eigenvalues of $\mathbf{T}_1$, and $\mathbf{Q}_2$ be the matrix consisting of the other $n_A-n_E$ eigenvectors of $\mathbf{T}_1$. Then
$\lim_{p_A\to\infty}t_2''=\lim_{p_A\to\infty}\mathbf{c}_{E,A}^H\mathbf{B}_2^{-1}\mathbf{c}_{E,A}= \|\mathbf{Q}_2\mathbf{c}_{E,A}\|^2$. It follows that
\begin{align}\label{eq:38}
  &\lim_{p_A\to\infty}\sigma_{\Delta s_{A|E}}^2
=\frac{t_2'-1}{t_2'}
+\frac{(1-\gamma_A)\|\mathbf{Q}_2\mathbf{c}_{E,A}\|^2}
{(n_A-1)\gamma_At_2'^2\|\mathbf{g}_{E,A}\|^2}.
\end{align}


\subsection{Secrecy rate of SIMO-MISO-STEEP}

With the effective capacities $C_{A|B}$ and $C_{A|E}$ shown earlier, an achievable secrecy rate of the SIMO-MISO-STEEP scheme from Alice to Bob  in bits per round-trip channel use is $R_{s,\texttt{S-M-STEEP}}^+$ with
\begin{align}\label{eq:RsAB}
  &R_{s,\texttt{S-M-STEEP}} =C_{A|B}-C_{A|E}=\log\frac{\sigma_{\Delta s_{A|E}}^2}{\sigma_{\Delta s_{A|B}}^2}.
\end{align}

Then, according to \eqref{eq:lim_AE}, and that $R_{s,\texttt{S-M-STEEP}}$ increases as $n_E$ decreases, we have:
\begin{Proposition}
  For all $n_E\geq 1$ and $p_B>0$,
\begin{equation}\label{eq:limit_SIMO-MISO-STEEP}
  \lim_{p_A\to\infty}R_{s,\texttt{S-M-STEEP}}>0.
\end{equation}
\end{Proposition}
Since $R_{s,\texttt{S-M-STEEP}}$ is a continuous function of $p_A$, there is a finite $\bar p_A$ such that $R_{s,\texttt{S-M-STEEP}}>0$ when $p_A>\bar p_A$. The expression of $\bar p_A$ is governed by the equation $\sigma_{\Delta s_{A|E}}^2=\sigma_{\Delta s_{A|B}}^2$.

Also note that the conventional SIMO scheme transmits a secret message (embedded in $r_B(k)$) from Bob to Alice while the SIMO-MISO-STEEP scheme transmits a secret message (embedded in $s_A(k)$) from Alice to Bob.

  \section{Conventional MISO Scheme}\label{sec:MISO_conv}

In this section, we present an optimal MISO scheme against optimized full-duplex multi-antenna eavesdropper. We will also present its secrecy rate $R_{\texttt{s,MISO,conv}}$ and its properties. Among these properties is that $R_{\texttt{s,MISO,conv}}^+=0$ even with infinite transmit power from user if Eve's channel is strong.

In a conventional MISO scheme, Alice sends out:
\begin{equation}\label{eq:xA1}
  \mathbf{x}_{A,1}(k)\doteq
  \sqrt{p_A\gamma_A}\mathbf{v}_Ar_A(k)+
\sqrt{\frac{p_A(1-\gamma_A)}{n_A-1}}\mathbf{V}_A\mathbf{n}_A(k)
\end{equation}
  where $[\mathbf{v}_A,\mathbf{V}_A]$ is an $n_A\times n_A$ unitary matrix, $\mathbf{v}_A=\mathbf{\bar h}_{B,A}^*$, and $\mathbf{n}_A(k)$ is an artificial noise.
Here $\gamma_A$ is exactly the power splitting factor between signal and artificial noise.

Then Bob receives
\begin{align}\label{eq:yB1k}
  &y_{B,1}(k) =\sqrt{p_A\gamma_A}\|\mathbf{h}_{B,A}\| r_A(k)+w_{B,1}(k)\notag\\
  &\,\,
  +\sqrt{p_{E,1}}\|\mathbf{h}_{B,E}\|n_{E,1}(k)
\end{align}
and Eve receives
\begin{align}\label{eq:yE1k_b}
  &\mathbf{y}_{E,1}(k) =\sqrt{p_A\gamma_A}\mathbf{g}_{E,A}r_A(k)+
  \sqrt{\frac{p_A(1-\gamma_A)}{n_A-1}}\mathbf{H}_{E,A}'\mathbf{n}_A(k)\notag\\
  &\,\,
  +\mathbf{w}_{E,1}(k)
\end{align}
with $\mathbf{g}_{E,A}=\mathbf{H}_{E,A}\mathbf{v}_A$, and $\mathbf{H}_{E,A}'=\mathbf{H}_{E,A}\mathbf{V}_A$.

Note that  $\mathbf{y}_{E,1}(k)$ and $\mathbf{y}_{E,2}(k)$ here in sections \ref{sec:MISO_conv} and \ref{sec:MISO-SIMO-STEEP} are different from the previous $\mathbf{y}_{E,1}(k)$ and $\mathbf{y}_{E,2}(k)$ in sections \ref{sec:SIMO_conv} and \ref{sec:SIMO-MISO-STEEP}.

We have assumed that Eve applies the optimized jamming noise $\sqrt{p_{E,1}}\mathbf{\bar h}_{B,E}^*n_{E,1}(k)$ against Bob. As before, the self-interference at Eve is zero (which in practice implies an upper limit on the jamming power from Eve).

Here, $r_A(k)\sim \mathcal{CN}(0,1)$, $\mathbf{n}_A(k)\sim\mathcal{CN}(0,\mathbf{I})$,  $w_{B,1}(k)\sim \mathcal{CN}(0,1)$, $n_{E,1}(k)\sim \mathcal{CN}(0,1)$, and $\mathbf{w}_{E,1}(k)\sim\mathcal{CN}(0,\mathbf{I})$.

\subsection{MMSE estimate of $r_A(k)$ by Bob}
The MMSE estimate $\hat r_{A|B}(k)$ of $r_A(k)$ by Bob from $y_{B,1}(k)$ in \eqref{eq:yB1k} is
\begin{align}\label{eq:hrABk}
  &\hat r_{A|B}(k) = \mathbb{E}\{r_A(k)y_{B,1}^*(k)\}(\mathbb{E}\{|y_{B,1}(k)|^2\})^{-1}y_{B,1}(k)
  \notag\\
  &=\frac{\sqrt{p_A\gamma_A}\|\mathbf{h}_{B,A}\|}{p_A\gamma_A\|\mathbf{h}_{B,A}\|^2
  +1+p_{E,1}\|\mathbf{h}_{B,E}\|^2}y_{B,1}(k).
\end{align}
The MSE of $\hat r_{A|B}(k)$  is
\begin{align}\label{eq:sDrA}
  &\sigma_{\Delta r_{A|B}}^2 = 1-\mathbb{E}\{r_A(k)y_{B,1}^*(k)\}(\mathbb{E}\{|y_{B,1}(k)|^2\})^{-1}
  \notag\\
  &\,\,\cdot(\mathbb{E}\{r_A(k)y_{B,1}^*(k)\})^*
  \notag\\
  &=
  \frac{1+p_{E,1}\|\mathbf{h}_{B,E}\|^2}
  {p_A\gamma_A\|\mathbf{h}_{B,A}\|^2
  +1+p_{E,1}\|\mathbf{h}_{B,E}\|^2}.
\end{align}

We will need the property that $0<\sigma_{\Delta r_{A|B}}^2<1$ if $p_A>0$.

\subsection{MMSE estimate of $r_A(k)$ by Eve}

Recall that $\mathbf{h}_{B,A}\in\mathcal{C}^{n_A\times 1}$ and $\mathbf{H}_{E,A}\in\mathcal{C}^{n_E\times n_A}$ are known to Eve.

\subsubsection{Elimination of artificial noise by Eve}
If $n_E\geq n_A$, there is a vector $\mathbf{t}_1\in\mathcal{C}^{n_E\times 1}$ such that
\begin{equation}\label{}
  \mathbf{t}_1^T\mathbf{H}_{E,A}=\mathbf{\bar h}_{B,A}^T.
\end{equation}
The minimum-norm solution of  $\mathbf{t}_1$ is $\mathbf{t}_1 =(\mathbf{H}_{E,A}^T)^\dag\mathbf{\bar h}_{B,A}$ where $\mathbf{H}^\dag$ is the pseudo-inverse of $\mathbf{H}$. If the SVD of $\mathbf{H}$ is $\mathbf{H}=\sum_{i=1}^{n_A}\sigma_i\mathbf{u}_i\mathbf{v}_i^H$, then $\mathbf{H}^\dag=\sum_{i=1}^{n_A}\frac{1}{\sigma_i}\mathbf{v}_i\mathbf{u}_i^H$ or $(\mathbf{H}^T)^\dag=\sum_{i=1}^{n_A}\frac{1}{\sigma_i}\mathbf{u}_i^*\mathbf{v}_i^T$. Therefore, it follows from \eqref{eq:yE1k_b} that
\begin{align}\label{eq:yE1P}
  &y_{E,1}'(k)\doteq  \mathbf{t}_1^T\mathbf{y}_{E,1}(k)
  =\sqrt{p_A\gamma_A}r_A(k)
  +w_{E,1}'(k)
\end{align}
where we have used $\mathbf{t}_1^T\mathbf{g}_{E,A}=\mathbf{t}_1^T\mathbf{H}_{E,A}\mathbf{v}_A
=\mathbf{\bar h}_{B,A}^T\mathbf{v}_A=1$,
$\mathbf{t}_1^T\mathbf{H}_{E,A}'=\mathbf{t}_1^T\mathbf{H}_{E,A}\mathbf{V}_A=\mathbf{\bar h}_{B,A}^T\mathbf{V}_A=0$, and $w_{E,1}'(k)\doteq\mathbf{t}_1^T\mathbf{w}_{E,1}(k)$.
Furthermore, $w_{E,1}'(k)$ is $\mathcal{CN}(0,\sigma_w^2)$ with
\begin{align}\label{}
  &\sigma_w^2 = \|\mathbf{t}_1\|^2=\mathbf{\bar h}_{B,A}^H(\mathbf{H}_{E,A}^T\mathbf{H}_{E,A}^*)^\dag\mathbf{\bar h}_{B,A}\notag\\
  &=\mathbf{\bar h}_{B,A}^H(\mathbf{H}_{E,A}^T\mathbf{H}_{E,A}^*)^{-1}\mathbf{\bar h}_{B,A},
\end{align}
which is invariant to $p_A\gamma_A$. Here $(\mathbf{H}^T\mathbf{H}^*)^\dag=\sum_{i=1}^{n_A}\frac{1}{\sigma_i^2}\mathbf{v}_i^*\mathbf{v}_i^T
=(\mathbf{H}^T\mathbf{H}^*)^{-1}$ with $\mathbf{H}\in\mathcal{C}^{n_E\times n_A}$ and $n_E\geq n_A$.

We see that for $n_E\geq n_A$, the artificial noise $\mathbf{n}_A(k)$ does not exist in $y_{E,1}'(k)$ although the SNR $\eta'$ of $y_{E,1}'(k)$ is $\frac{p_A\gamma_A}{\mathbf{\bar h}_{B,A}^H(\mathbf{H}_{E,A}^T\mathbf{H}_{E,A}^*)^{-1}\mathbf{\bar h}_{B,A}}$ while the SNR $\eta$ of $y_{E,1}(k)$ in \eqref{eq:yE1k_b} without $\mathbf{n}_A(k)$ would be $p_A\gamma_A\|\mathbf{g}_{E,A}\|^2
=p_A\gamma_A\mathbf{\bar h}_{B,A}^H(\mathbf{H}_{E,A}^T\mathbf{H}_{E,A}^*)\mathbf{\bar h}_{B,A}$. Here
\begin{align}\label{}
  &\frac{\eta}{\eta'}=\left (\mathbf{\bar h}_{B,A}^H(\mathbf{H}_{E,A}^T\mathbf{H}_{E,A}^*)\mathbf{\bar h}_{B,A}\right )\left (\mathbf{\bar h}_{B,A}^H(\mathbf{H}_{E,A}^T\mathbf{H}_{E,A}^*)^{-1}\mathbf{\bar h}_{B,A}\right )\notag\\
  &\geq 1.
\end{align}
 The inequality of the above can be proven by using the fact that the expectation of a positive random number $X$ times the expectation of $X^{-1}$ is larger than or equal to one.

However, the MMSE estimate of $r_A(k)$ from $y_{E,1}'(k)$ in \eqref{eq:yE1P} is suboptimal for Eve since the construction of $y_{E,1}'(k)$ is ad hoc.

\subsubsection{Optimal estimation of $r_A(k)$ by Eve}
The MMSE estimate $\hat r_{A|E}(k)$ of $r_A(k)$ directly from $\mathbf{y}_{E,1}(k)$ in \eqref{eq:yE1k_b} is optimal for Eve. The MSE of $\hat r_{A|E}(k)$ is
\begin{align}\label{eq:sDrAE}
  &\sigma_{\Delta r_{A|E}}^2 =
  1-p_A\gamma_A\mathbf{g}_{E,A}^H\left (p_A\gamma_A\mathbf{g}_{E,A}\mathbf{g}_{E,A}^H\right .\notag\\
  &\,\,\left .+\frac{p_A(1-\gamma_A)}{n_A-1}\mathbf{H}_{E,A}'\mathbf{H}_{E,A}'^H
  +\mathbf{I}\right )^{-1}\mathbf{g}_{E,A}.
\end{align}

\subsection{Properties of $\sigma_{\Delta r_{A|E}}^2$ in \eqref{eq:sDrAE}}
Using the matrix inverse lemma, we can rewrite \eqref{eq:sDrAE} as
\begin{align}\label{}
  &\sigma_{\Delta r_{A|E}}^2 =1-p_A\gamma_A\mathbf{g}_{E,A}^H\left (\mathbf{A}_3^{-1}
  -\right .\notag\\
  &\,\,\left .\frac{p_A(1-\gamma_A)}{n_A-1}\mathbf{A}_3^{-1}\mathbf{H}_{E,A}'\mathbf{B}_3^{-1}
  \mathbf{H}_{E,A}'^H\mathbf{A}_3^{-1}\right )\mathbf{g}_{E,A}
\end{align}
with $\mathbf{A}_3=p_A\gamma_A\mathbf{g}_{E,A}\mathbf{g}_{E,A}^H+\mathbf{I}$ and $\mathbf{B}_3
=\mathbf{I}+\frac{p_A(1-\gamma_A)}{n_A-1}\mathbf{H}_{E,A}'^H\mathbf{A}_3^{-1}\mathbf{H}_{E,A}'$. Using $\mathbf{g}_{E,A}^H\mathbf{A}_3^{-1}=\frac{1}{p_A\gamma_A\|\mathbf{g}_{E,A}\|^2+1}\mathbf{g}_{E,A}^H$, we have
\begin{align}\label{eq:sDrAE_2}
  &\sigma_{\Delta r_{A|E}}^2
  =1-p_A\gamma_A\left (\frac{\|\mathbf{g}_{E,A}\|^2}{p_A\gamma_A\|\mathbf{g}_{E,A}\|^2+1}-\frac{p_A(1-\gamma_A)}{n_A-1}
  \right .\notag\\
  &\,\,\left .\cdot\frac{1}{(p_A\gamma_A\|\mathbf{g}_{E,A}\|^2+1)^2}\mathbf{g}_{E,A}^H
  \mathbf{H}_{E,A}'\mathbf{B}_3^{-1}
  \mathbf{H}_{E,A}'^H\mathbf{g}_{E,A}\right )\notag\\
  &=\frac{1}{p_A\gamma_A\|\mathbf{g}_{E,A}\|^2+1}+\notag\\
  &\,\,\frac{\gamma_A(1-\gamma_A)p_A^2}{(n_A-1)
  (p_A\gamma_A\|\mathbf{g}_{E,A}\|^2+1)^2}\mathbf{g}_{E,A}^H
  \mathbf{H}_{E,A}'\mathbf{B}_3^{-1}
  \mathbf{H}_{E,A}'^H\mathbf{g}_{E,A}.
\end{align}
Using $\mathbf{A}_3^{-1}=\mathbf{I}-\frac{p_A\gamma_A}{1+p_A\gamma_A\|\mathbf{g}_{E,A}\|^2}\mathbf{g}_{E,A}
\mathbf{g}_{E,A}^H$, we have
\begin{align}\label{}
  &\mathbf{B}_3
=\mathbf{I}+\frac{p_A(1-\gamma_A)}{n_A-1}\left (\mathbf{H}_{E,A}'^H\mathbf{H}_{E,A}'\right .\notag\\
&\,\,\left .-\frac{p_A\gamma_A}{1+p_A\gamma_A\|\mathbf{g}_{E,A}\|^2}\mathbf{H}_{E,A}'^H\mathbf{g}_{E,A}
\mathbf{g}_{E,A}^H\mathbf{H}_{E,A}'\right)\notag\\
&\approx \mathbf{I}+\frac{p_A(1-\gamma_A)}{n_A-1}\mathbf{T}_1,
\end{align}
 where $\mathbf{T}_1=\mathbf{H}_{E,A}'^H\mathbf{H}_{E,A}'
-\mathbf{H}_{E,A}'^H\mathbf{\bar g}_{E,A}
\mathbf{\bar g}_{E,A}^H\mathbf{H}_{E,A}'$, and the approximation is subject to large $p_A$. This above property of $\mathbf{B}_3$ is similar to $\mathbf{B}_2$ in \eqref{eq:B2}.
Furthermore, for $n_E\geq n_A$,
 $\mathbf{T}_1>0$ as discussed before.

Therefore, for large $p_A$ and $n_E\geq n_A$, \eqref{eq:sDrAE_2} becomes
\begin{align}\label{eq:sDrAE_3}
  &\sigma_{\Delta r_{A|E}}^2
  \approx
  \frac{1}{p_A
  \gamma_A\|\mathbf{g}_{E,A}\|^2}\left (1+t_3\right )
\end{align}
with
$t_3=\mathbf{\bar g}_{E,A}^H
  \mathbf{H}_{E,A}'\mathbf{T}_1^{-1}
  \mathbf{H}_{E,A}'^H\mathbf{\bar g}_{E,A}>0$
 which is invariant to $p_A$. We also see that $t_3$ is invariant to any scaling on $\mathbf{H}_{E,A}'$.

\subsection{Secrecy rate of the conventional MISO scheme}
Using \eqref{eq:sDrA} and \eqref{eq:sDrAE}, an achievable secrecy rate from Alice to Bob using the conventional MISO scheme in bits per channel use is $R_{s,\texttt{MISO,conv}}^+$ with
\begin{equation}\label{eq:R_s_conv_AB}
  R_{s,\texttt{MISO,conv}}=\log\frac{\sigma_{\Delta r_{A|E}}^2}{\sigma_{\Delta r_{A|B}}^2}.
\end{equation}

It follows from \eqref{eq:sDrA} and \eqref{eq:sDrAE_3} that for $n_E\geq n_A$,
\begin{align}\label{}
&\lim_{p_A\to\infty} R_{s,\texttt{MISO,conv}}=\log\frac{1}{p_A
  \gamma_A\|\mathbf{g}_{E,A}\|^2}\left (1+t_3\right )\notag\\
  &\,\,-\log \frac{1+p_{E,1}\|\mathbf{h}_{B,E}\|^2}
  {p_A\gamma_A\|\mathbf{h}_{B,A}\|^2
  }\notag\\
  &=\log\frac{\|\mathbf{h}_{B,A}\|^2(1+t_3)}{\|\mathbf{g}_{E,A}\|^2(1+p_{E,1}\|\mathbf{h}_{B,E}\|^2)}.
\end{align}
Therefore, we have the following:
\begin{Proposition}
  For $n_E\geq n_A$ and large $p_A$:
\begin{enumerate}
  \item $R_{s,\texttt{MISO,conv}}$ is invariant to $\gamma_A$ provided $0<\gamma_A<1$, i.e., invariant to the fraction of power used for artificial noise.
  \item $R_{s,\texttt{MISO,conv}}$ is increased by the use of artificial noise because of $t_3>0$. But  $t_3$ is invariant to any scaling on $\mathbf{H}_{E,A}$, or equivalently $t_3$ is invariant to the distance between Alice and Eve.
  \item $R_{s,\texttt{MISO,conv}}$ reduces as the scaling on $\mathbf{H}_{E,A}$ (equivalently on $\mathbf{g}_{E,A}$) increases and/or as $p_{E,1}$ increases.
  \item $R_{s,\texttt{MISO,conv}}^+=0$  if
      \begin{equation}\label{eq:space_MISO}
        \|\mathbf{g}_{E,A}\|^2(1+p_{E,1}\|\mathbf{h}_{B,E}\|^2)>
      \|\mathbf{h}_{B,A}\|^2(1+t_3).
      \end{equation}
\end{enumerate}
\end{Proposition}

It remains open to prove or disprove that \eqref{eq:space_MISO} is sufficient for $R_{s,\texttt{MISO,conv}}^+=0$ under $n_E\geq n_A$ and \emph{any} $p_A>0$. A challenge here is to prove or disprove that $R_{s,\texttt{MISO,conv}}$ is monotonically increasing function of $p_A$. This issue will not be addressed in this paper.


\section{MISO-SIMO-STEEP}\label{sec:MISO-SIMO-STEEP}
In this section, we present the MISO-SIMO-STEEP scheme, the phase 1 of which has a strong connection with the MISO scheme shown in the previous section. Much of this section focuses on phase 2 of MISO-SIMO-STEEP. The secrecy rate $R_{\texttt{s,M-S-STEEP}}$ of this scheme is analyzed and shown to be positive with a sufficiently large phase-2 power from user despite Eve's finite jamming power and finite eavesdropping channel quality.

\subsection{Phase 1 of MISO-SIMO-STEEP}
In phase 1, Alice broadcasts a sequence of scalar probes along with artificial noise, which is similar to the conventional MISO scheme except that the transmitted signal $r_A(k)$ from Alice is now \emph{uncoded}.

Note that if $\gamma_A=\frac{1}{n_A}$, then the overall signal transmitted by Alice here, i.e.,  $\mathbf{x}_{A,1}(k)$ in \eqref{eq:xA1}, is equivalent to the probing vector $\sqrt{\frac{p_A}{n_A}}\mathbf{x}_A(k)$ defined in  \cite{Hua_STEEP_2025} and \cite{Hua_Rahman_2024_ICC}  where $\mathbf{v}_A^H\mathbf{x}_A(k)$ for MISO channel becomes an effective probe to be estimated by Bob. So, the probing signal transmitted by Alice in  \cite{Hua_STEEP_2025} and \cite{Hua_Rahman_2024_ICC} has an artificial noise implicitly included as long as $n_A>n_B$.

\subsection{Phase 2 of MISO-SIMO-STEEP}
In phase 2, Bob sends out $\sqrt{p_B}(\hat r_{A|B}(k)+s_B(k))$ with power $p_B(\sigma_{\hat r_{A|B}}^2+1)<2p_B$, and Alice receives
\begin{align}\label{eq:yA2k}
  &\mathbf{y}_{A,2}(k) = \sqrt{p_B}\mathbf{h}_{A,B}(\hat r_{A|B}(k)+s_B(k))\notag\\
  &\,\,+\mathbf{w}_A(k)+\sqrt{p_{E,2}}\mathbf{H}_{A,E}\mathbf{v}_{A,E} n_{E,2}(k),
\end{align}
where $\sqrt{p_{E,2}}\mathbf{v}_{A,E} n_{E,2}(k)$ is the jamming noise from Eve. The optimal $\mathbf{v}_{A,E}$ under large $p_{E,2}$ is the same as that discussed in section \ref{sec:SIMO_conv}.

Note that
Eve  now receives
\begin{equation}\label{}
  \mathbf{y}_{E,2}(k)=\sqrt{p_B}\mathbf{h}_{E,B}(\hat r_{A|B}(k)+s_B(k))+\mathbf{w}_{E,2}(k),
\end{equation}
where the self-interference at Eve is assumed as before to be zero. Again $\mathbf{y}_{E,1}(k)$ and  $\mathbf{y}_{E,2}(k)$ here in sections \ref{sec:MISO_conv} and \ref{sec:MISO-SIMO-STEEP} are different from those in sections \ref{sec:SIMO_conv} and \ref{sec:SIMO-MISO-STEEP}.

\subsection{Optimal Estimation of $s_B(k)$ at Alice}
We know $\mathbb{E}\{\mathbf{y}_{A,2}(k)|r_A(k)\}=\sqrt{p_B}\mathbf{h}_{A,B}\sigma_{\hat r_{A|B}}^2 r_A(k)$. Then
\begin{align}\label{eq:DyA2k}
  &\Delta \mathbf{y}_{A,2}(k)\doteq \mathbf{y}_{A,2}(k)-\mathbb{E}\{\mathbf{y}_{A,2}(k)|r_A(k)\}\notag\\
  &=\sqrt{p_B}\mathbf{h}_{A,B}(\hat r_{A|B}(k)-\sigma_{\hat r_{A|B}}^2 r_A(k)+s_B(k))+
  \mathbf{w}_A(k)\notag\\
  &\,\,+\sqrt{p_{E,2}}\mathbf{H}_{A,E}\mathbf{v}_{A,E}n_{E,2}(k).
\end{align}
Then
the MSE of the MMSE estimate  $\hat s_{B|A}(k)$ of $s_B(k)$ from the knowledge of $\mathbf{y}_{A,2}(k)$ and $r_A(k)$  is
\begin{align}\label{eq:sDsBA}
  &\sigma_{\Delta s_{B|A}}^2
  =1-p_B\mathbf{h}_{A,B}^H\left (p_B(\sigma_{\Delta r_{A|B}}^2\sigma_{\hat r_{A|B}}^2+1)\mathbf{h}_{A,B}\mathbf{h}_{A,B}^H+\mathbf{I}\right .\notag\\
  &\,\,\left .+p_{E,2}\mathbf{H}_{A,E}\mathbf{v}_{A,E}\mathbf{v}_{A,E}^H\mathbf{H}_{A,E}^H\right )^{-1}\mathbf{h}_{A,B}.
\end{align}

\subsection{Properties of $\sigma_{\Delta s_{B|A}}^2$ in \eqref{eq:sDsBA}}
Applying the matrix inverse lemma, we have
\begin{align}\label{eq:sDsBA2}
  &\sigma_{\Delta s_{B|A}}^2
  =1-p_B\mathbf{h}_{A,B}^H\left (\mathbf{A}_4\right .\notag\\
  &\,\,\left .+p_{E,2}\mathbf{H}_{A,E}\mathbf{v}_{A,E}\mathbf{v}_{A,E}^H\mathbf{H}_{A,E}^H\right )^{-1}\mathbf{h}_{A,B}\notag\\
  &=1-p_B\mathbf{h}_{A,B}^H\left (\mathbf{A}_4^{-1}\right .\notag\\
  &\,\,\left .-p_{E,2}\mathbf{A}_4^{-1}\mathbf{H}_{A,E}\mathbf{v}_{A,E}b_4^{-1}
  \mathbf{v}_{A,E}^H\mathbf{H}_{A,E}^H\mathbf{A}_4^{-1}\right )\mathbf{h}_{A,B}.
\end{align}
where $\mathbf{A}_4=p_B(\sigma_{\Delta r_{A|B}}^2\sigma_{\hat r_{A|B}}^2+1)\mathbf{h}_{A,B}\mathbf{h}_{A,B}^H+\mathbf{I}$ and $b_4=1+p_{E,2}\mathbf{v}_{A,E}^H\mathbf{H}_{A,E}^H\mathbf{A}_4^{-1}
\mathbf{H}_{A,E}\mathbf{v}_{A,E}$. Using $\mathbf{A}_4^{-1}
=\mathbf{I}-\frac{p_B(\sigma_{\Delta r_{A|B}}^2\sigma_{\hat r_{A|B}}^2+1)}{
p_B(\sigma_{\Delta r_{A|B}}^2\sigma_{\hat r_{A|B}}^2+1)\|\mathbf{h}_{A,B}\|^2+1}\mathbf{h}_{A,B}\mathbf{h}_{A,B}^H$, we have
\begin{align}\label{}
 & b_4=1+p_{E,2}\left (\|\mathbf{H}_{A,E}\mathbf{v}_{A,E}\|^2\right .\notag\\
 &\,\,\left . -\frac{p_B(\sigma_{\Delta r_{A|B}}^2\sigma_{\hat r_{A|B}}^2+1)}{
p_B(\sigma_{\Delta r_{A|B}}^2\sigma_{\hat r_{A|B}}^2+1)\|\mathbf{h}_{A,B}\|^2+1}|\mathbf{h}_{A,B}^H\mathbf{H}_{A,E}
\mathbf{v}_{A,E}|^2\right )\notag\\
&=1+p_{E,2}\left (c_{A,E}'^2\right .\notag\\
 &\,\,\left . -\frac{p_B(\sigma_{\Delta r_{A|B}}^2\sigma_{\hat r_{A|B}}^2+1)\|\mathbf{h}_{A,B}\|^2}
  {p_B(\sigma_{\Delta r_{A|B}}^2\sigma_{\hat r_{A|B}}^2+1)\|\mathbf{h}_{A,B}\|^2+1}c_{A,E}^2
\right )
\end{align}
where $c_{A,E}^2=|\mathbf{\bar h}_{A,B}^H\mathbf{H}_{A,E}
\mathbf{v}_{A,E}|^2$ and $c_{A,E}'^2=\|\mathbf{H}_{A,E}\mathbf{v}_{A,E}\|^2$.
Then
\begin{equation}\label{}
  \lim_{p_B\to\infty}b_4
  =1+p_{E,2}(c_{A,E}'^2-c_{A,E}^2).
\end{equation}
Using $\mathbf{h}_{A,B}^H\mathbf{A}_4^{-1}
=\frac{1}{p_Bt_4+1}\mathbf{h}_{A,B}^H$ with $t_4=(\sigma_{\Delta r_{A|B}}^2\sigma_{\hat r_{A|B}}^2+1)\|\mathbf{h}_{A,B}\|^2$, we have
\begin{align}\label{eq:sDsBA3}
  &\sigma_{\Delta s_{B|A}}^2
  =1-p_B\left (\frac{\|\mathbf{h}_{A,B}\|^2}{p_Bt_4+1}\right .\notag\\
  &\,\,\left .-\frac{p_{E,2}}{(p_Bt_4+1)^2}
  \frac{\|\mathbf{h}_{A,B}^H\mathbf{H}_{A,E}\mathbf{v}_{A,E}\|^2}{b_4}\right )\notag\\
  &=1-p_B\left (\frac{\|\mathbf{h}_{A,B}\|^2}{p_Bt_4+1}-\frac{p_{E,2}\|\mathbf{h}_{A,B}\|^2c_{A,E}^2}
  {(p_Bt_4+1)^2b_4}\right )\notag\\
  &=\frac{p_B(t_4-\|\mathbf{h}_{A,B}\|^2)+1}{p_Bt_4+1}+
  \frac{p_Bp_{E,2}\|\mathbf{h}_{A,B}\|^2c_{A,E}^2}{(p_Bt_4+1)^2b_4}.
\end{align}
It is obvious that  the second term in \eqref{eq:sDsBA3} becomes zero if $p_B\to\infty$. One can also verify that the first term in \eqref{eq:sDsBA3} has a nonzero limit as $p_B\to\infty$ if $p_A>0$ or equivalent if $\sigma_{\Delta r_{A|B}}^2\sigma_{\hat r_{A|B}}^2>0$. Therefore,
it follows from \eqref{eq:sDsBA3} that for $p_A>0$,
\begin{equation}\label{}
 \lim_{p_B\to \infty} \sigma_{\Delta s_{B|A}}^2
 =\frac{\sigma_{\Delta r_{A|B}}^2\sigma_{\hat r_{A|B}}^2}{\sigma_{\Delta r_{A|B}}^2\sigma_{\hat r_{A|B}}^2+1}.
\end{equation}

\subsection{Optimal Estimation of $s_B(k)$  at Eve}
Recall that the signals received by Eve during phases 1 and 2 of MISO-SIMO-STEEP are:
\begin{equation}\label{}
  \left \{\begin{array}{c}
            \mathbf{y}_{E,1}(k) = \sqrt{p_A\gamma_A}\mathbf{g}_{E,A}r_A(k)+
  \sqrt{\frac{p_A(1-\gamma_A)}{n_A-1}}\mathbf{H}_{E,A}'\mathbf{n}_A(k)\\
  +\mathbf{w}_{E,1}(k), \\
            \mathbf{y}_{E,2}(k)=\sqrt{p_B}\mathbf{h}_{E,B}(\hat r_{A|B}(k)+s_B(k))+\mathbf{w}_{E,2}(k).
          \end{array}
  \right .
\end{equation}
The MSE of the MMSE estimate of $s_B(k)$ from $\mathbf{y}_{E,1}(k)$ and $\mathbf{y}_{E,2}(k)$ is
\begin{equation}\label{}
  \sigma_{\Delta s_{B|E}}^2
  =1-\mathbf{q}^H\mathbf{Q}^{-1}\mathbf{q}
\end{equation}
with $\mathbf{q}^H=\mathbb{E}\{s_B[\mathbf{y}_{E,1}^H(k),\mathbf{y}_{E,2}^H(k)]\}
=[0,\sqrt{p_B}\mathbf{h}_{E,B}^H]$ and
\begin{equation}\label{}
  \mathbf{Q}\doteq\mathbb{E}\left \{\left [\begin{array}{c}
                                  \mathbf{y}_{E,1}(k) \\
                                  \mathbf{y}_{E,2}(k)
                                \end{array}
  \right ]\left [\begin{array}{c}
                                  \mathbf{y}_{E,1}(k) \\
                                  \mathbf{y}_{E,2}(k)
                                \end{array}
  \right ]^H\right \}=\left [\begin{array}{cc}
                      \mathbf{Q}_{1,1} & \mathbf{Q}_{1,2} \\
                      \mathbf{Q}_{1,2}^H & \mathbf{Q}_{2,2}
                    \end{array}
   \right ].
\end{equation}
Here $\mathbf{Q}_{2,2}=p_B(\sigma_{\hat r_{A|B}}^2+1)\mathbf{h}_{E,B}\mathbf{h}_{E,B}^H+\mathbf{I}$, $\mathbf{Q}_{1,2}=\sigma_{\hat r_{A|B}}^2\sqrt{p_A\gamma_Ap_B}
  \mathbf{g}_{E,A}\mathbf{h}_{E,B}^H$, and
\begin{align}\label{}
  &\mathbf{Q}_{1,1}
  =p_A\gamma_A\mathbf{g}_{E,A}\mathbf{g}_{E,A}^H
  +\frac{p_A(1-\gamma_A)}{n_A-1}\mathbf{H}_{E,A}'\mathbf{H}_{E,A}'^H
  +\mathbf{I}.
\end{align}

It follows that
\begin{equation}\label{eq:sDsBE}
  \sigma_{\Delta s_{B|E}}^2
  =1-p_B\mathbf{h}_{E,B}^H(\mathbf{Q}_{2,2}-\mathbf{Q}_{1,2}^H\mathbf{Q}_{1,1}^{-1}
  \mathbf{Q}_{1,2})^{-1}\mathbf{h}_{E,B}.
\end{equation}

\subsection{Properties of $\sigma_{\Delta s_{B|E}}^2$ in \eqref{eq:sDsBE}}
Here
$\mathbf{Q}_{1,2}^H\mathbf{Q}_{1,1}^{-1}
  \mathbf{Q}_{1,2}=t_4'p_B\mathbf{h}_{E,B}\mathbf{h}_{E,B}^H
$
where
\begin{align}\label{eq:t4}
  &t_4'=\sigma_{\hat r_{A|B}}^4p_A\gamma_A\mathbf{g}_{E,A}^H
  \left (p_A\gamma_A\mathbf{g}_{E,A}\mathbf{g}_{E,A}^H\right .\notag\\
  &\,\,\left .+\frac{p_A(1-\gamma_A)}{n_A-1}\mathbf{H}_{E,A}'\mathbf{H}_{E,A}'^H
  +\mathbf{I}\right )^{-1}
  \mathbf{g}_{E,A}.
\end{align}
Hence, \eqref{eq:sDsBE} becomes
\begin{align}\label{eq:sDsBE2}
  &\sigma_{\Delta s_{B|E}}^2
  =1-\notag\\
  &\,\,p_B\mathbf{h}_{E,B}^H(p_B(\sigma_{\hat r_{A|B}}^2+1-t_4')\mathbf{h}_{E,B}\mathbf{h}_{E,B}^H+\mathbf{I}
  )^{-1}\mathbf{h}_{E,B}\notag\\
  &=1-\frac{p_B\|\mathbf{h}_{E,B}\|^2}{p_B(\sigma_{\hat r_{A|B}}^2+1-t_4')\|\mathbf{h}_{E,B}\|^2+1}\notag\\
  &=\frac{p_B(\sigma_{\hat r_{A|B}}^2-t_4')\|\mathbf{h}_{E,B}\|^2+1}{p_B(\sigma_{\hat r_{A|B}}^2+1-t_4')\|\mathbf{h}_{E,B}\|^2+1},
\end{align}
which implies that if $\sigma_{\hat r_{A|B}}^2-t_4'>0$,
\begin{equation}\label{}
  \lim_{p_B\to\infty}\sigma_{\Delta s_{B|E}}^2
  =\frac{\sigma_{\hat r_{A|B}}^2-t_4'}{\sigma_{\hat r_{A|B}}^2-t_4'+1}.
\end{equation}
We will show next that if $p_A>0$, then $t_4'<\sigma_{\hat r_{A|B}}^4<\sigma_{\hat r_{A|B}}^2$, and hence
\begin{align}\label{eq:lim_sigma}
 & \lim_{p_B\to\infty}\sigma_{\Delta s_{B|E}}^2
  >\frac{\sigma_{\hat r_{A|B}}^2-\sigma_{\hat r_{A|B}}^4}{\sigma_{\hat r_{A|B}}^2-\sigma_{\hat r_{A|B}}^4+1}\notag\\
  &
  =\frac{\sigma_{\Delta r_{A|B}}^2\sigma_{\hat r_{A|B}}^2}{\sigma_{\Delta r_{A|B}}^2\sigma_{\hat r_{A|B}}^2+1}
  =\lim_{p_B\to\infty}\sigma_{\Delta s_{B|A}}^2,
\end{align}
where we have used $\sigma_{\Delta r_{A|B}}^2\sigma_{\hat r_{A|B}}^2=\sigma_{\hat r_{A|B}}^2-\sigma_{\hat r_{A|B}}^4$.
The actual gap between $\lim_{p_B\to\infty}\sigma_{\Delta s_{B|E}}^2$ and $\lim_{p_B\to\infty}\sigma_{\Delta s_{B|A}}^2$ however depends on $p_A$ among other parameters.

\subsubsection{Proof of $t_4'<\sigma_{\hat r_{A|B}}^4$ under $p_A>0$}
To show $t_4'<\sigma_{\hat r_{A|B}}^4$,
let us now rewrite $t_4'$ in \eqref{eq:t4}, using the matrix inverse lemma, as
\begin{align}\label{eq:t4_2}
  &t_4'
  =\sigma_{\hat r_{A|B}}^4p_A\gamma_A\mathbf{g}_{E,A}^H\left (
  \mathbf{A}_4'^{-1}\right .\notag\\
  &\,\,\left .-\frac{p_A(1-\gamma_A)}{n_A-1}\mathbf{A}_4'^{-1}\mathbf{H}_{E,A}'\mathbf{B}_4'^{-1}
  \mathbf{H}_{E,A}'^H
  \mathbf{A}_4'^{-1}\right )\mathbf{g}_{E,A}
\end{align}
where $\mathbf{A}_4'=p_A\gamma_A\mathbf{g}_{E,A}\mathbf{g}_{E,A}^H+\mathbf{I}$ and $\mathbf{B}_4'=\mathbf{I}+\frac{p_A(1-\gamma_A)}{n_A-1}\mathbf{H}_{E,A}'^H
\mathbf{A}_4'^{-1}\mathbf{H}_{E,A}'$. Using $\mathbf{g}_{E,A}^H\mathbf{A}_4'^{-1}
=\frac{1}{p_A\gamma_A\|\mathbf{g}_{E,A}\|^2+1}\mathbf{g}_{E,A}^H$, we have
\begin{align}\label{eq:t4_3}
  &t_4'
  =\sigma_{\hat r_{A|B}}^4\left (\frac{p_A
  \gamma_A\|\mathbf{g}_{E,A}\|^2}{p_A\gamma_A\|\mathbf{g}_{E,A}\|^2+1}\right .\notag\\
  &\,\,\left .-
  \frac{p_A^2\gamma_A(1-\gamma_A)}{n_A-1}\frac{\|\mathbf{g}_{E,A}\|^2}
  {(p_A\gamma_A\|\mathbf{g}_{E,A}\|^2+1)^2}t_4''\right )\notag\\
  &<\sigma_{\hat r_{A|B}}^4\frac{p_A
  \gamma_A\|\mathbf{g}_{E,A}\|^2}{p_A\gamma_A\|\mathbf{g}_{E,A}\|^2+1}<\sigma_{\hat r_{A|B}}^4,
\end{align}
where $t_4''=\mathbf{c}_{E,A}^H\mathbf{B}_4'^{-1}
  \mathbf{c}_{E,A}$ and $\mathbf{c}_{E,A}=\mathbf{H}_{E,A}'^H\mathbf{\bar g}_{E,A}$.

\subsection{Secrecy rate of MISO-SIMO-STEEP}

Similar to the notion behind \eqref{eq:RsAB}, and using \eqref{eq:sDsBA} and \eqref{eq:sDsBE}, an achievable secrecy rate of MISO-SIMO-STEEP from Bob to Alice in bits per round-trip channel use is $R_{s,\texttt{M-S-STEEP}}^+$ with
\begin{equation}\label{eq:RsBA}
  R_{s,\texttt{M-S-STEEP}}=C_{B|A}-C_{B|E}=\log\frac{\sigma_{\Delta s_{B|E}}^2}{\sigma_{\Delta s_{B|A}}^2}.
\end{equation}
Then from \eqref{eq:lim_sigma}, we have:
\begin{Proposition}
For $p_A>0$,
\begin{equation}\label{eq:limit_MISO-SIMO-STEEP}
  \lim_{p_B\to\infty}R_{s,\texttt{M-S-STEEP}}>0.
\end{equation}
Namely, under virtually any channel conditions and $p_A>0$, there is  a finite $\bar p_B$ such that $R_{s,\texttt{M-S-STEEP}}>0$ when $p_B>\bar p_B$.
\end{Proposition}

 Like the conventional SIMO versus SIMO-MISO-STEEP, the conventional MISO scheme and the MISO-SIMO-STEEP scheme transmit secret messages in the opposite directions. But for upper layers, such a change of direction has little impact. This is because any secret received by Alice from Bob can be immediately used to protect a message from Alice to Bob.

 However, the secrecy rates $R_{s,\texttt{S-M-STEEP}}$ and $R_{s,\texttt{M-S-STEEP}}$ are in bits per round-trip symbol interval while $R_{s,\texttt{SIMO,conv}}$ and $R_{s,\texttt{MISO,conv}}$ are in bits per single-trip symbol interval. A useful comparison should be between $R_{s,\texttt{S-M-STEEP}}^+$ and $R_{s,\texttt{SIMO,conv}}^+ + R_{s,\texttt{MISO,conv}}^+$ or between $R_{s,\texttt{M-S-STEEP}}^+$ and $R_{s,\texttt{SIMO,conv}}^+ + R_{s,\texttt{MISO,conv}}^+$. Such a comparison will be shown in section \ref{sec:Simulation}.

 \section{Complexity, Latency and Other Issues}\label{sec:complexity}

 The preceding discussions included both the protocols of SIMO-MISO-STEEP and MISO-SIMO-STEEP and the analyses of the secrecy rates of the two forms of STEEP. In this section, we highlight the key steps in these protocols, which makes clear that the added complexity for STEEP in comparison to a conventional scheme is not high. We also discuss the latency  and other issues of importance.

 \subsection{Complexity of SIMO-MISO-STEEP}\label{sec:SIMO-MISO-STEEP_b}
 In phase 1, Bob simply transmits a random scalar sequence $r_B(k)$, and Alice receives $\mathbf{y}_{A,1}(k)$ as in \eqref{eq:yA1k} and computes $\hat r_{B|A}(k)$ as in \eqref{eq:hrBAk}. This is the main additional complexity.

 In phase 2, Alice transmits a vector sequence $\mathbf{x}_{A,2}(k)$ as shown in \eqref{eq:xA2} where $s_A(k)$ is a coded sequence with a secret message $\mathcal{M}_A$ from Alice, and Bob receives $y_{B,2}(k)$ as in \eqref{eq:yB2k}. For reliable detection of $\mathcal{M}_A$ at Bob, the channel coding applied to $s_A(k)$ at Alice should be done at a rate in bits/s/Hz close to the SIMO-MISO-STEEP-induced effective capacity from Alice to Bob $C_{A|B}$ in \eqref{eq:CAB}.
Bob only needs to compute  $\Delta y_{B,2}(k)$ as in \eqref{eq:DyB2k} and detect $\mathcal{M}_A$ from the sequence $\Delta y_{B,2}(k)$.

 After the phase 2 operation, the entropy of $\mathcal{M}_A$ given Eve's knowledge can be (approximately) as high as $KR^+_{\texttt{s,S-M-STEEP}}$ where $K$ is the range of $k=1,\cdots,K$, i.e., $\mathbb{H}(\mathcal{M}_A|\texttt{Eve})\leq KR_{\texttt{s,S-M-STEEP}}=K(C_{A|B}-C_{E|B})$. This entropy could be much smaller than $|\mathcal{M}_A|\leq KC_{A|B}$. Here $\mathbb{H}(\cdot|\texttt{Eve})$ denotes the entropy conditioned on Eve's observations, and $|\cdot|$ denotes the size.

  Note that we can also convert the above $C_{A|B}$ into $C_{A|B}=\log(1+\texttt{SNR}_{A|B})$ where $\texttt{SNR}_{A|B}$ is an effective SNR for $s_A(k)$ at Bob.   If $K$ is in the range of several thousands, practical channel codes such as the LDPC codes can typically achieve a SNR gap of less than 0.5dB from capacity.

 Alternatively, if $R^+_{\texttt{s,S-M-STEEP}}$ is known to users, a wiretap channel code (or coset code) is available to keep  $\mathcal{M}_A$ in complete secrecy from Eve at a rate approaching $R^+_{\texttt{s,S-M-STEEP}}$, i.e., $\mathbb{H}(\mathcal{M}_A|\texttt{Eve}) =|\mathcal{M}_A|\leq KR^+_{\texttt{s,S-M-STEEP}}$.

 \subsection{Complexity of MISO-SIMO-STEEP}
 In phase 1, Alice transmits the random vector sequence $\mathbf{x}_{A,1}(k)$ as in \eqref{eq:xA1}, and Bob receives $y_{B,1}(k)$ as in \eqref{eq:yB1k} and computes $\hat r_{A|B}(k)$ as in \eqref{eq:hrABk}.

 In phase 2, Bob transmits $\sqrt{p_B}(\hat r_{A|B}(k)+s_B(k))$ where $s_B(k)$ is a coded sequence with a message $\mathcal{M}_B$ from Bob, and Alice receives $\mathbf{y}_{A,2}(k)$ as in \eqref{eq:yA2k}. For reliable detection of $\mathcal{M}_B$ at Alice, the data rate of $\mathcal{M}_B$ in bits/s/Hz must be less than the MISO-SIMO-STEEP-induced effective capacity from Bob to Alice $C_{B|A}=-\log\sigma_{\Delta s_{B|A}}^2$ with $\sigma_{\Delta s_{B|A}}^2$ given in \eqref{eq:sDsBA}.
And Alice  needs to compute $\Delta\mathbf{y}_{A,2}(k)$ as in \eqref{eq:DyA2k} and detect $\mathcal{M}_B$ from the sequence $\Delta\mathbf{y}_{A,2}(k)$.

Unlike SIMO-MISO-STEEP where artificial noise (AN) from user is applied in phase 2, MISO-SIMO-STEEP uses AN from user in phase 1. This is a key reason that MISO-SIMO-STEEP tends to yield a higher secrecy rate than SIMO-MISO-STEEP.

After phase 2, the entropy of $\mathcal{M}_B$ given Eve's knowledge can be  as high as $KR^+_{\texttt{s,M-S-STEEP}}$.
The other discussions in section \ref{sec:SIMO-MISO-STEEP_b} also apply similarly here.

\subsection{Knowledge of Eve's Channels}
Two types of channel state information (CSI) at Eve are relevant here. One is Eve's jamming CSI (i.e., $\mathbf{H}_{A,E}$ and $\mathbf{h}_{B,E}$), and the other is Eve's receive CSI (i.e., $\mathbf{H}_{E,A}$ and $\mathbf{h}_{E,B}$).

For reliable transmission of $\mathcal{M}_A$ for SIMO-MISO-STEEP and $\mathcal{M}_B$ for MISO-SIMO-STEEP, the received jamming powers (including the received jamming power at Bob $p_E\|\mathbf{h}_{B,E}\|^2$ and the received jamming covariance matrix at Alice $p_E\mathbf{H}_{A,E}\mathbf{v}_{A,E}\mathbf{v}_{A,E}^H\mathbf{H}_{A,E}^H$)  need to be known because the effective channel capacity in either form of STEEP depends on them. In most situations, these received jamming powers can be measured online in real time by users before STEEP is executed. In other words, Eve's jamming CSI is not exactly needed.

As for Eve's receive CSI, the users do not need it either in order to  transmit $\mathcal{M}_A$ or $\mathcal{M}_B$ reliably. But the actual value of the secrecy rate of STEEP depends on $\mathbf{H}_{E,A}$ and $\mathbf{h}_{E,B}$. For most practical applications, $\mathbf{H}_{E,A}$ and $\mathbf{h}_{E,B}$ are unknown to users, so are $R^+_{\texttt{s,S-M-STEEP}}$ and $R^+_{\texttt{s,M-S-STEEP}}$, and hence a wiretap channel code in this case is not applicable to keep $\mathcal{M}_A$ (or $\mathcal{M}_B$) in complete secrecy from Eve. In this case, a practical approach is to use a standard channel code to keep $\mathcal{M}_A$ (or $\mathcal{M}_B$) at a rate as close as possible to $C_{A|B}$ (or $C_{B|A}$), and the (partial) secrecy in $\mathcal{M}_A$ (or $\mathcal{M}_B$) is governed by the secrecy rate $R_{\texttt{s,S-M-STEEP}}$ (or $R_{\texttt{s,M-S-STEEP}}$) as discussed before.

It is also worth noting that (for example) if $R_{\texttt{s,S-M-STEEP}}=C_{A|B}-C_{A|E}>0$ and $\mathcal{M}_A$ is transmitted at a rate close to  $C_{A|B}$,  Bob can receive $\mathcal{M}_A$ reliably but Eve cannot. In practice, it is still difficult (although not impossible) for Eve to compute, and make use of, her ambiguity space about $\mathcal{M}_A$ subject to $R_{\texttt{s,S-M-STEEP}}>0$ and a large $K$ \cite{Hua_AAA_2025}. Until Eve is able to compute and make use of that ambiguity space, the ``entropy'' of $\mathcal{M}_A$ against Eve is practically close to $KC_{A|B}$ rather than the strict measure $KR_{\texttt{s,S-M-STEEP}}=K(C_{A|B}-C_{A|E})$.

\subsection{Latency}
STEEP is a physical-layer round-trip scheme, which could double the latency in comparison to a conventional single-trip scheme especially if time division duplex (TDD) is used between Alice and Bob. But the packet $\mathcal{P}_1$ transmitted in phase 1 is uncoded. So, as mentioned below \eqref{eq:xA2}, if frequency division duplex (FDD) is used, the packet $\mathcal{P}_2$ transmitted in phase 2 could be mostly overlapped with $\mathcal{P}_1$ in time. This way, the latency can be substantially reduced unless the propagation time dominates the packet duration as in satellite communications.

If Eve's channel cannot be made weaker than user's channel via any other means (most often due to infrastructure cost), then STEEP is among the very few options to achieve a positive secrecy rate. A conventional protocol for secret-key generation (SKG) using the data sets established in phase 1 of STEEP would incur much more latency, as it requires additional communications for quantization, information reconciliation and privacy amplification. SKG using the estimates of reciprocal channel parameters between Alice and Bob would also cause a significant latency (before any secret message can be transmitted), and furthermore its secret-key rate in bits per second is  highly limited due to limited channel variability in practice.

\subsection{How to Ensure a Positive Secrecy Rate}
Without any knowledge of Eve's receive CSI, no physical layer security method could ensure a positive secrecy rate.  If Eve's parameters (such as $p_E$ and $n_E$) are known to be bounded, a positive secrecy rate of STEEP can be achieved simply by increasing  the phase-2 transmit power. See \eqref{eq:limit_SIMO-MISO-STEEP} for SIMO-MISO-STEEP and \eqref{eq:limit_MISO-SIMO-STEEP} for MISO-SIMO-STEEP. In other words, if users know an upper bound on each of the key parameters along with a valid statistical model for both users's and Eve's CSI, then the users can choose a phase-2 power and a corresponding phase-1 power to ensure a positive secrecy rate with a high confidence. The computations needed for such a decision can  be done online or offline. Examples of such computations are shown in the next section.

Note that the secrecy rate $R_s$ of STEEP is measured by (secret) bits per round-trip channel use. The corresponding secrecy rate in (secret) bits per second is in the order of $WR_s$ with $W$ being the bandwidth used by each symbol. For many practical wireless channels, $W$  ranges from 10MHz to 100MHz. If $R_s=0.01$, $WR_s$ ranges from $10^5$ to $10^6$ in (secret) bits per second, which is very significant for transmission of extremely sensitive information such as a long secret key.

\section{Secrecy Rates and Simulation Results}\label{sec:Simulation}

It is shown by \eqref{eq:limit_SIMO-MISO-STEEP} and \eqref{eq:limit_MISO-SIMO-STEEP}  that for SIMO-MISO-STEEP and MISO-SIMO-STEEP, a positive secrecy rate can be achieved by choosing a sufficiently large power in the second phase, subject to virtually any given channel conditions and given jamming power by Eve. It is also shown by \eqref{eq:space_SIMO} and \eqref{eq:space_MISO} that for the conventional SIMO and MISO schemes under $n_E\geq n_A$, the achievable secrecy rate is  zero within a large space of channel conditions and jamming powers by Eve.

It is also important to add that for the conventional MISO scheme, if $n_A>n_E$, then the secrecy rate $R_{s,\texttt{MISO,conv}}$ would scale with large $p_A$ as $\mathcal{O}(\log p_A)$, which is the same as the scaling law of MISO-SIMO-STEEP subject to $p_B\gg p_A\gg 1$. In other words, the secure degree of freedom of the conventional MISO scheme and MISO-SIMO-STEEP is known to be one subject to $n_A>n_E\geq n_B=1$ \cite{Hua_STEEP_2025}, \cite{HuaMaksud2024} But the conventional MISO scheme does not need $p_B$.

The above suggests that STEEP is more important for situations where $n_E\geq n_A$ or $n_E$ has to be assumed to be no less than $n_A$ to be conservative. So, the simulation shown below focuses only on the challenging case where $n_E\geq n_A$. This is a case where virtually all schemes in \cite{Thien2022}-\cite{Abdalla2023} would fail to yield a positive secrecy rate within a large space of channel conditions.

We now provide simulation results to illustrate the secrecy performances of all four schemes subject to the same conditions of channels and powers. The simulations are based on the expressions of $R_{s,\texttt{SIMO,conv}}$, $R_{s,\texttt{S-M-STEEP}}$, $R_{s,\texttt{MISO,conv}}$ and $R_{s,\texttt{M-S-STEEP}}$ given by \eqref{eq:R_s_conv_BA}, \eqref{eq:RsAB}, \eqref{eq:R_s_conv_AB} and \eqref{eq:RsBA} respectively.

We will choose $p_{E,1}=p_{E,2}=p_E$, $\mathbf{H}_{A,E}=\mathbf{H}_{E,A}^T$, $\mathbf{h}_{A,B}=\mathbf{h}_{B,A}$, and $\mathbf{h}_{B,E}=\mathbf{h}_{E,B}$. Subject to the above, we choose all channel parameters independently from $\mathcal{CN}(0,1)$ with $10^5$ independent realizations.
We also choose $p_E=10$dB and $\gamma_A=0.5$.

Note that $p_A$ is the power consumed by Alice for MISO and MISO-SIMO-STEEP, and $p_B$ is the power consumed by Bob for SIMO and SIMO-MISO-STEEP. But $p_A+p_A\gamma_A\sigma_{\hat r_{B|A}}^2<(1+\gamma_A)p_A$ is the actual power consumed by Alice for SIMO-MISO-STEEP, and $p_B+p_B\sigma_{\hat r_{A|B}}^2<2p_B$ is the actual power consumed by Bob for MISO-SIMO-STEEP.

\subsection{Optimal phase-1 power for STEEP}
The secrecy rate of either SIMO-MISO-STEEP or MISO-SIMO-STEEP  depends on its phase-1 and phase-2 powers. While its dependency on its phase-2 power is generally monotonically increasing, its dependency on its phase-1 power is not. The optimal phase-1 power for given channel conditions is typically somewhere larger than zero and less than the phase-2 power.


 Since the optimal phase-1 power depends on channel realizations, we show next some figures to illustrate the distributions of $(\hat R_s, p_{1,opt})$ where  $\hat R_s$ is the maximized secrecy rate and $p_{1,opt}$ is the optimal phase-1 power, i.e., $\hat R_s=\max_{p_1}R_s$, and $p_{1,opt}=arg\max_{p_1}R_s$. Here $p_1=p_A$ and $R_s=R_{s,\texttt{M-S-STEEP}}$ for MISO-SIMO-STEEP, and $p_1=p_B$ and $R_s=R_{s,\texttt{S-M-STEEP}}$ for SIMO-MISO-STEEP.

\begin{figure}[ht]
\begin{minipage}[b]{0.48\linewidth}
\centering
\includegraphics[width=\textwidth]{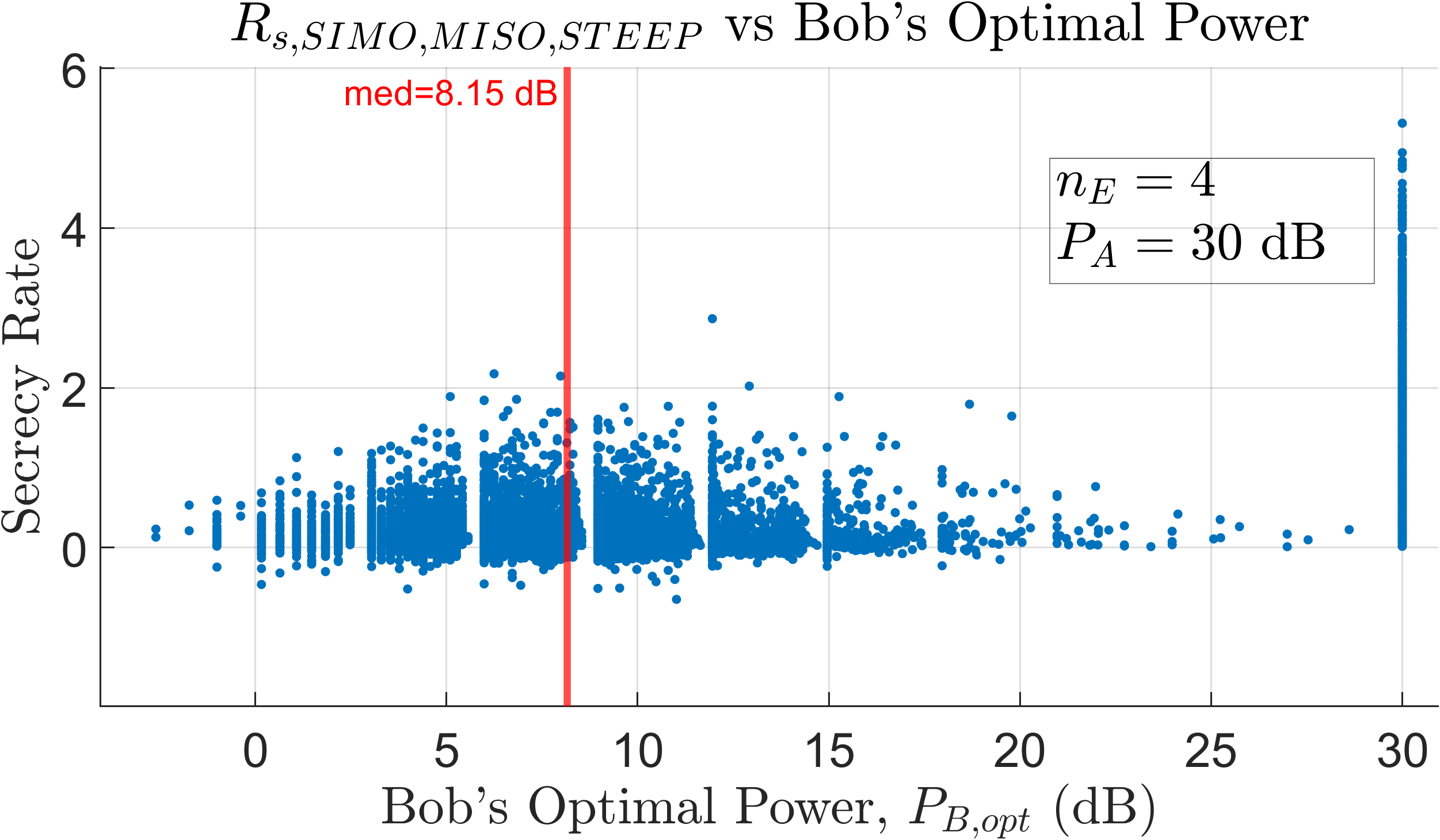}
\subcaption{$p_A=30$dB. }
\label{fig:secrecy_steep_pB_30}
\end{minipage}
\hspace{0.1cm}
\begin{minipage}[b]{0.48\linewidth}
\centering
\includegraphics[width=\textwidth]{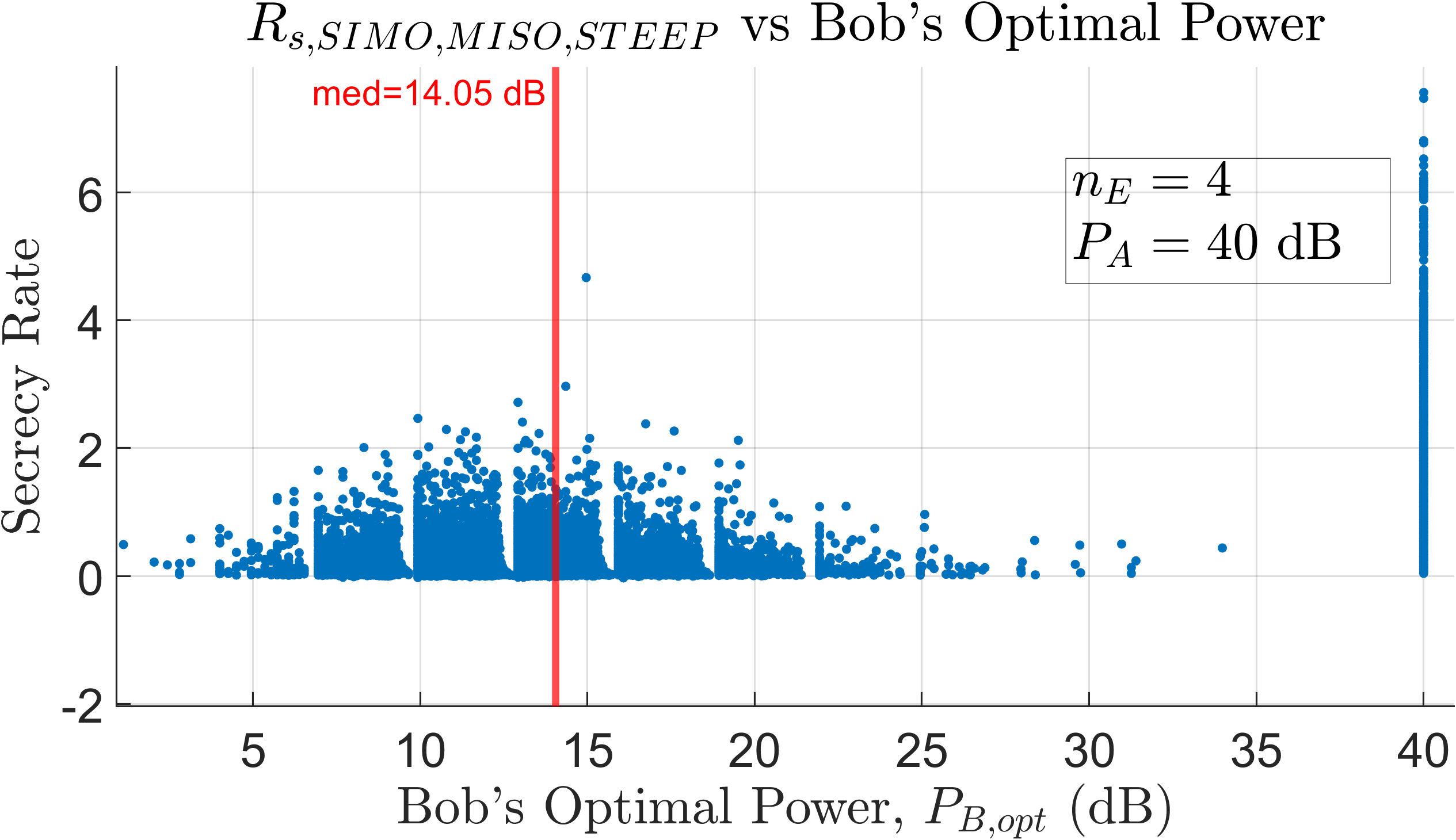}
\subcaption{$p_A=40$dB. }
\label{fig:secrecy_steep_pB_40}
\end{minipage}\\
\\
\begin{minipage}[b]{0.48\linewidth}
\centering
\includegraphics[width=\textwidth]{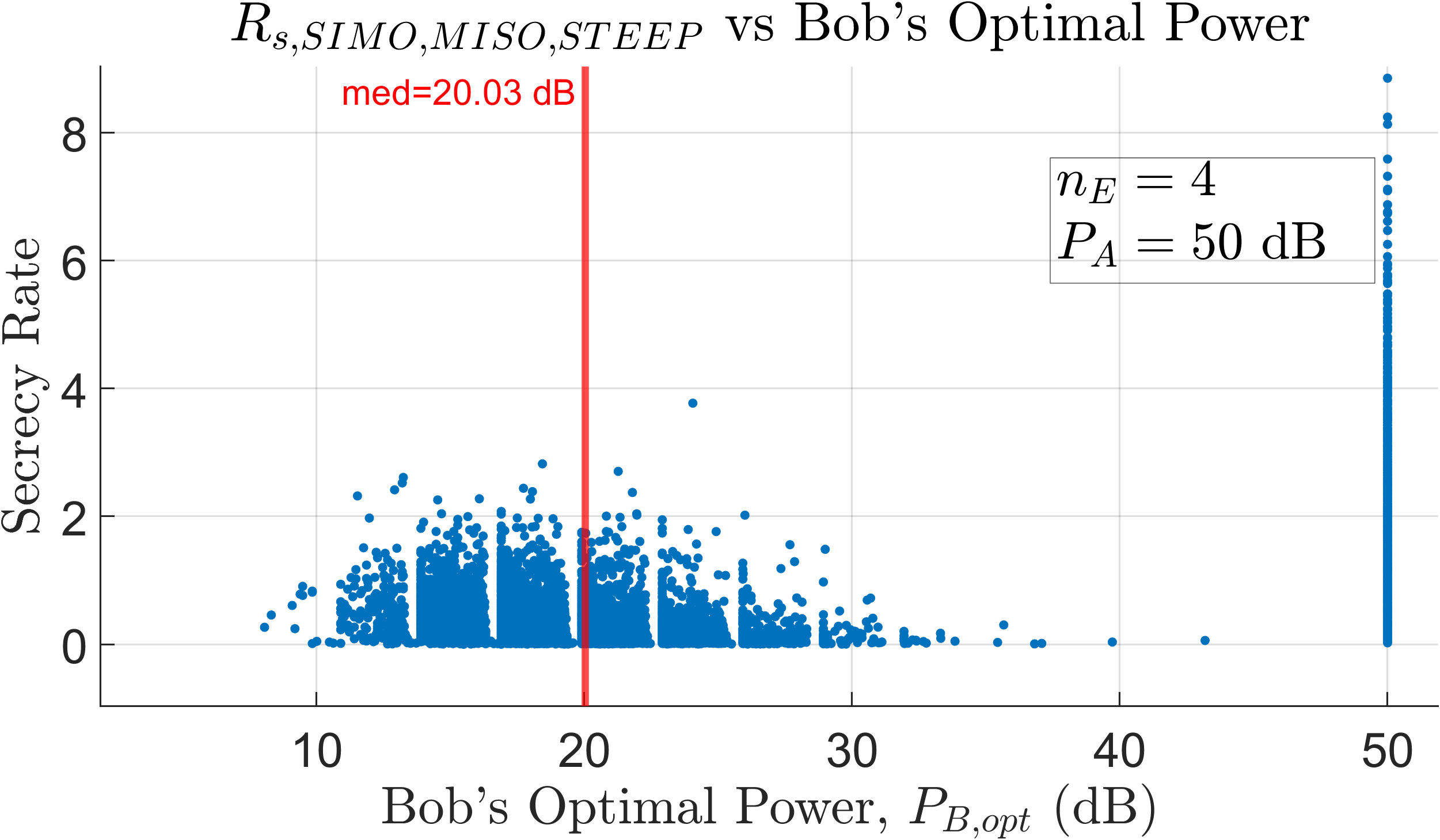}
\subcaption{$p_A=50$dB,}
\label{fig:secrecy_steep_pB_30_rho_0}
\end{minipage}
\hspace{0.1cm}
\begin{minipage}[b]{0.48\linewidth}
\centering
\includegraphics[width=\textwidth]{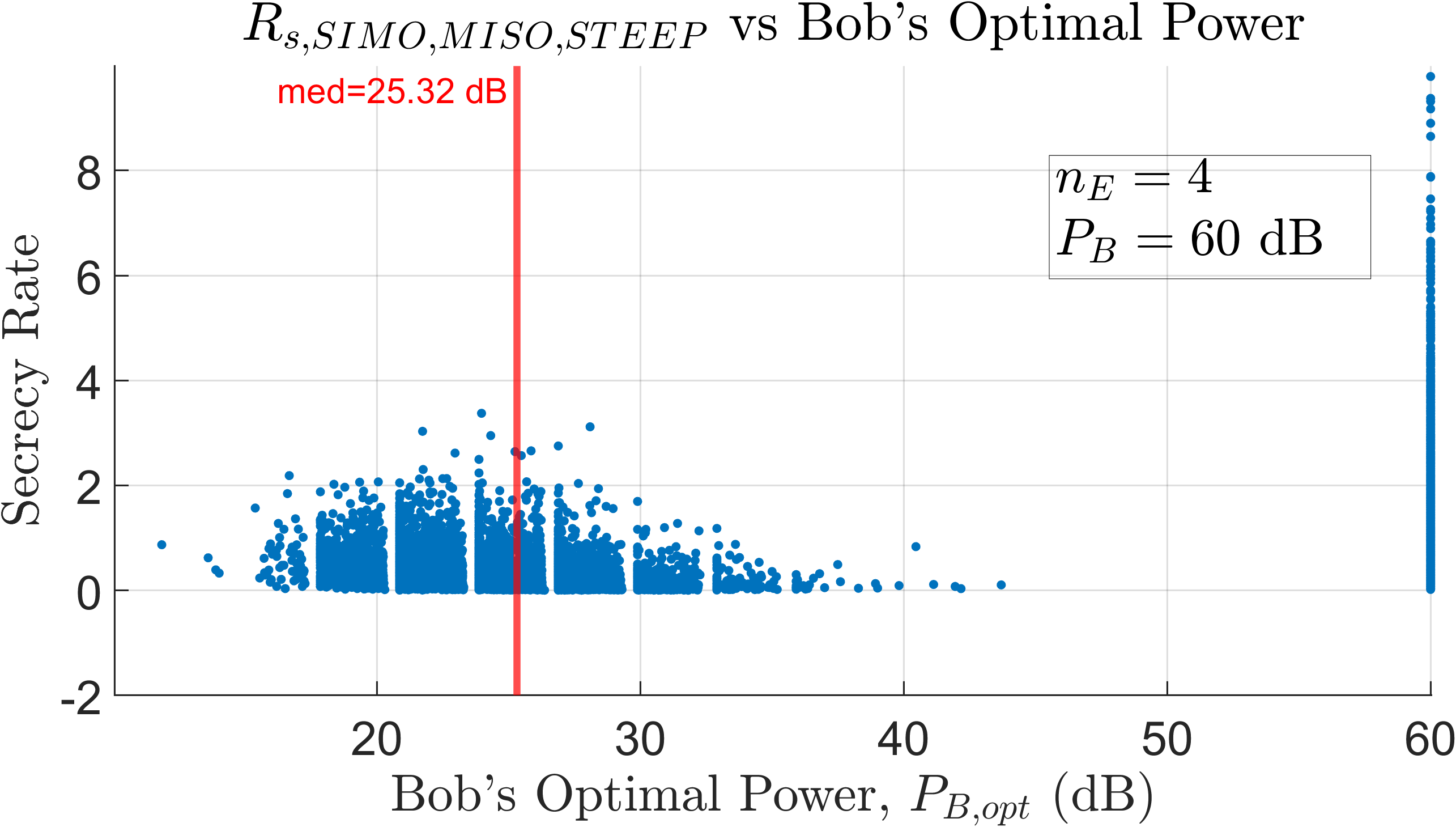}
\subcaption{$p_A=60$dB. }
\label{fig:secrecy_steep_pB_40_rho_0}
\end{minipage}
\caption{Distributions of  $(\hat R_{s,\texttt{S-M-STEEP}},p_{B,opt})$ with $p_E=10$dB, $n_A=4$, $n_E=4$. In this and other figures, $R_{s,SIMO,MISO,STEEP}$ in the legend is the same as $R_{s,\texttt{S-M-STEEP}}$. We will also use $R_{s,MISO,SIMO,STEEP}$ in the legend for $R_{s,\texttt{M-S-STEEP}}$.}
\label{fig:secrecy_rate_vs_optimized_p_B_simo_miso_steep_n_E_4}
\end{figure}

Shown in Fig. \ref{fig:secrecy_rate_vs_optimized_p_B_simo_miso_steep_n_E_4} are the distributions of $(\hat R_s, p_{1,opt})$ for SIMO-MISO-STEEP, i.e., $\hat R_{s,\texttt{S-M-STEEP}}$ versus $p_{B,opt}$,  subject to $10^5$ independent channel realizations. The four subplots in Fig. \ref{fig:secrecy_rate_vs_optimized_p_B_simo_miso_steep_n_E_4} correspond to $p_A=30,40,50,60$dB respectively.

The optimization over $p_B$ was done via bi-section search of $p_B$ (in linear space) between its minimum choice zero and its maximum choice $p_A$. This is the reason behind the periodic (thin) white bands in the plots, which show a 3dB gap between every two adjacent white bands along the optimal $p_B$ in dB. In other words, the converged optimal value of $p_B$ will not land on any of the points equal to $2^{-l}p_A$ with $l=1,2, \cdots$, and more generally the converged optimal value of $p_B$ does not fall into a region around any of these points. But the width of each of the white bands reduces as the maximum number of iterations of the bi-section search increases.

For vast majority of the $10^5$ realizations, the optimal $p_B$ is somewhere far away from zero and $p_A$. But for some few cases, the optimal $p_B$ equals to its chosen upper bound $p_A$. These cases correspond to the situations where Eve's receive channel in phase 1 is very weak compared to user's receiver channel in phase 1 so that the phase-1 transmitter should apply the maximum power. Also note that when $\hat p_B=p_A$, the secrecy rate $R_{s,\texttt{S-M-STEEP}}$ still depends on Eve's receive channel in phase 2 relative to user's receive channel in phase 2, which is shown by the vertical cluster of dots on the far right side of each of the four plots in Fig. \ref{fig:secrecy_rate_vs_optimized_p_B_simo_miso_steep_n_E_4}.

The median value of $p_{B,opt}$ in each of the four plots in Fig. \ref{fig:secrecy_rate_vs_optimized_p_B_simo_miso_steep_n_E_4} is marked by a red vertical line, which will be denoted by $\bar p_{B,opt}$. Without the knowledge of Eve's receive channels (other than its statistical distribution), $\bar p_{B,opt}$ is clearly a good choice of the phase-1 power for SIMO-MISO-STEEP.

\begin{figure}[ht]
\begin{minipage}[b]{0.48\linewidth}
\centering
\includegraphics[width=\textwidth]{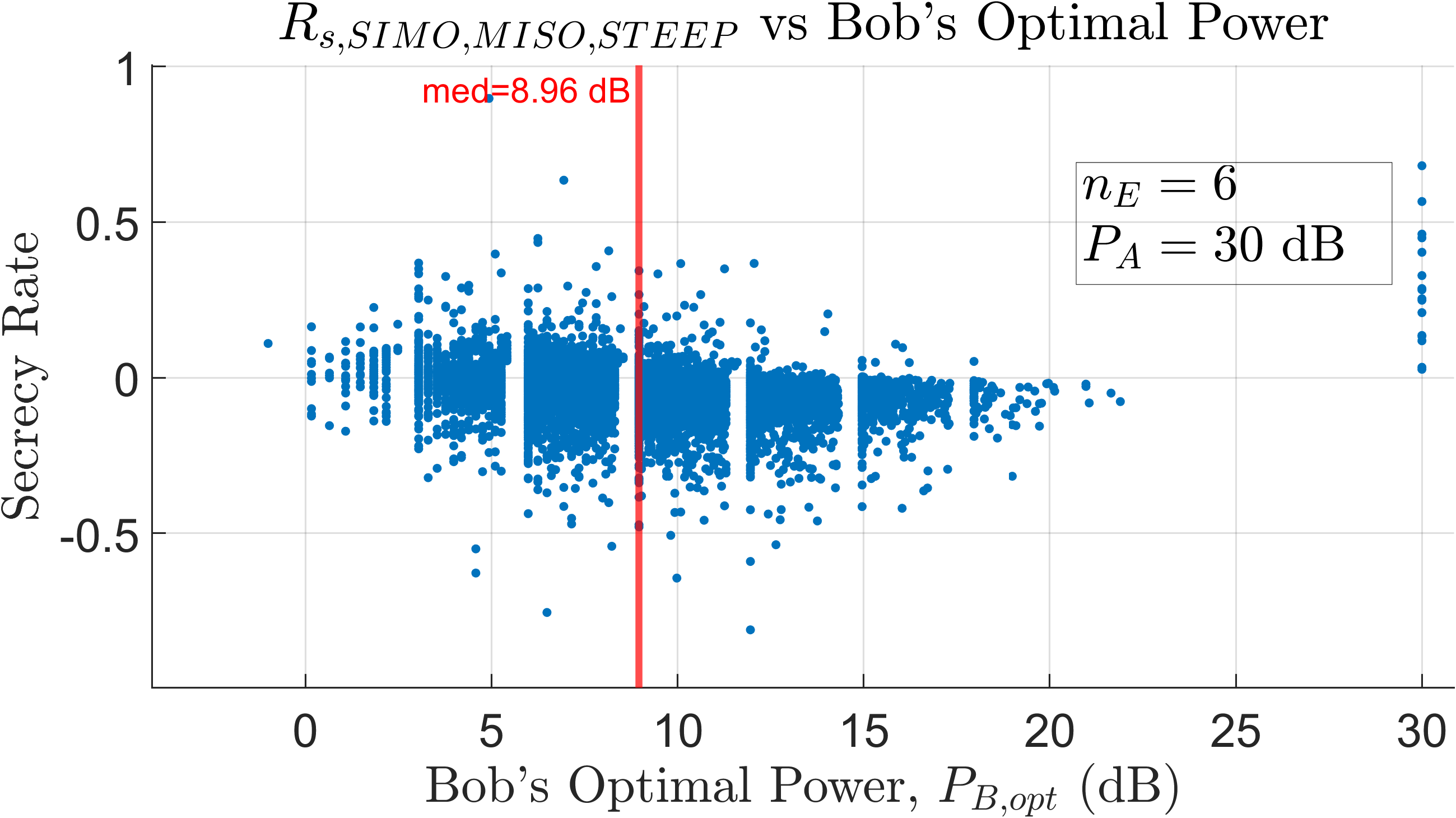}
\subcaption{$p_A=30$dB. }
\label{fig:secrecy_steep_pB_30}
\end{minipage}
\hspace{0.1cm}
\begin{minipage}[b]{0.48\linewidth}
\centering
\includegraphics[width=\textwidth]{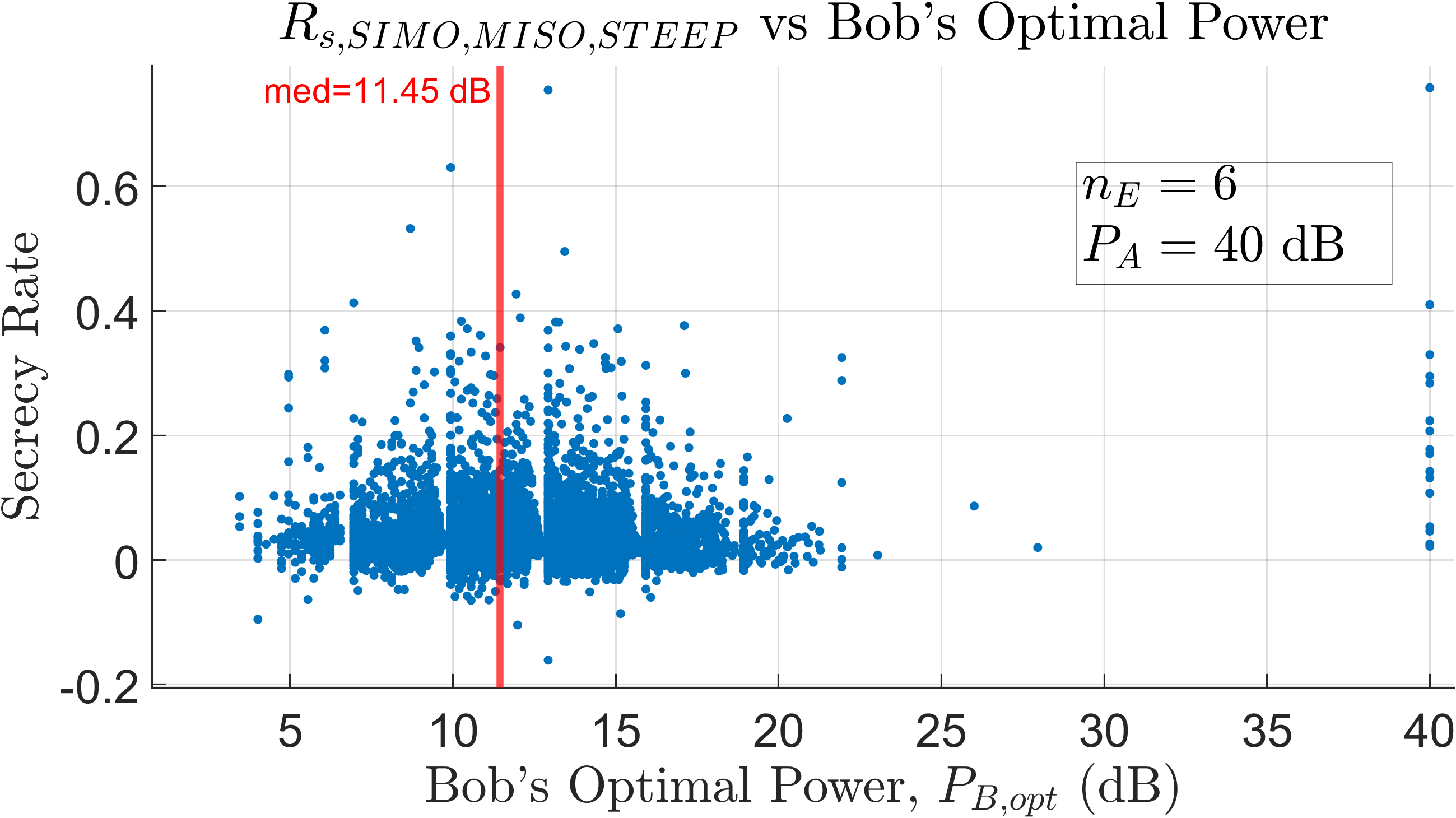}
\subcaption{$p_A=40$dB. }
\label{fig:secrecy_steep_pB_40}
\end{minipage}\\
\\
\begin{minipage}[b]{0.48\linewidth}
\centering
\includegraphics[width=\textwidth]{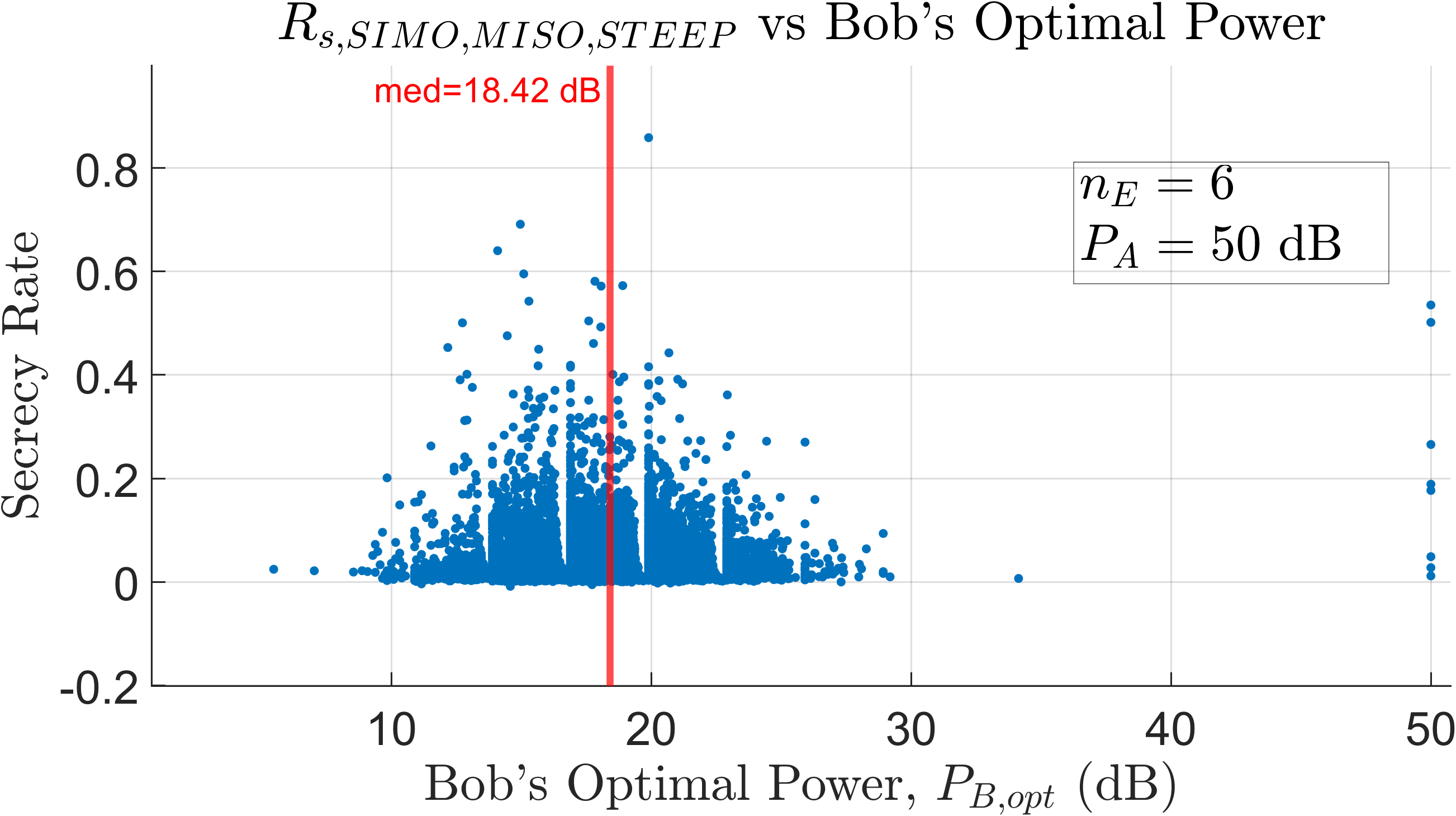}
\subcaption{$p_A=50$dB,}
\label{fig:secrecy_steep_pB_30_rho_0}
\end{minipage}
\hspace{0.1cm}
\begin{minipage}[b]{0.48\linewidth}
\centering
\includegraphics[width=\textwidth]{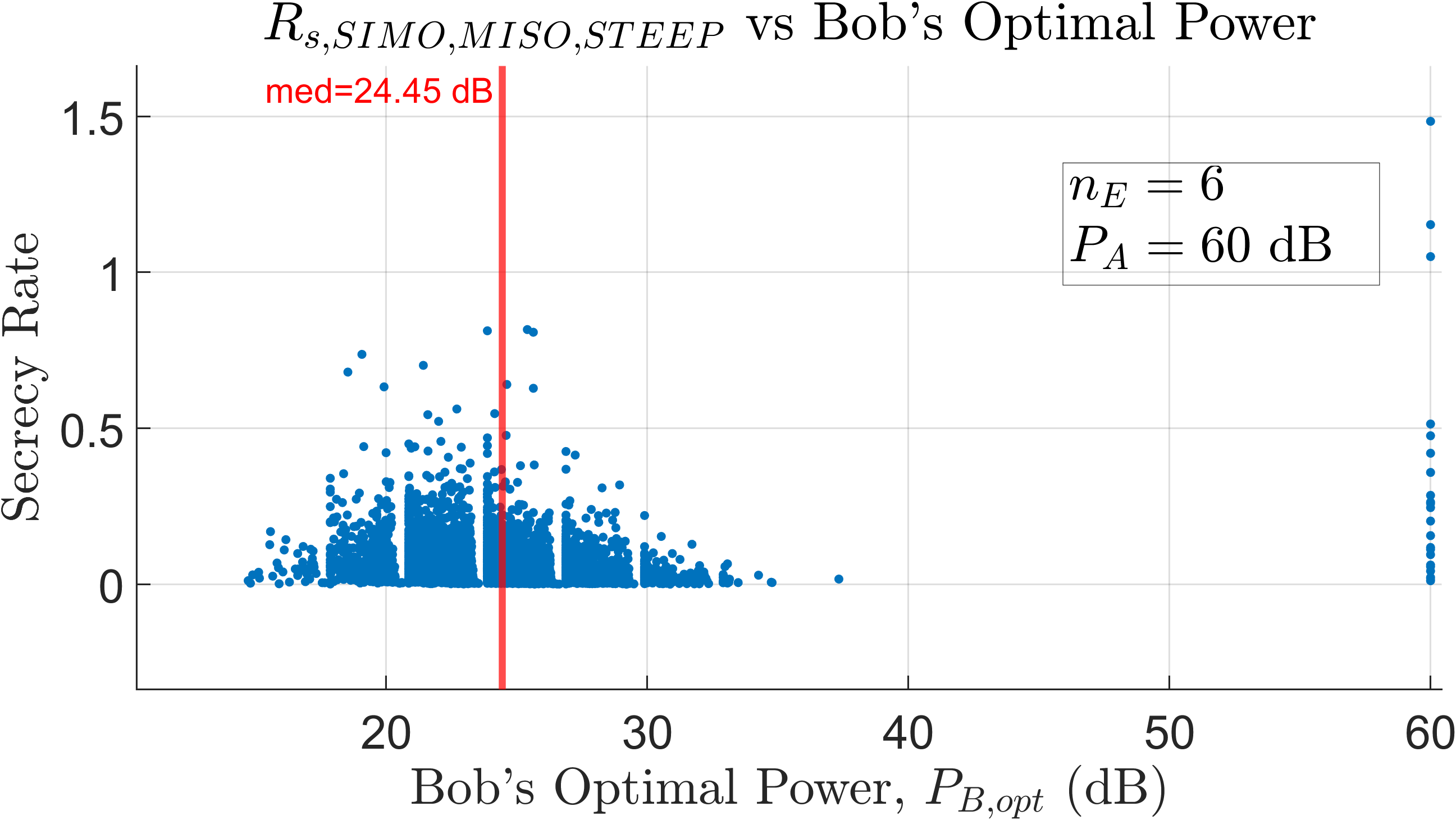}
\subcaption{$p_A=60$dB. }
\label{fig:secrecy_steep_pB_40_rho_0}
\end{minipage}
\caption{Distributions of  $(\hat R_{s,\texttt{S-M-STEEP}},p_{B,opt})$ with $p_E=10$dB, $n_A=4$, $n_E=6$.}
\label{fig:secrecy_rate_vs_optimized_p_B_simo_miso_steep_n_E_6}
\end{figure}

\begin{figure}[ht]
\begin{minipage}[b]{0.48\linewidth}
\centering
\includegraphics[width=\textwidth]{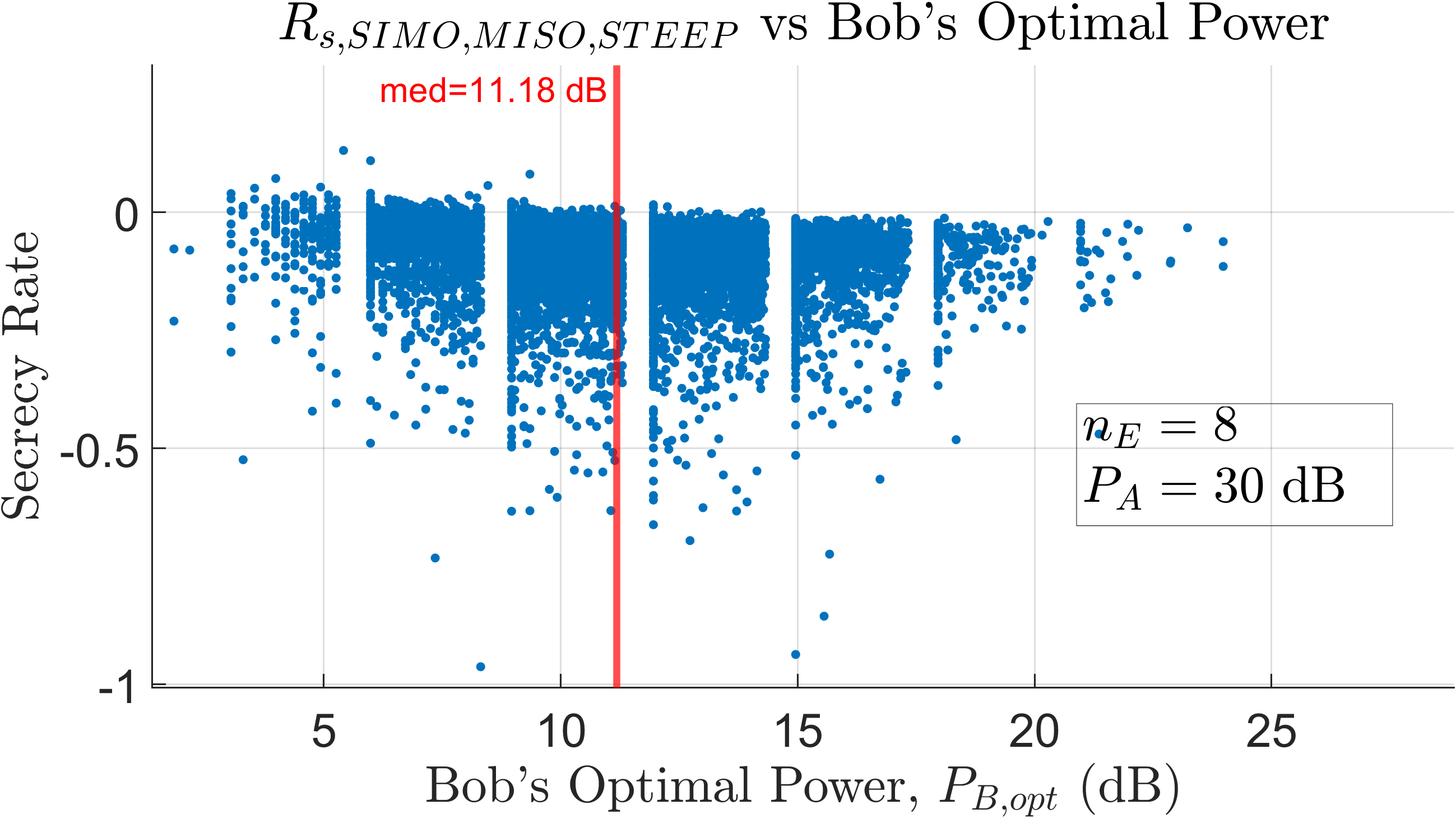}
\subcaption{$p_A=30$dB. }
\label{fig:secrecy_steep_pB_30}
\end{minipage}
\hspace{0.1cm}
\begin{minipage}[b]{0.48\linewidth}
\centering
\includegraphics[width=\textwidth]{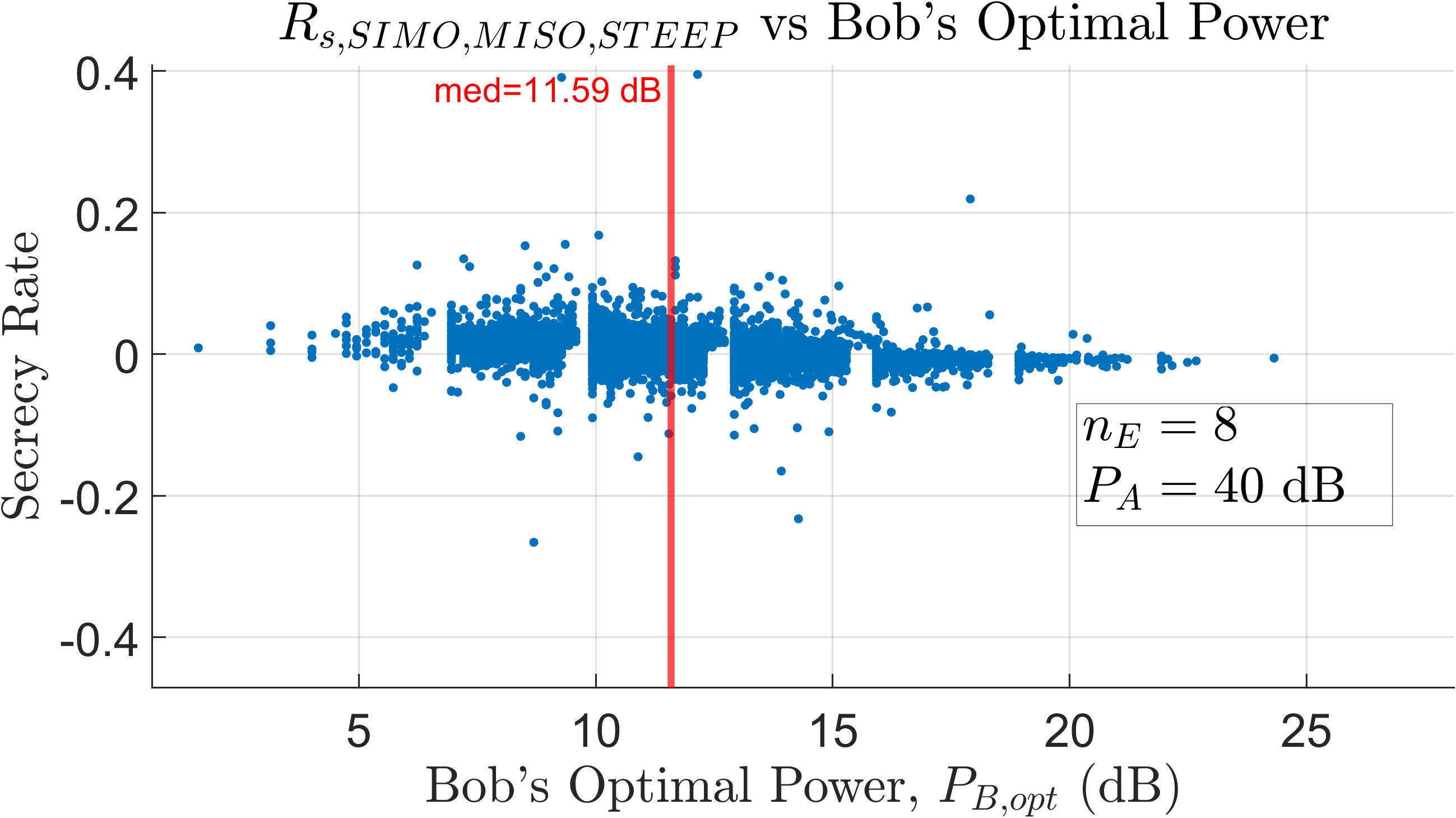}
\subcaption{$p_A=40$dB. }
\label{fig:secrecy_steep_pB_40}
\end{minipage}\\
\\
\begin{minipage}[b]{0.48\linewidth}
\centering
\includegraphics[width=\textwidth]{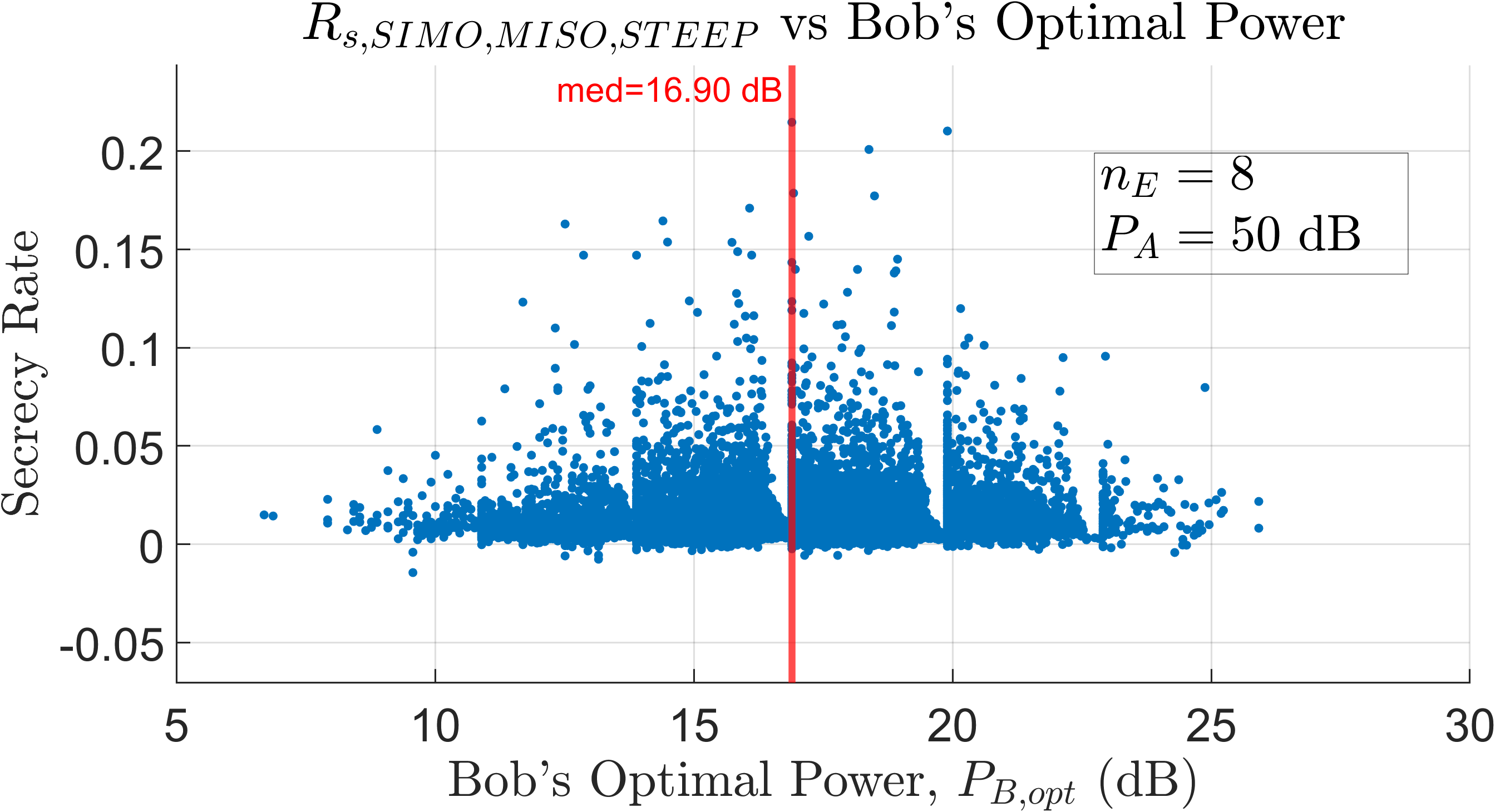}
\subcaption{$p_A=50$dB,}
\label{fig:secrecy_steep_pB_30_rho_0}
\end{minipage}
\hspace{0.1cm}
\begin{minipage}[b]{0.48\linewidth}
\centering
\includegraphics[width=\textwidth]{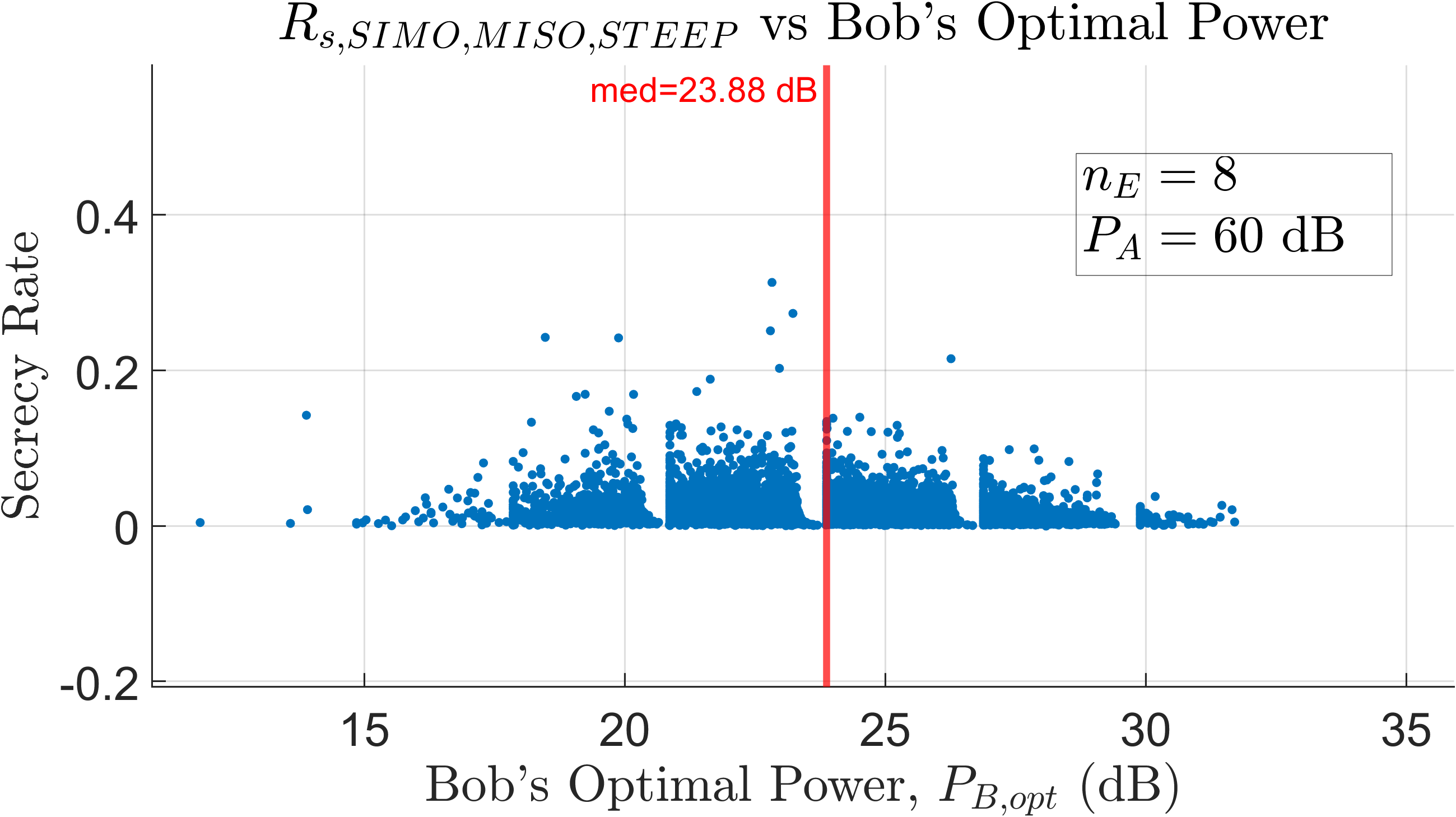}
\subcaption{$p_A=60$dB. }
\label{fig:secrecy_steep_pB_40_rho_0}
\end{minipage}
\caption{Distributions of  $(\hat R_{s,\texttt{S-M-STEEP}},p_{B,opt})$ with $p_E=10$dB, $n_A=4$, $n_E=8$.}
\label{fig:secrecy_rate_vs_optimized_p_B_simo_miso_steep_n_E_8}
\end{figure}

Similar to Fig. \ref{fig:secrecy_rate_vs_optimized_p_B_simo_miso_steep_n_E_4} for $n_E=4$, Figs. \ref{fig:secrecy_rate_vs_optimized_p_B_simo_miso_steep_n_E_6} and \ref{fig:secrecy_rate_vs_optimized_p_B_simo_miso_steep_n_E_8} are shown for $n_E=6$ and $n_E=8$ respectively. Because of the increased values of $n_E$, the optimal secrecy rate $\hat R_{s,\texttt{S-M-STEEP}}$ of SIMO-MISO-STEEP is mostly non-positive when its phase-2 power $p_A$ is not large enough. But when $p_A=60$dB, we see that for both $n_E=6$ and $n_E=8$, $\hat R_{s,\texttt{S-M-STEEP}}$ stays above zero with near probability one.

We also see that with a large $n_E$  the optimal phase-1 power $p_{B,opt}$ of SIMO-MISO-STEEP is rarely equal to the given maximum $p_A$. This is because with a large $n_E$ Eve's receive channel in phase 1 is rarely weaker than user's.

\begin{figure}[ht]
\begin{minipage}[b]{0.48\linewidth}
\centering
\includegraphics[width=\textwidth]{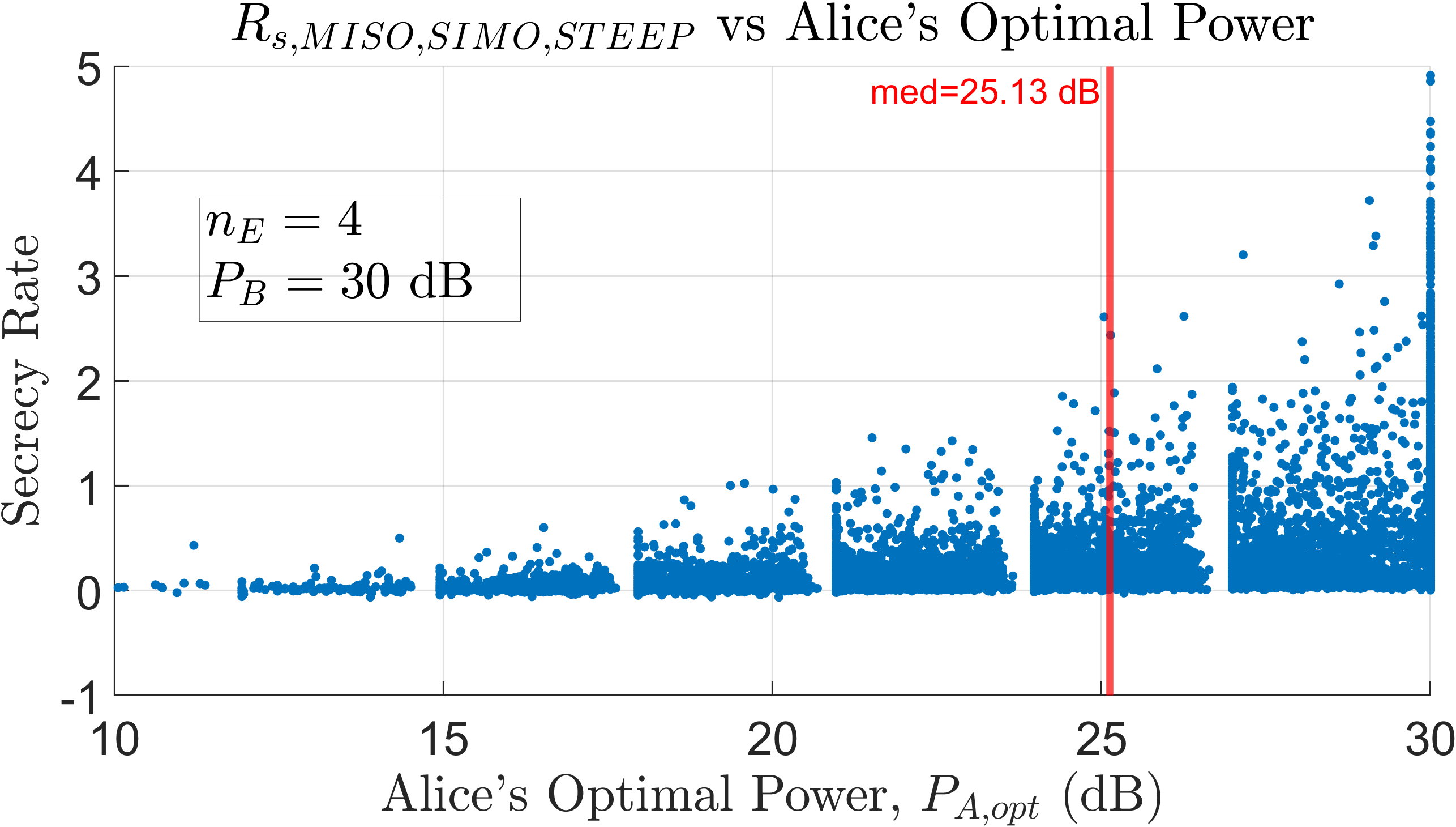}
\subcaption{$p_B=30$dB. }
\label{fig:secrecy_steep_pB_30}
\end{minipage}
\hspace{0.1cm}
\begin{minipage}[b]{0.48\linewidth}
\centering
\includegraphics[width=\textwidth]{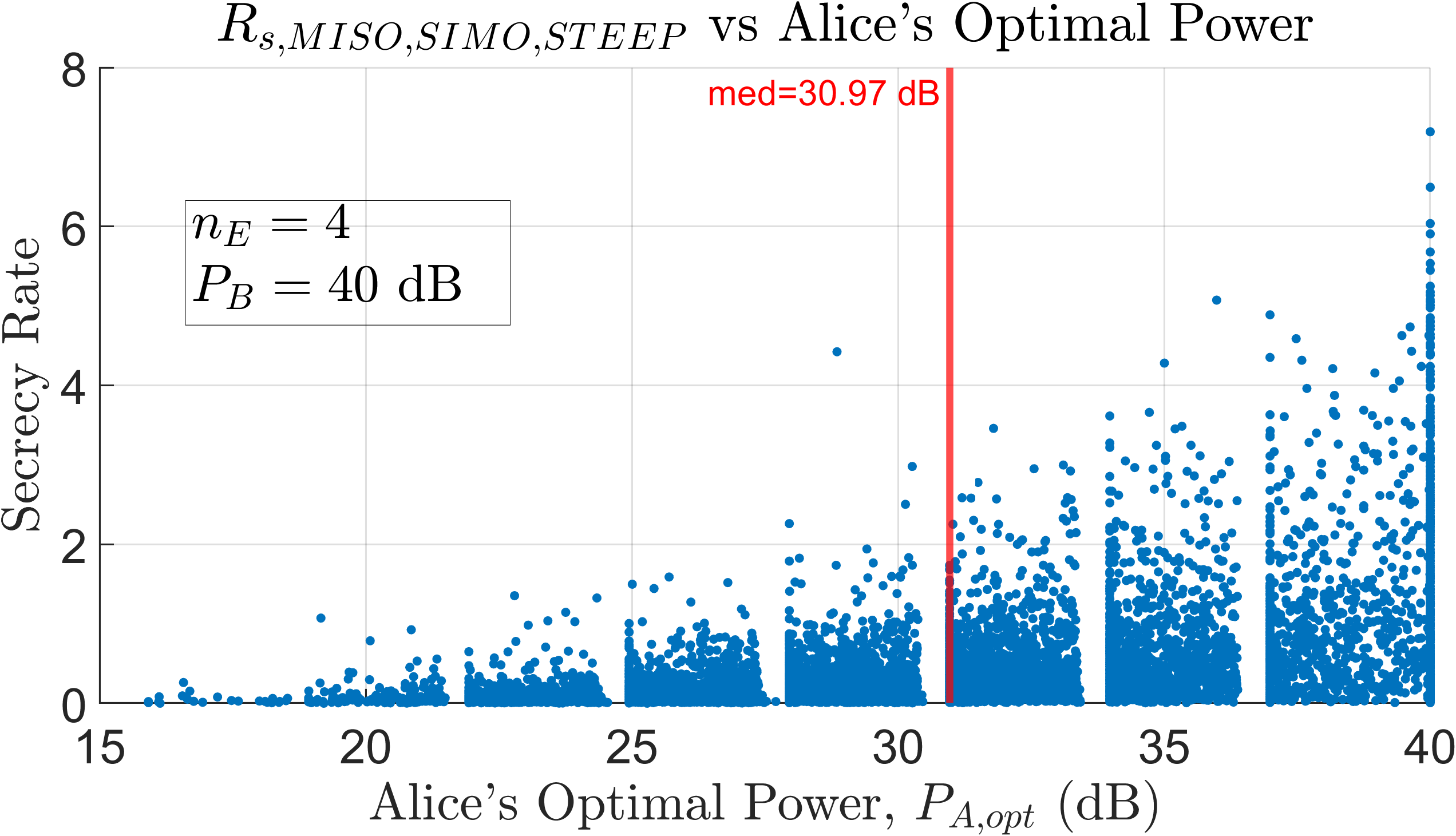}
\subcaption{$p_B=40$dB. }
\label{fig:secrecy_steep_pB_40}
\end{minipage}\\
\\
\begin{minipage}[b]{0.48\linewidth}
\centering
\includegraphics[width=\textwidth]{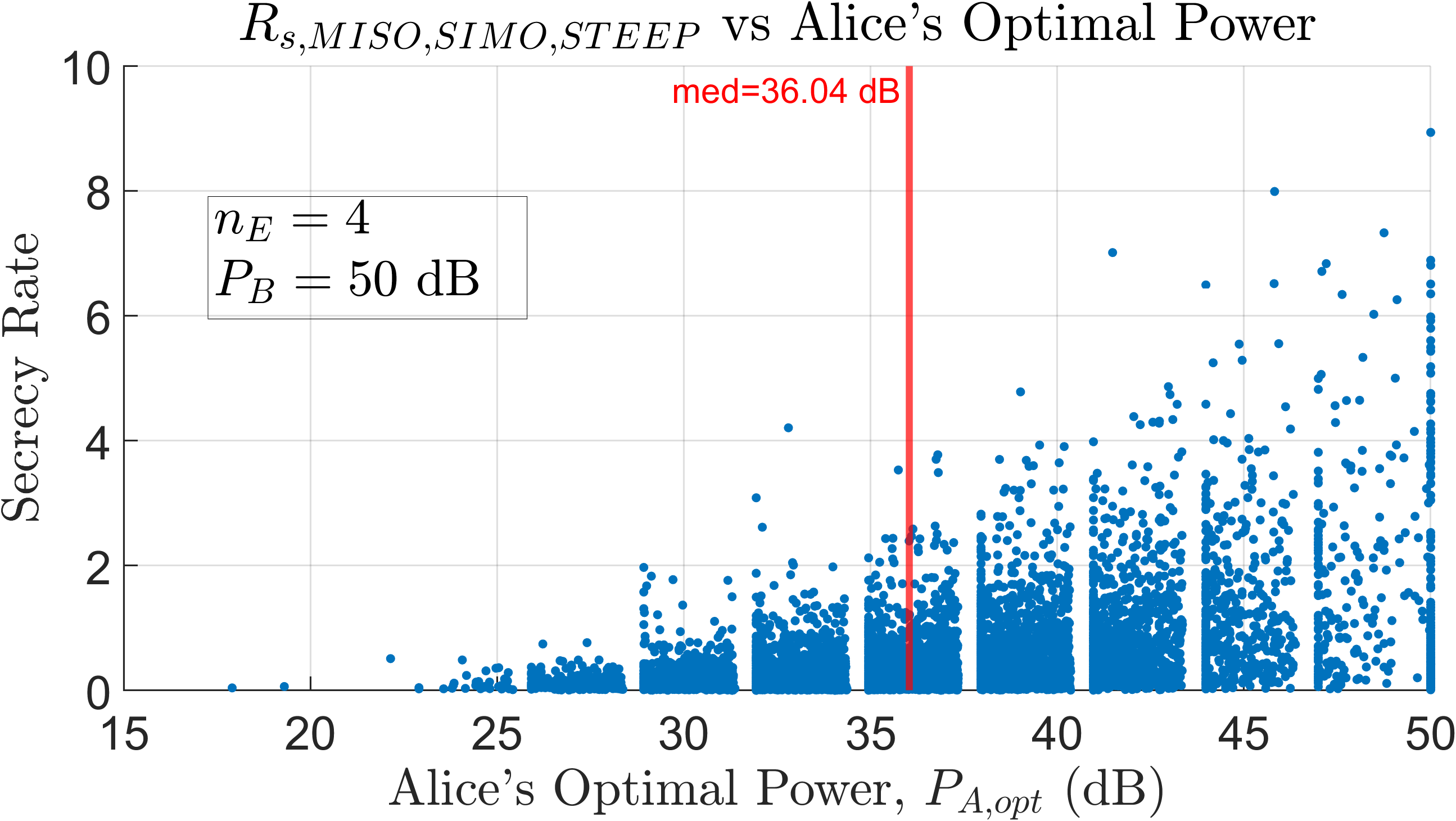}
\subcaption{$p_B=50$dB,}
\label{fig:secrecy_steep_pB_30_rho_0}
\end{minipage}
\hspace{0.1cm}
\begin{minipage}[b]{0.48\linewidth}
\centering
\includegraphics[width=\textwidth]{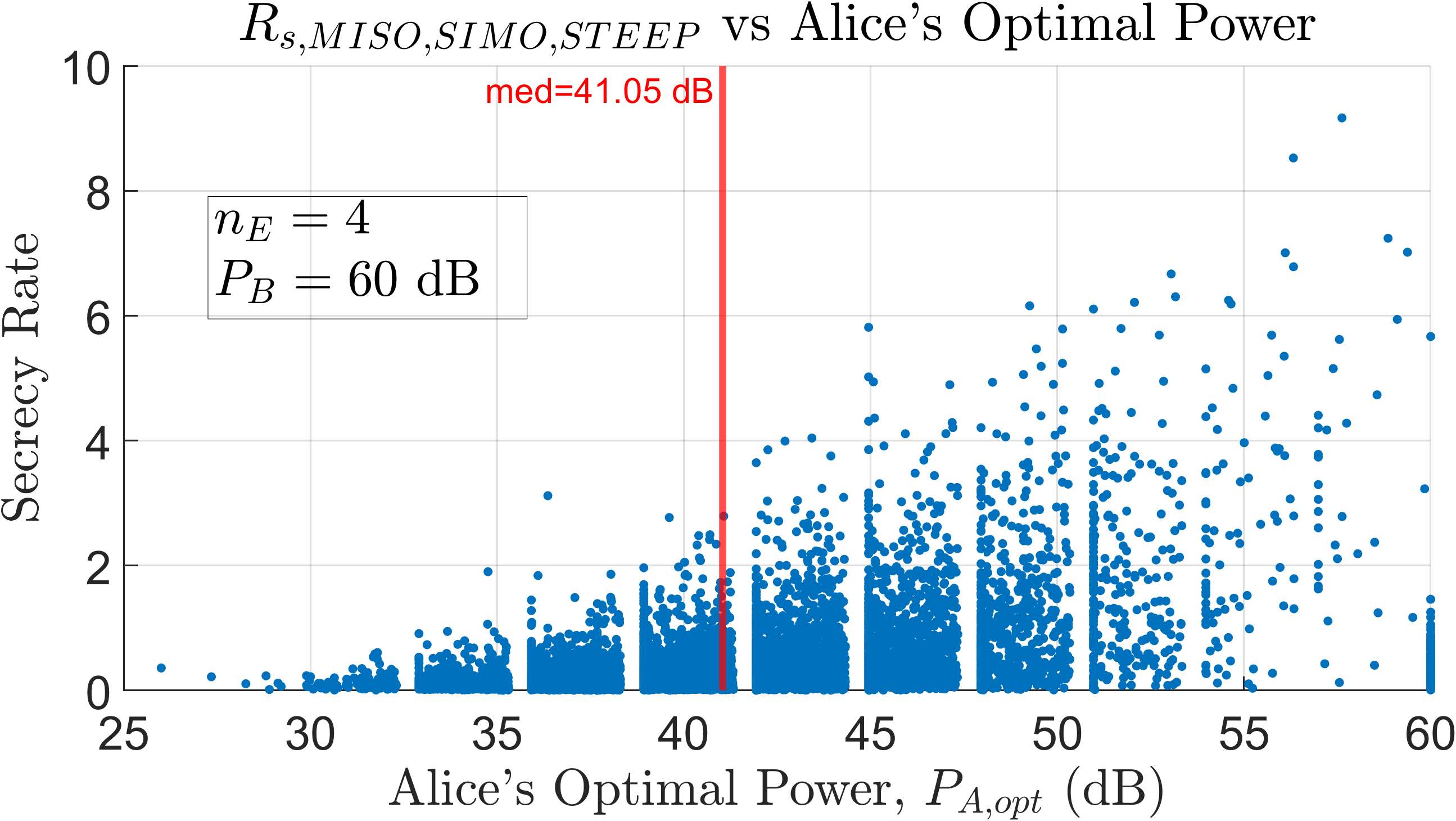}
\subcaption{$p_B=60$dB. }
\label{fig:secrecy_steep_pB_40_rho_0}
\end{minipage}
\caption{Distributions of $(\hat R_{s,\texttt{M-S-STEEP}},p_{A,opt})$ with $p_E=10$dB, $n_A=4$, $n_E=4$.}
\label{fig:secrecy_rate_vs_optimized_p_a_miso_simo_steep_n_E_4}
\end{figure}

\begin{figure}[ht]
\begin{minipage}[b]{0.48\linewidth}
\centering
\includegraphics[width=\textwidth]{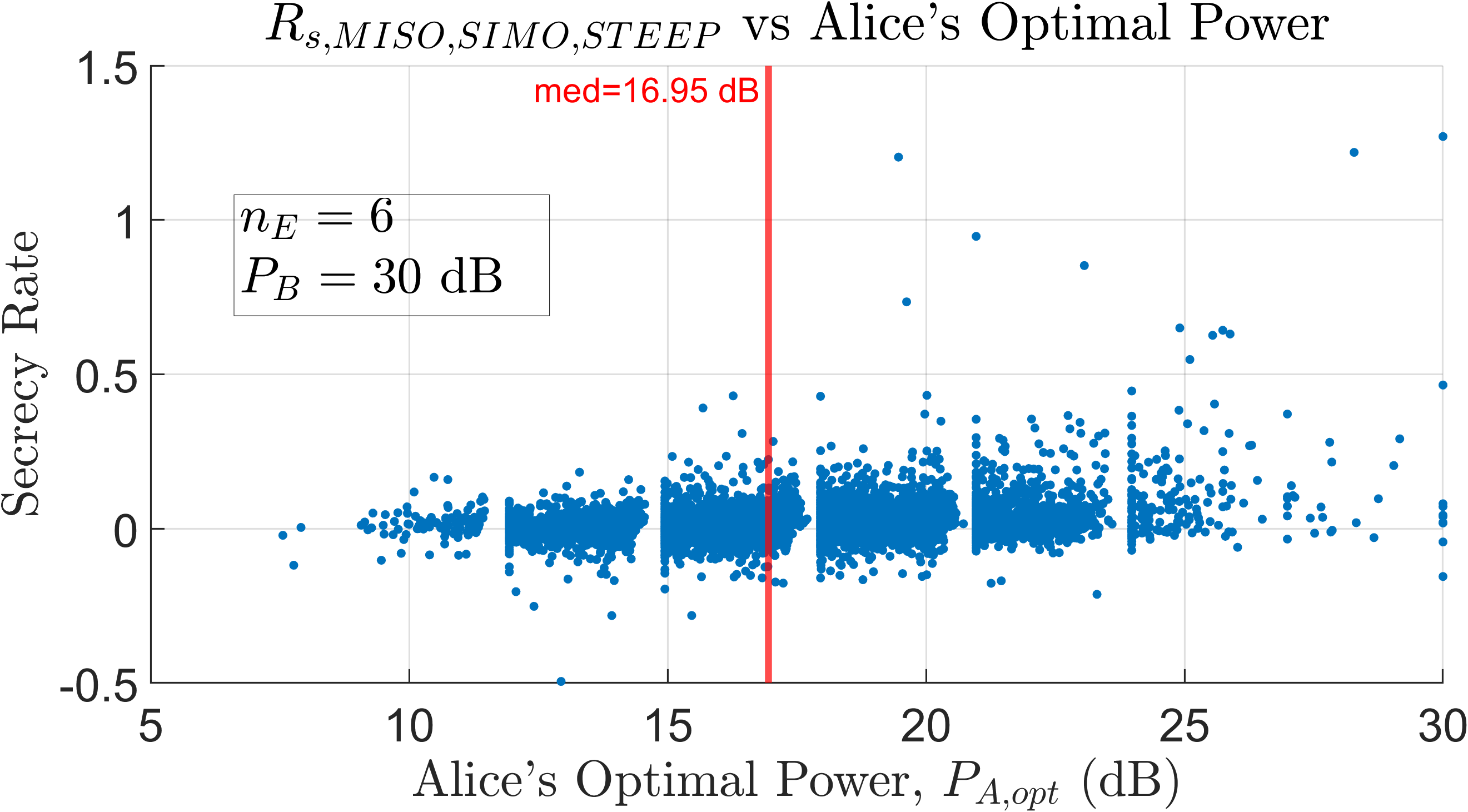}
\subcaption{$p_B=30$dB. }
\label{fig:secrecy_steep_pB_30}
\end{minipage}
\hspace{0.1cm}
\begin{minipage}[b]{0.48\linewidth}
\centering
\includegraphics[width=\textwidth]{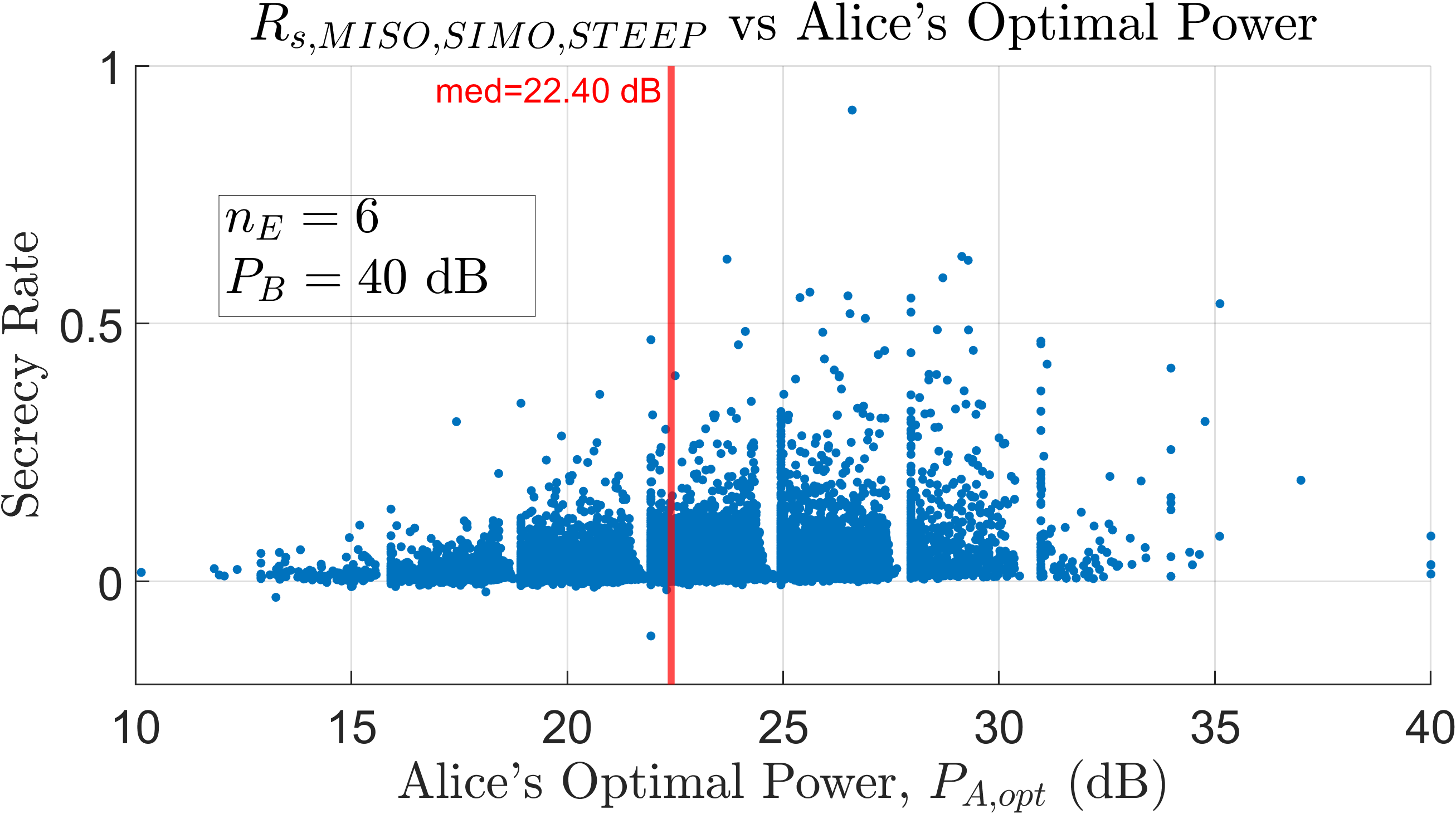}
\subcaption{$p_B=40$dB. }
\label{fig:secrecy_steep_pB_40}
\end{minipage}\\
\\
\begin{minipage}[b]{0.48\linewidth}
\centering
\includegraphics[width=\textwidth]{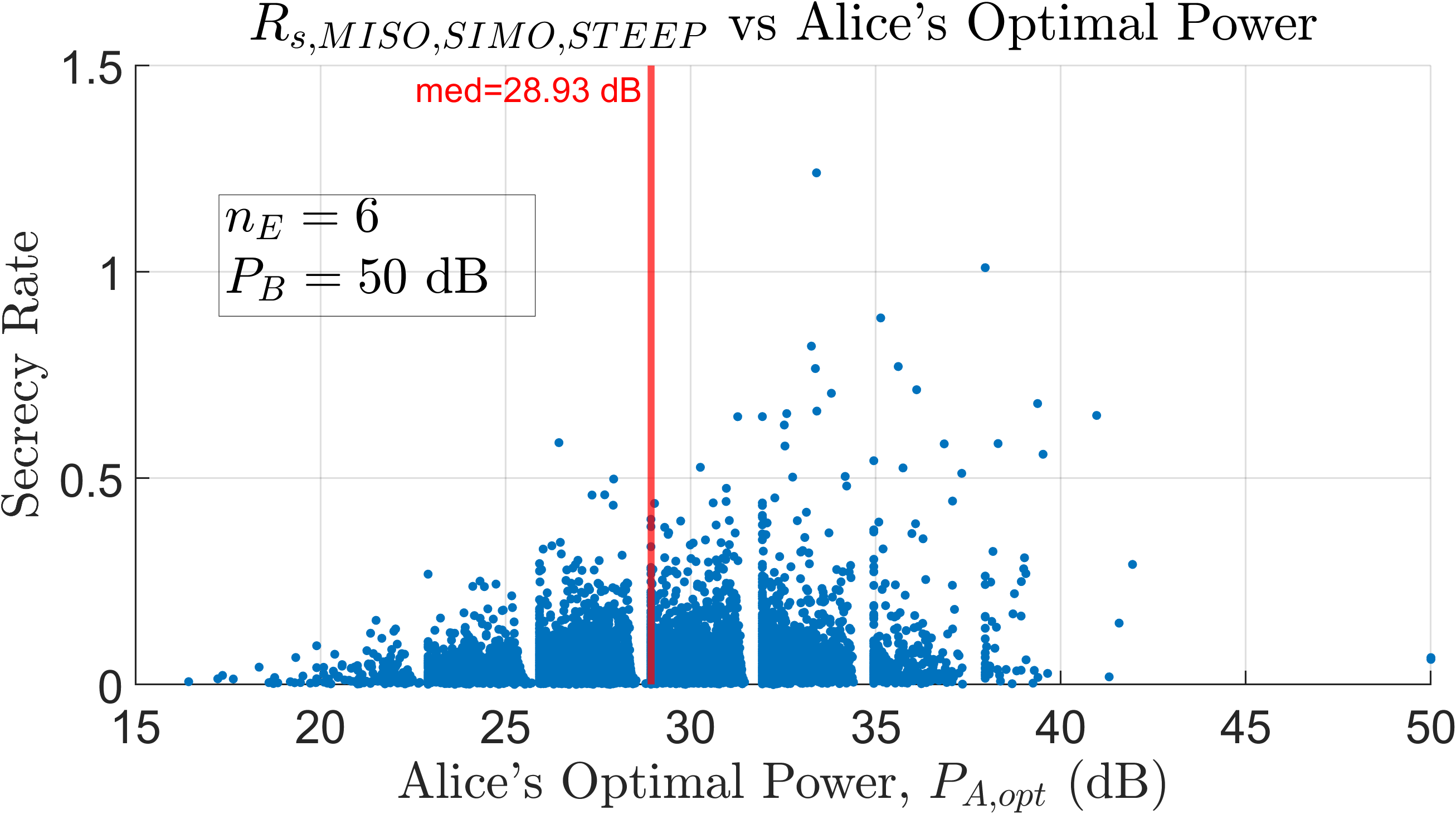}
\subcaption{$p_B=50$dB,}
\label{fig:secrecy_steep_pB_30_rho_0}
\end{minipage}
\hspace{0.1cm}
\begin{minipage}[b]{0.48\linewidth}
\centering
\includegraphics[width=\textwidth]{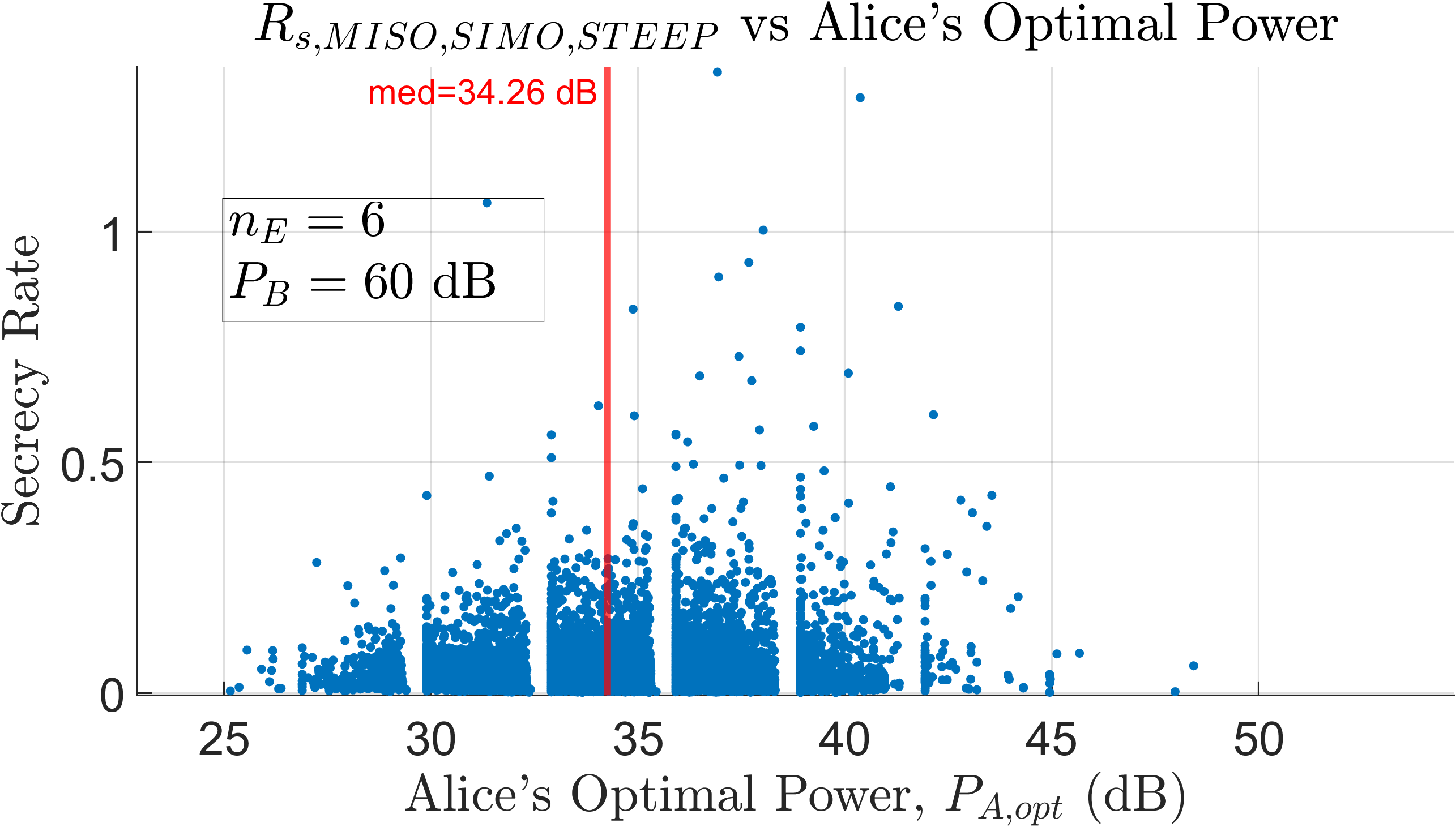}
\subcaption{$p_B=60$dB. }
\label{fig:secrecy_steep_pB_40_rho_0}
\end{minipage}
\caption{Distributions of  $(\hat R_{s,\texttt{M-S-STEEP}},p_{A,opt})$ with $p_E=10$dB, $n_A=4$, $n_E=6$.}
\label{fig:secrecy_rate_vs_optimized_p_a_miso_simo_steep_n_E_6}
\end{figure}

\begin{figure}[ht]
\begin{minipage}[b]{0.48\linewidth}
\centering
\includegraphics[width=\textwidth]{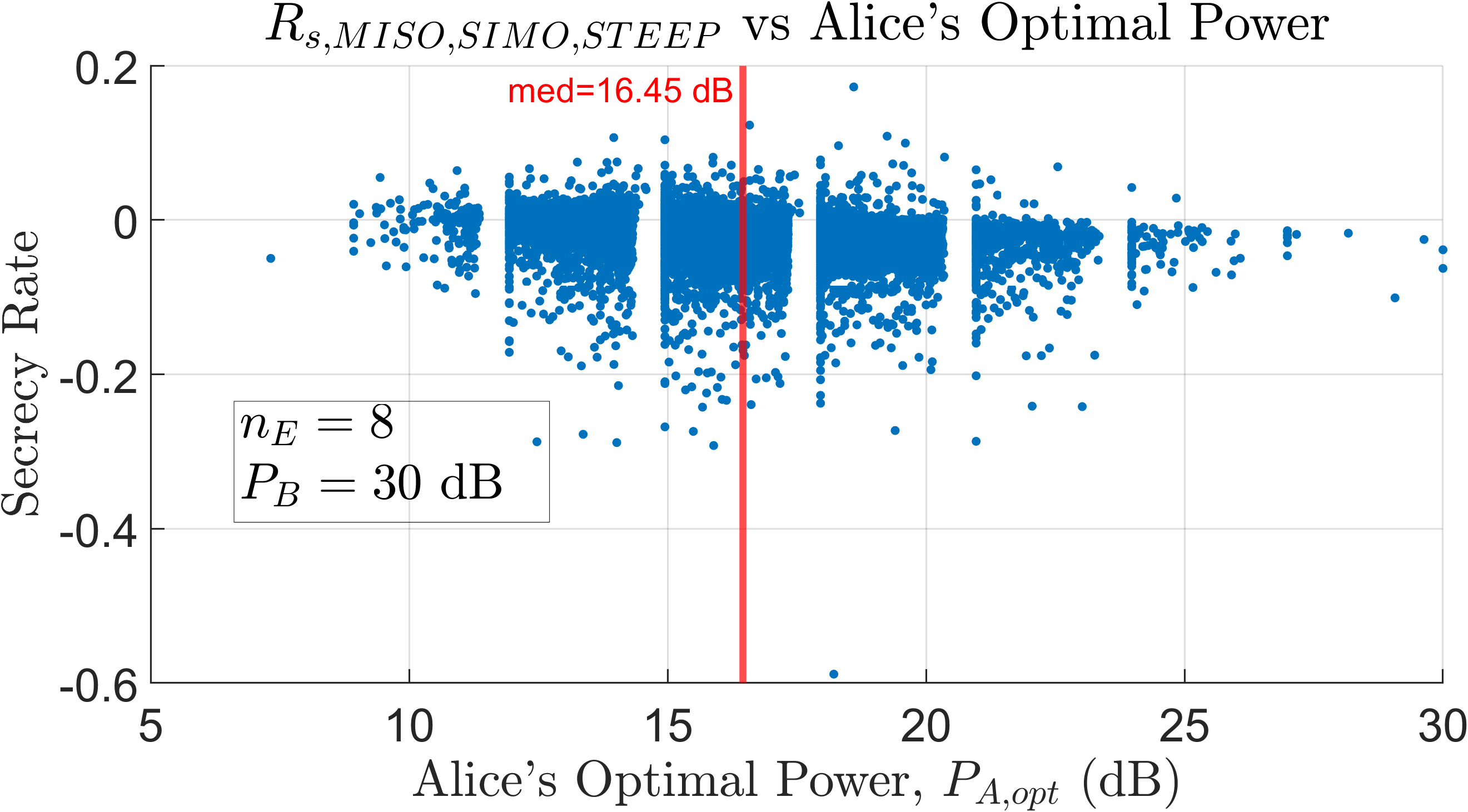}
\subcaption{$p_B=30$dB. }
\label{fig:secrecy_steep_pB_30}
\end{minipage}
\hspace{0.1cm}
\begin{minipage}[b]{0.48\linewidth}
\centering
\includegraphics[width=\textwidth]{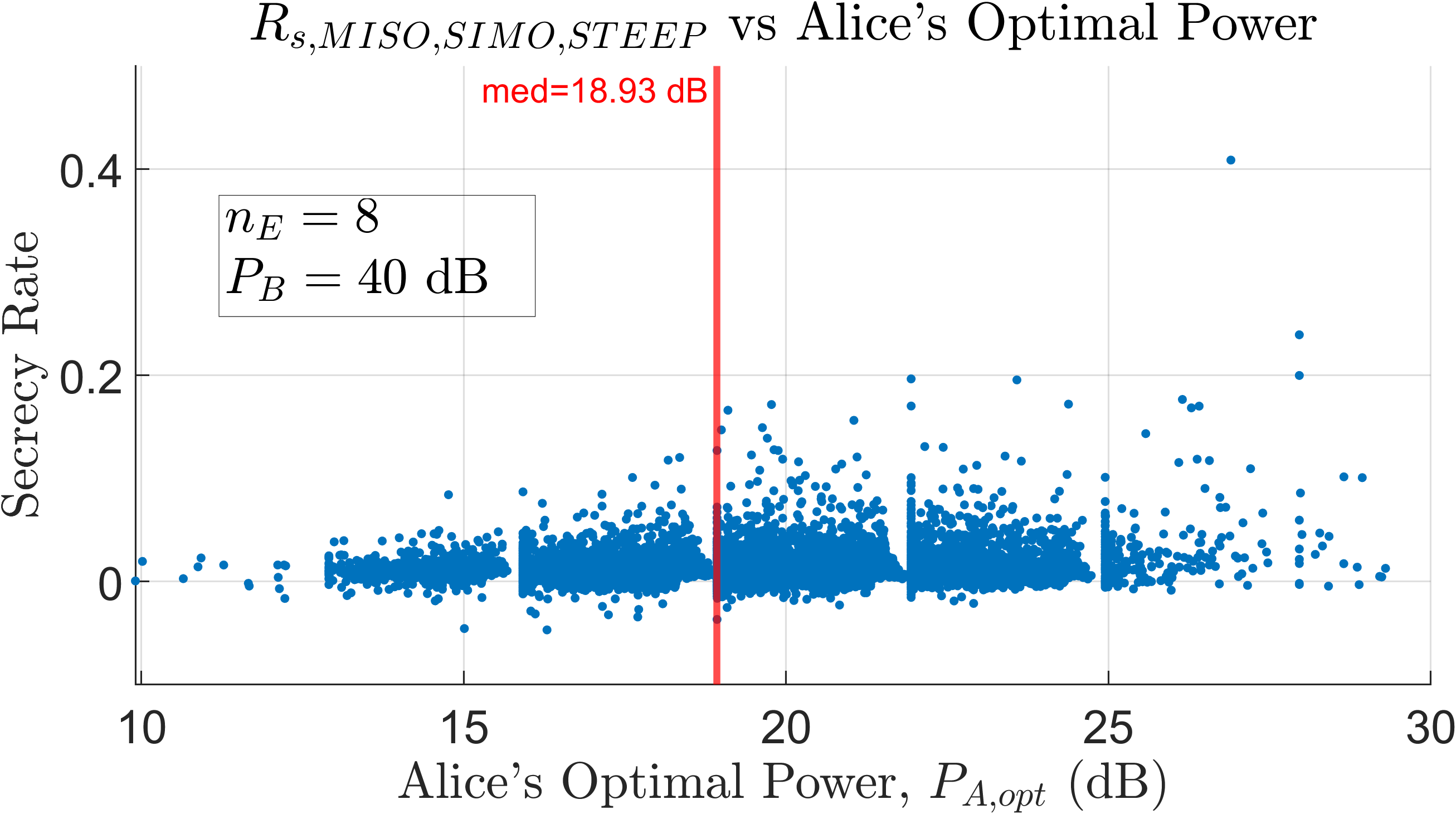}
\subcaption{$p_B=40$dB. }
\label{fig:secrecy_steep_pB_40}
\end{minipage}\\
\\
\begin{minipage}[b]{0.48\linewidth}
\centering
\includegraphics[width=\textwidth]{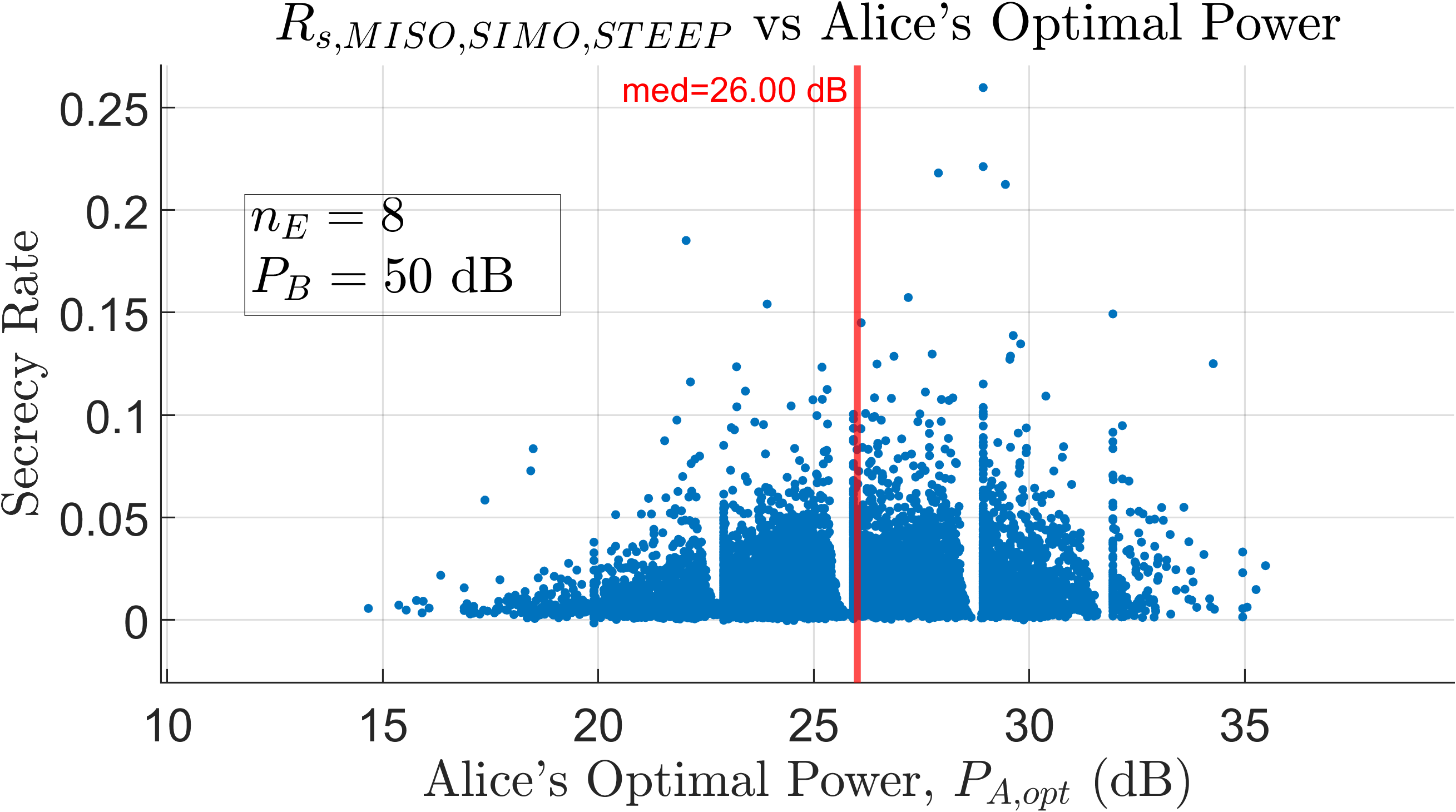}
\subcaption{$p_B=50$dB,}
\label{fig:secrecy_steep_pB_30_rho_0}
\end{minipage}
\hspace{0.1cm}
\begin{minipage}[b]{0.48\linewidth}
\centering
\includegraphics[width=\textwidth]{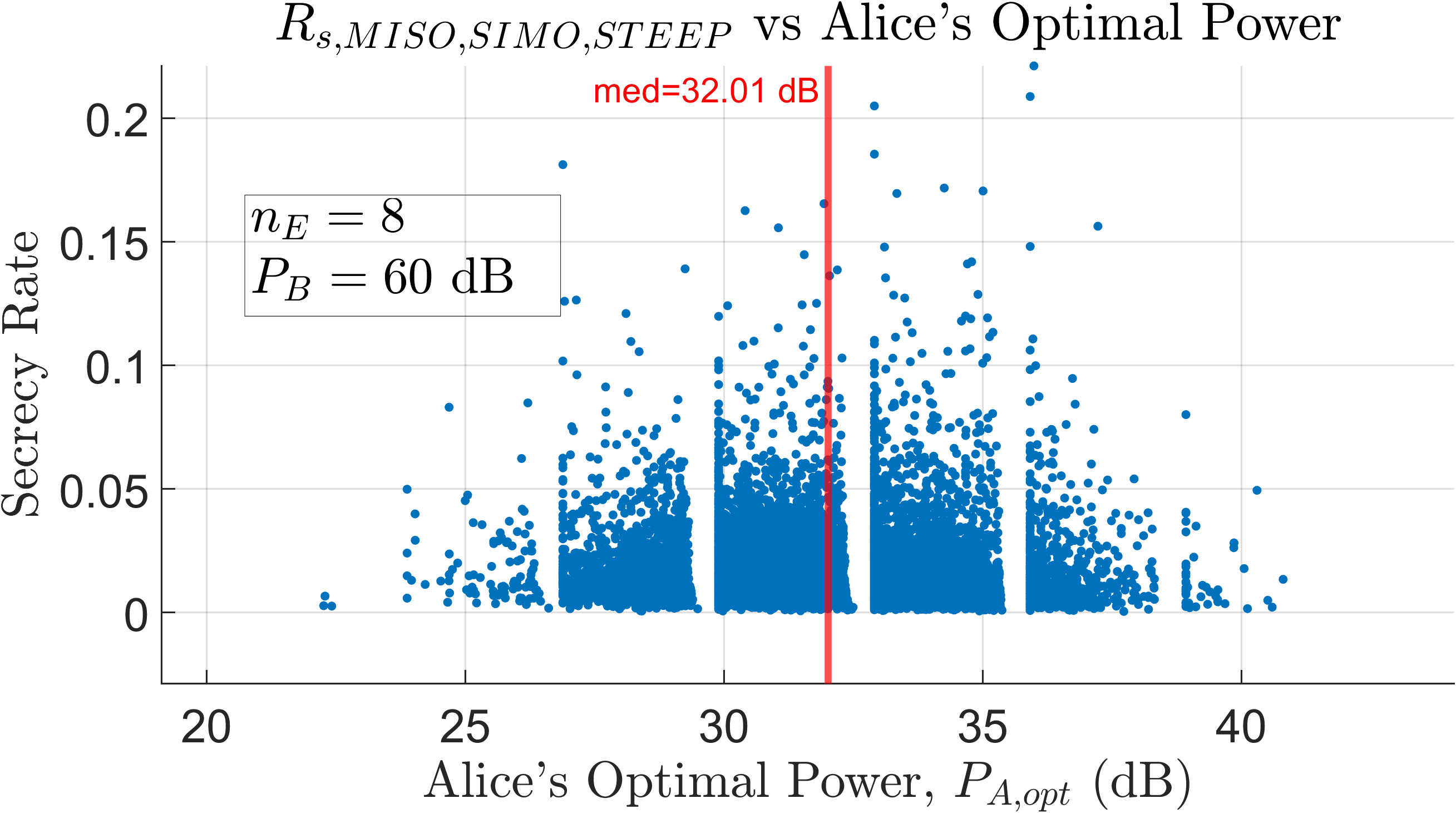}
\subcaption{$p_B=60$dB. }
\label{fig:secrecy_steep_pB_40_rho_0}
\end{minipage}
\caption{Distributions of $(\hat R_{s,\texttt{M-S-STEEP}},p_{A,opt})$ with $p_E=10$dB, $n_A=4$, $n_E=8$.}
\label{fig:secrecy_rate_vs_optimized_p_a_miso_simo_steep_n_E_8}
\end{figure}

Shown in Figs. \ref{fig:secrecy_rate_vs_optimized_p_a_miso_simo_steep_n_E_4}, \ref{fig:secrecy_rate_vs_optimized_p_a_miso_simo_steep_n_E_6} and \ref{fig:secrecy_rate_vs_optimized_p_a_miso_simo_steep_n_E_8} are the distributions of $(\hat R_{s,\texttt{M-S-STEEP}},p_{A,opt})$ for MISO-SIMO-STEEP. Most of the features one can see for MISO-SIMO-STEEP from Figs. \ref{fig:secrecy_rate_vs_optimized_p_a_miso_simo_steep_n_E_4}, \ref{fig:secrecy_rate_vs_optimized_p_a_miso_simo_steep_n_E_6} and \ref{fig:secrecy_rate_vs_optimized_p_a_miso_simo_steep_n_E_8} are similar to those for SIMO-MISO-STEEP from Figs. \ref{fig:secrecy_rate_vs_optimized_p_B_simo_miso_steep_n_E_4}, \ref{fig:secrecy_rate_vs_optimized_p_B_simo_miso_steep_n_E_6} and \ref{fig:secrecy_rate_vs_optimized_p_B_simo_miso_steep_n_E_8}. These figures should be now self-explanatory.

In practice, it would be too costly to conduct a bisection search to optimize the phase-1 power for every realization of channel parameters. Also, Eve's receive channels are often hidden from users. But a feasible choice of the phase-1 power in practice is the median of the optimal phase-1 power values computed offline based on a statistical model. This is what we adopt next in comparing the secrecy rates of SIMO-MISO-STEEP and MISO-SIMO-STEEP with those of the conventional schemes.

\subsection{Comparison with the conventional schemes}
We now present numerical results to compare SIMO-MISO-STEEP and MISO-SIMO-STEEP with the conventional SIMO and MISO schemes. Taking into account that STEEP is a round-trip scheme where both nodes need to consume some amount of transmit power, we  define
\begin{align}\label{eq:Delta_S_M_STEEP}
    &\Delta R_{s,\texttt{SIMO,MISO}}\doteq R_{s,\texttt{S-M-STEEP}}^{+} -  R_{s,\texttt{SIMO,conv}}^{+}\notag\\
     &\,\,- R_{s,\texttt{MISO,conv}}^{+},
\end{align}
\begin{align}\label{eq:Delta_M_S_STEEP}
    &\Delta R_{s,\texttt{MISO,SIMO}}\doteq R_{s,\texttt{M-S-STEEP}}^{+} -  R_{s,\texttt{SIMO,conv}}^{+} \notag\\
    &\,\,- R_{s,\texttt{MISO,conv}}^{+},
\end{align}
where every term in \eqref{eq:Delta_S_M_STEEP} is computed using the same parameters including the same realization of all channels, and the same holds for \eqref{eq:Delta_M_S_STEEP}. Here $\Delta R_{s,\texttt{SIMO,MISO}}$ measures a secrecy-rate advantage of SIMO-MISO-STEEP over the conventional SIMO and MISO schems, i.e., the difference between the secrecy rate of SIMO-MISO-STEEP and the sum secrecy rate of the conventional SIMO and MISO schemes; and $\Delta R_{s,\texttt{MISO,SIMO}}$ measures a secrecy-rate advantage of MISO-SIMO-STEEP over the conventional SIMO and MISO schemes.

Note that the phase-1 power used for SIMO-MISO-STEEP is now chosen to be $p_B=\bar p_{B,opt}$ which is the median of the optimal phase-1 powers for SIMO-MISO-STEEP as obtained from Figs. \ref{fig:secrecy_rate_vs_optimized_p_B_simo_miso_steep_n_E_4}, \ref{fig:secrecy_rate_vs_optimized_p_B_simo_miso_steep_n_E_6} and \ref{fig:secrecy_rate_vs_optimized_p_B_simo_miso_steep_n_E_8}. For fairness, the power used by Alice for $R_{s,\texttt{MISO,conv}}^{+}$ in \eqref{eq:Delta_S_M_STEEP} is now chosen to be $p_A+p_A\gamma_A\sigma_{\hat r_{B|A}}^2<(1+\gamma_A)p_A$.

Similarly, the phase-1 power used for MISO-SIMO-STEEP is now chosen to be $p_A=\bar p_{A,opt}$, which is the median of the optimal phase-1 powers for MISO-SIMO-STEEP as shown in Figs. \ref{fig:secrecy_rate_vs_optimized_p_a_miso_simo_steep_n_E_4}, \ref{fig:secrecy_rate_vs_optimized_p_a_miso_simo_steep_n_E_6} and \ref{fig:secrecy_rate_vs_optimized_p_a_miso_simo_steep_n_E_8}. And the power used by Bob for $R_{s,\texttt{SIMO,conv}}^{+}$ in \eqref{eq:Delta_M_S_STEEP} is now chosen to be $p_B+p_B\sigma_{\hat r_{A|B}}^2<2p_B$.

Also note that we have found that the secrecy rate of either SIMO-MISO-STEEP or MISO-SIMO-STEEP as a function of its phase-1 power has a broad peak, which is not very sensitive to a deviation from the above choices of phase-1 powers. The good choices of phase-1 powers are typically much smaller than the phase-2 power.

\begin{figure}[ht]
\begin{minipage}[b]{0.48\linewidth}
\centering
\includegraphics[width=4.5cm,height=4.5cm]{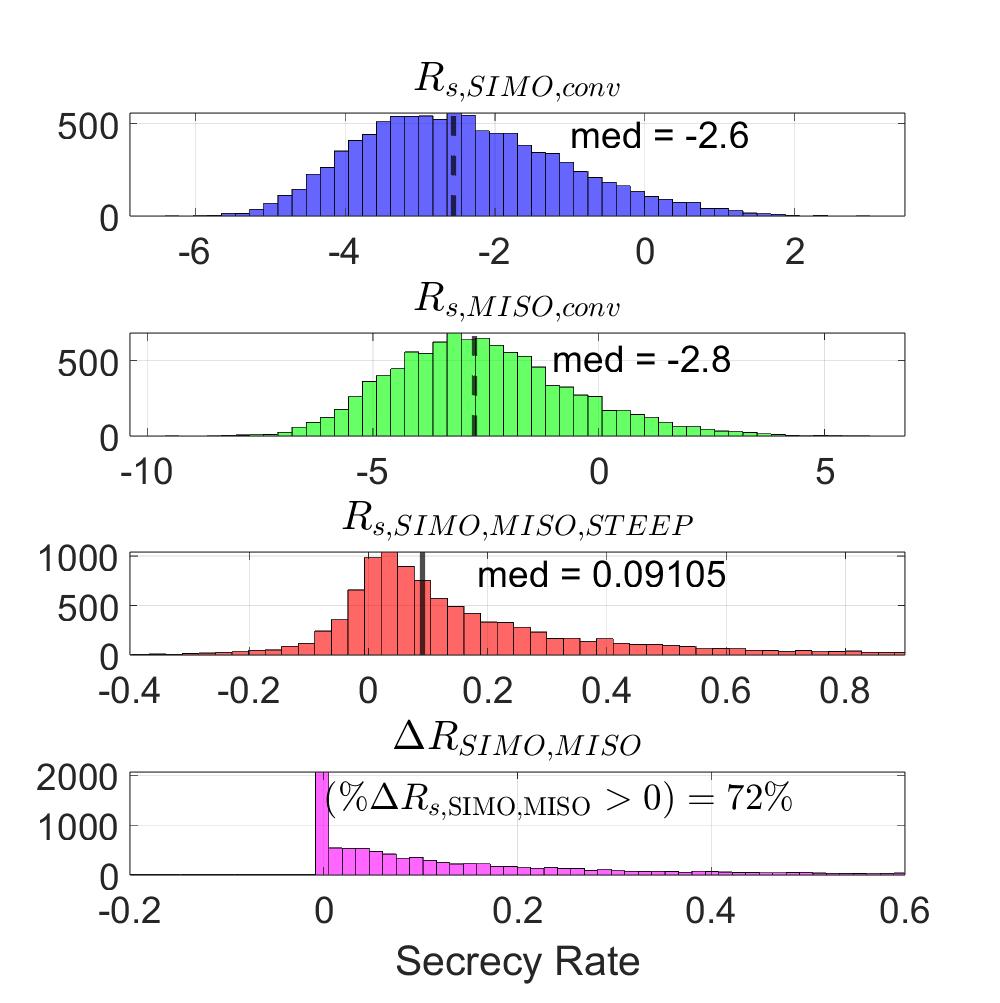}
\subcaption{$p_A=30$, $p_B=\bar p_{B,opt}=8.2$dB. }
\label{fig:secrecy_steep_pB_30}
\end{minipage}
\hspace{0.01cm}
\begin{minipage}[b]{0.48\linewidth}
\centering
\includegraphics[width=4.5cm,height=4.5cm]{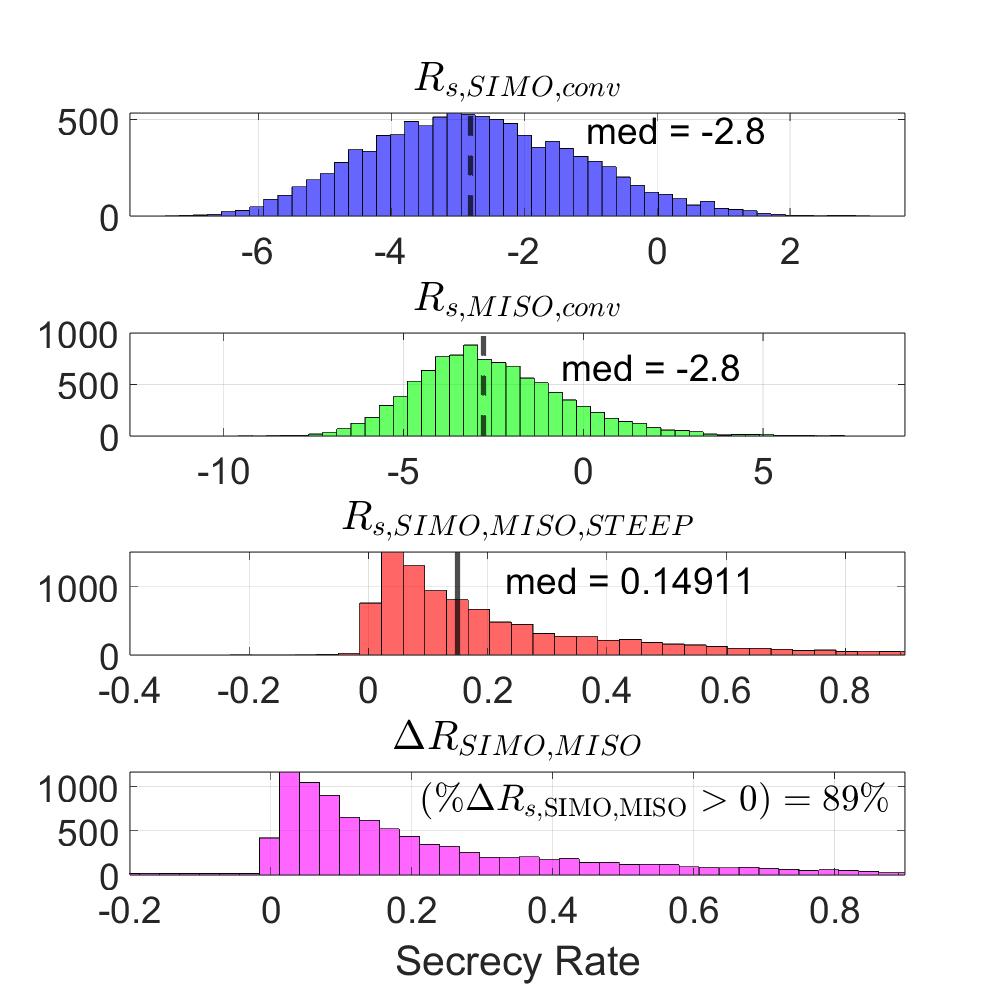}
\subcaption{$p_A=40$, $p_B=\bar p_{B,opt}=14.2$dB. }
\label{fig:secrecy_steep_pB_40}
\end{minipage}\\
\caption{Comparison of SIMO-MISO-STEEP with the conventional using $p_E=10$dB, $n_A=4$, $n_E=4$.}
\label{fig:comp_optimized_simo_miso_steep_n_E_4}
\end{figure}

Shown in Fig. \ref{fig:comp_optimized_simo_miso_steep_n_E_4} are the histograms of the individual terms in \eqref{eq:Delta_S_M_STEEP} with $p_E=10$dB, $n_A=4$ and $n_E=4$. And two options of $p_A$, i.e., 30dB and 40dB, are shown. We see that the secrecy-rate advantage of SIMO-MISO-STEEP is mostly positive while the secrecy rates of the conventional SIMO and MISO schemes are mostly non-positive. The median values of the secrecy rates are shown for the upper three rows of plots. For the bottom row of plots, the percentages of positive secrecy-rate advantages of SIMO-MISO-STEEP are also shown.

\begin{figure}[ht]
\begin{minipage}[b]{0.48\linewidth}
\centering
\includegraphics[width=4.5cm,height=4.5cm]{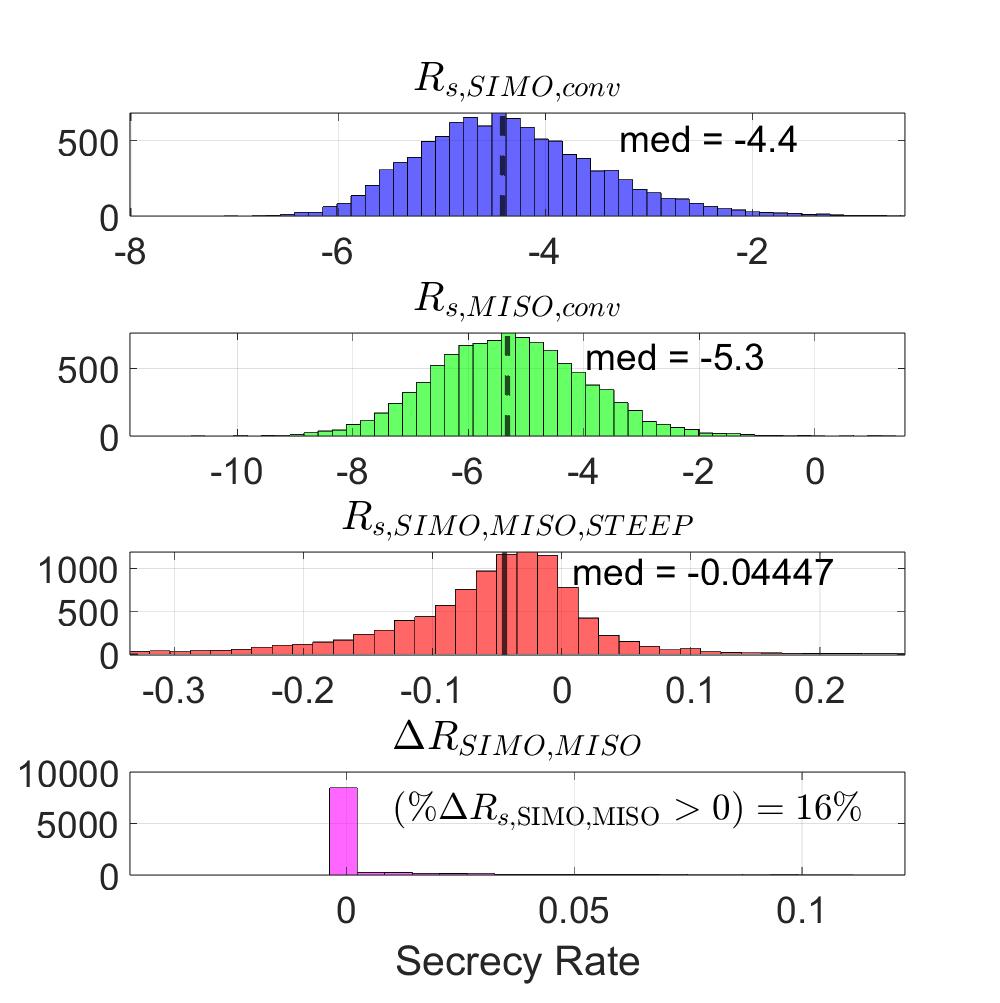}
\subcaption{$p_A=30$dB, $p_B=\bar p_{B,opt}=9$dB. }
\label{fig:secrecy_steep_pB_30}
\end{minipage}
\hspace{0.01cm}
\begin{minipage}[b]{0.48\linewidth}
\centering
\includegraphics[width=4.5cm,height=4.5cm]{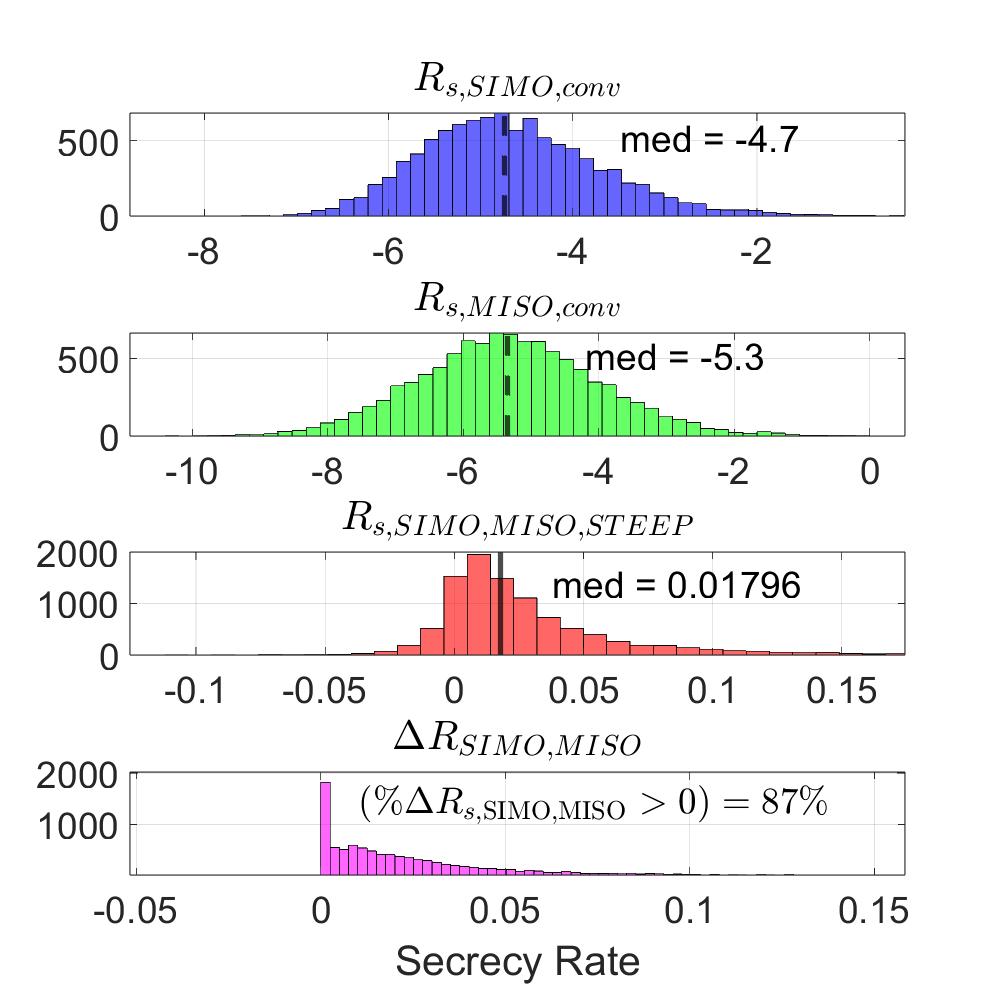}
\subcaption{$p_A=40$dB, $p_B=\bar p_{B,opt}=11.3$dB. }
\label{fig:secrecy_steep_pB_40}
\end{minipage}\\
\caption{Comparison of SIMO-MISO-STEEP with the conventional using $p_E=10$dB, $n_A=4$, $n_E=6$.}
\label{fig:comp_optimized_simo_miso_steep_n_E_6}
\end{figure}

\begin{figure}[ht]
\begin{minipage}[b]{0.48\linewidth}
\centering
\includegraphics[width=4.5cm,height=4.5cm]{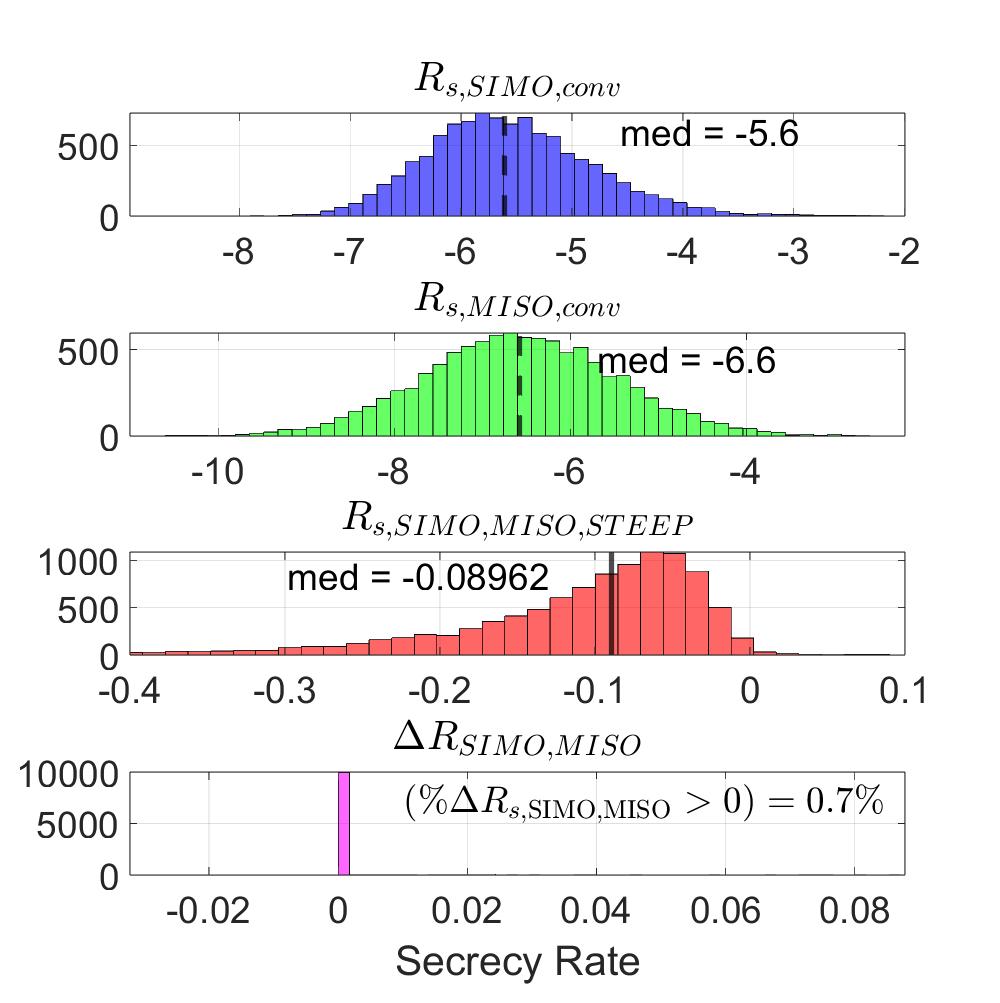}
\subcaption{$p_A=30$dB, $p_B=\bar p_{B,opt}=11.1$dB. }
\label{fig:secrecy_steep_pB_30}
\end{minipage}
\hspace{0.01cm}
\begin{minipage}[b]{0.48\linewidth}
\centering
\includegraphics[width=4.5cm,height=4.5cm]{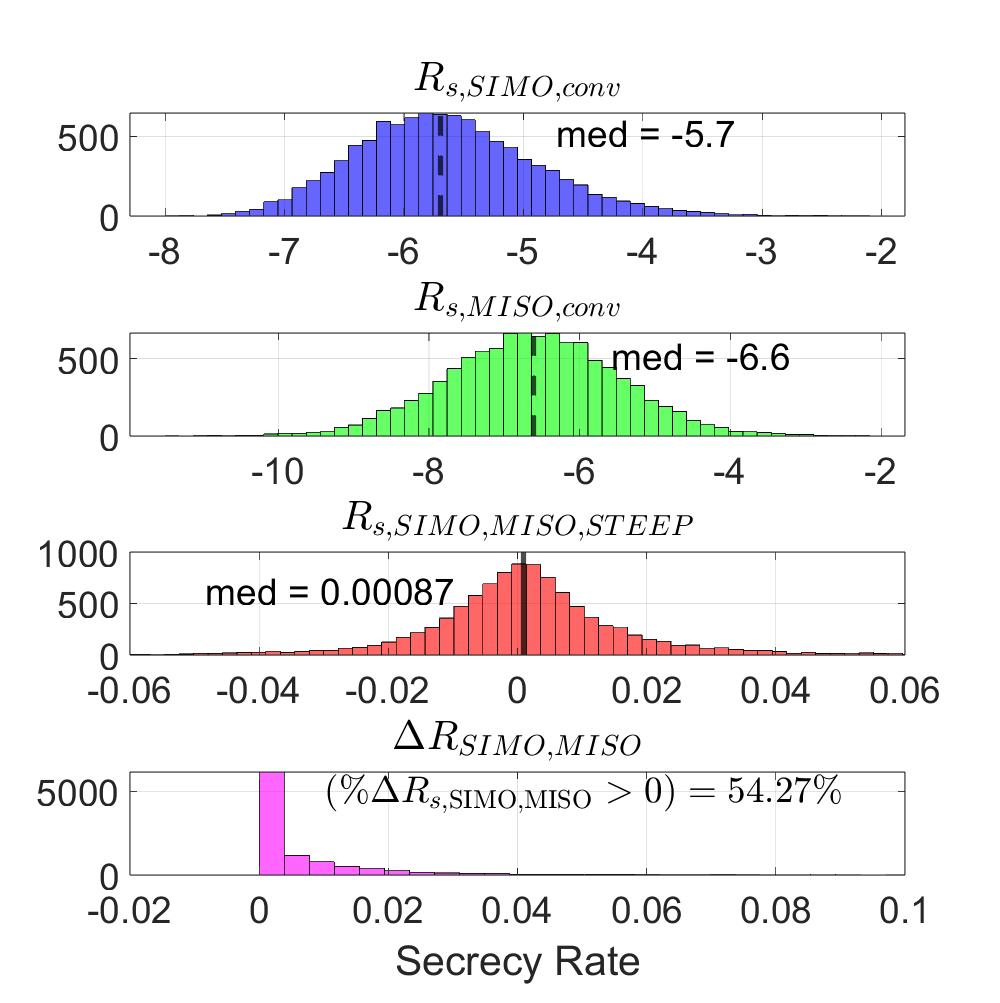}
\subcaption{$p_A=40$dB, $p_B=\bar p_{B,opt}=11.6$dB. }
\label{fig:secrecy_steep_pB_40}
\end{minipage}\\
\caption{Comparison of SIMO-MISO-STEEP with the conventional using $p_E=10$dB, $n_A=4$, $n_E=8$.}
\label{fig:comp_optimized_simo_miso_steep_n_E_8}
\end{figure}

Similar to Fig. \ref{fig:comp_optimized_simo_miso_steep_n_E_4}, Figs. \ref{fig:comp_optimized_simo_miso_steep_n_E_6} and \ref{fig:comp_optimized_simo_miso_steep_n_E_8} show the histograms of the individual terms in
\eqref{eq:Delta_S_M_STEEP} but using $n_E=6$ and $n_E=8$ respectively. In these cases, due to the increased $n_E$, the secrecy rates of the  conventional SIMO and MISO schemes are almost all non-positive while the secrecy rate of SIMO-MISO-STEEP is still positive for a significant probability especially for the case of $p_A=40$dB. The secrecy-rate advantage of SIMO-MISO-STEEP (i.e., $\Delta R_{s,\texttt{SIMO,MISO}}$) is consequently also positive for the same probabilities.

Note that for the bottom two plots in Fig. \ref{fig:comp_optimized_simo_miso_steep_n_E_8}, for example, the positive components of $\Delta R_{s,\texttt{SIMO,MISO}}$ are overshadowed by a large concentration of $\Delta R_{s,\texttt{SIMO,MISO}}=0$ (due to ``rounding'' of all negative values of $R_{s,\texttt{S-M-STEEP}}$ to zero). But the probability of $\Delta R_{s,\texttt{SIMO,MISO}}>0$ is indicated in each of the two plots.

Let $\texttt{M}_{\texttt{SIMO}}$, $\texttt{M}_{\texttt{MISO}}$, $\texttt{M}_{\texttt{SIMO-MISO}}$ and $\texttt{M}_{\texttt{MISO-SIMO}}$ be the medians of  $R_{\texttt{s,SIMO,conv}}$, $R_{\texttt{s,MISO,conv}}$, $R_{\texttt{s,S-M-STEEP}}$ and $R_{\texttt{s,M-S-STEEP}}$, respectively. Also let $\rho_{\texttt{SIMO-MISO}}$ be the probability that $\Delta R_{\texttt{s,SIMO,MISO}}>0$, and $\rho_{\texttt{MISO-SIMO}}$ be the probability that $\Delta R_{\texttt{s,MISO,SIMO}}>0$.

Table \ref{Table_1} summarizes the key numbers from Figs. \ref{fig:comp_optimized_simo_miso_steep_n_E_4}-\ref{fig:comp_optimized_simo_miso_steep_n_E_8}.
We see from Table \ref{Table_1} that the secrecy rates of all schemes decrease as $n_E$ increases, the (negative) median of $R_{\texttt{s,SIMO,conv}}$ decreases as $p_B$ increases, but the (negative) median of $R_{\texttt{s,MISO,conv}}$ is virtually invariant to $p_A$. Also the median of the optimal power $p_{B,opt}$ (in phase 1 here) increases with the given power $p_A$ (in phase 2 here). But the numerical pattern of $\bar p_{B,opt}$ versus $n_E$ was found to be increasing at $p_A=30$dB, decreasing at $p_A=50,60$dB, and inconsistent at $p_A=40$dB. Also see Figs \ref{fig:secrecy_rate_vs_optimized_p_B_simo_miso_steep_n_E_4}-\ref{fig:secrecy_rate_vs_optimized_p_B_simo_miso_steep_n_E_8}.

\begin{table}
  \centering
\begin{tabular}{|c|c|c|c|c|c|c|}
  \hline
  $n_E$ & 4 & 4 & 6 & 6 & 8 & 8 \\
  \hline
  $p_A$(dB) & 30 & 40 & 30 & 40 & 30 & 40 \\\hline
  $p_B=\bar p_{B,opt}$(dB) & 8.2  &14.2  & 9 & 11.3 & 11.1 &  11.6\\\hline
  $\texttt{M}_{\texttt{SIMO}}$ & -2.6 & -2.8 &-4.4  &-4.7  & -5.6 & -5.7 \\\hline
   $\texttt{M}_{\texttt{MISO}}$& -2.8 & -2.8 & -5.3 & -5.3 & -6.6 & -6.6 \\\hline
   $\texttt{M}_{\texttt{SIMO-MISO}}$& 0.09 & 0.15 & -0.04 & 0.02 & -0.09 & 0.0009 \\\hline
   $\rho_{\texttt{SIMO-MISO}}$&72\%  &89\%  & 16\% & 87\% & 0.7\% & 54\% \\
  \hline
\end{tabular}
  \caption{Comparison of key numbers from Figs \ref{fig:comp_optimized_simo_miso_steep_n_E_4}-\ref{fig:comp_optimized_simo_miso_steep_n_E_8}
  where $\bar p_{B,opt}$ is the median of the optimal $p_B$ for SIMO-MISO-STEEP.}
  \label{Table_1}
\end{table}

\begin{table}
  \centering
 \begin{tabular}{|c|c|c|c|c|c|c|}
   \hline
   $n_E$ & 4 & 4 & 6 & 6 & 8 & 8 \\\hline
     $p_B$(dB) & 30 & 40 & 30 & 40 & 30 & 40 \\\hline
   $p_A=\bar p_{A,opt}$(dB) & 25 & 31 & 16.8 & 22.4 & 16.4 & 18.9 \\\hline
   $\texttt{M}_{\texttt{SIMO}}$ & -2.9 & -3.0 & -5.4 & -5.4 & -6.6 & -6.6 \\\hline
   $\texttt{M}_{\texttt{MISO}}$ & -2.8 & -2.8 & -4.9 & -5.2 & -5.9 & -6.1 \\\hline
   $\texttt{M}_{\texttt{MISO-SIMO}}$ & 0.13 & 0.17 & 0.006 & 0.02 & -0.02 & 0.008 \\\hline
   $\rho_{\texttt{MISO-SIMO}}$ & 86\% & 94\% & 61\% & 90\% & 15.3\% & 88.6\% \\
   \hline
 \end{tabular}
  \caption{Comparison of key numbers from Figs. \ref{fig:comp_optimized_miso_simo_steep_n_E_4}-\ref{fig:comp_optimized_miso_simo_steep_n_E_8} where $\bar p_{A,opt}$ is the median of the optimal $p_A$ for MISO-SIMO-STEEP.}
  \label{Table_2}
\end{table}

\begin{figure}[ht]
\begin{minipage}[b]{0.48\linewidth}
\centering
\includegraphics[width=4.5cm,height=4.5cm]{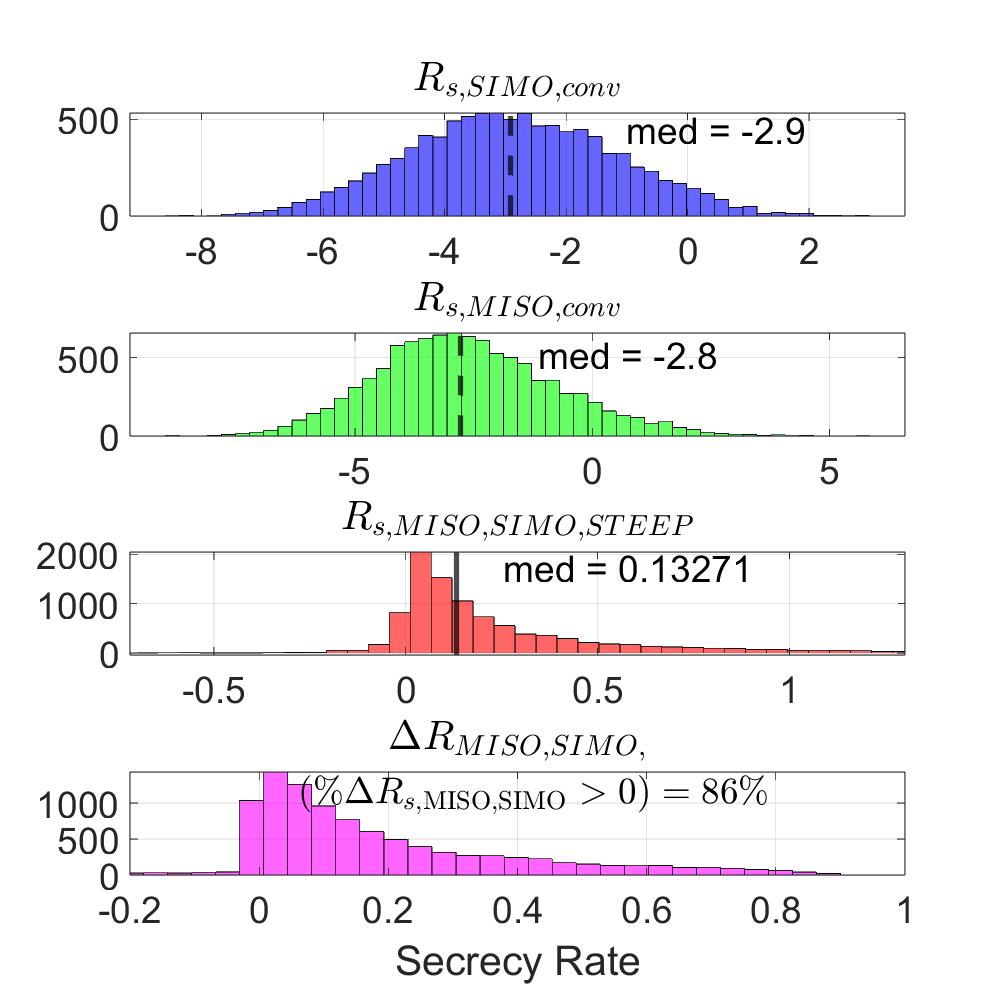}
\subcaption{$p_B=30$dB, $p_A=\bar p_{A,opt}=25$dB. }
\label{fig:secrecy_steep_pa_30}
\end{minipage}
\hspace{0.1cm}
\begin{minipage}[b]{0.48\linewidth}
\centering
\includegraphics[width=4.5cm,height=4.5cm]{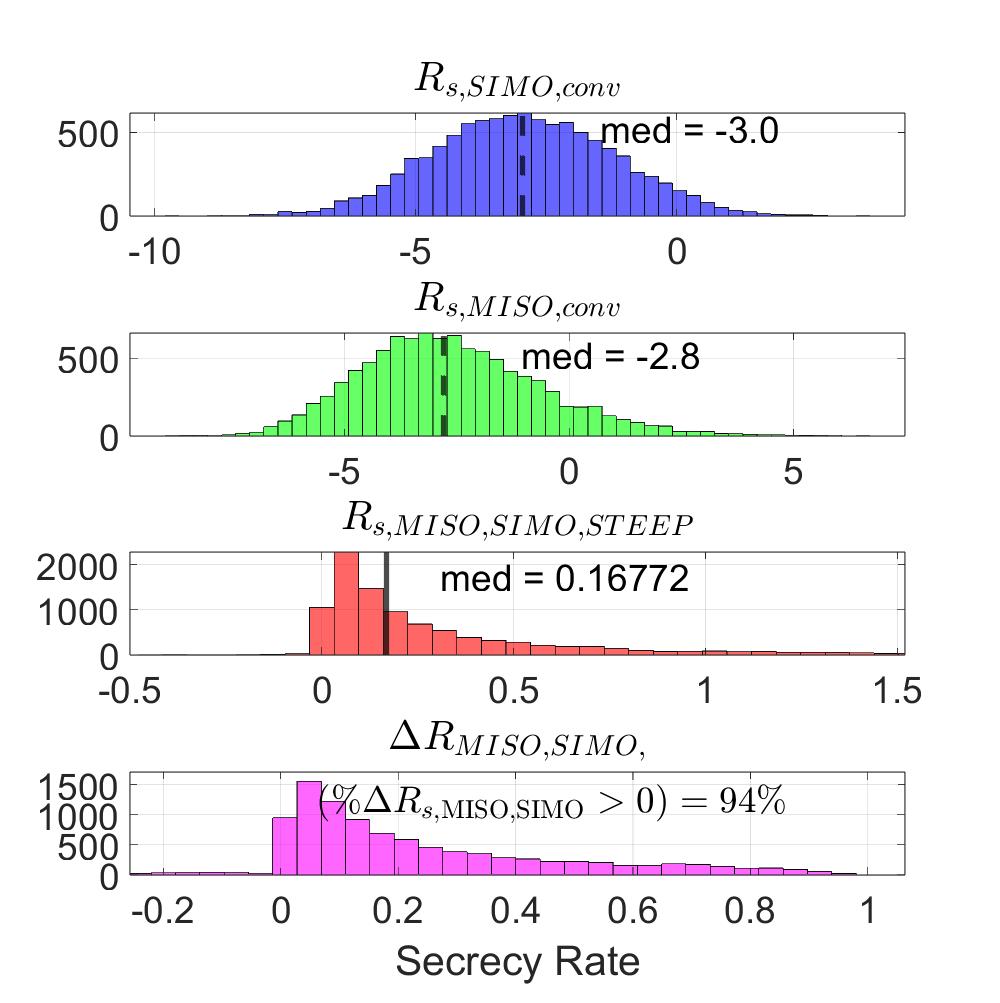}
\subcaption{$p_B=40$dB, $p_A=\bar p_{A,opt}=31$dB. }
\label{fig:secrecy_steep_pa_40}
\end{minipage}\\
\caption{Comparison of MISO-SIMO-STEEP with the conventional using $p_E=10$dB, $n_A=4$ and $n_E=4$.}
\label{fig:comp_optimized_miso_simo_steep_n_E_4}
\end{figure}
\begin{figure}[ht]
\begin{minipage}[b]{0.48\linewidth}
\centering
\includegraphics[width=4.5cm,height=4.5cm]{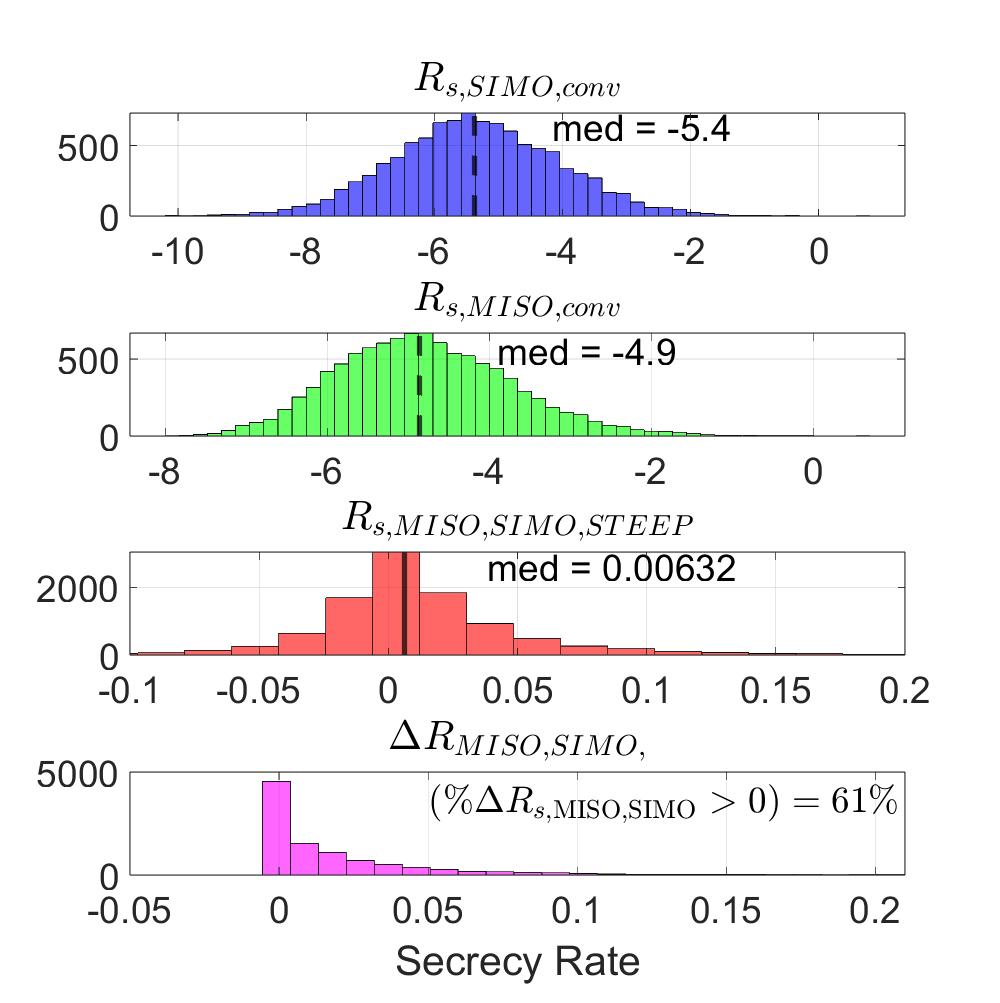}
\subcaption{$p_B=30$dB, $p_A=\bar p_{A,opt}=16.8$dB. }
\label{fig:secrecy_steep_pb_30}
\end{minipage}
\hspace{0.1cm}
\begin{minipage}[b]{0.48\linewidth}
\centering
\includegraphics[width=4.5cm,height=4.5cm]{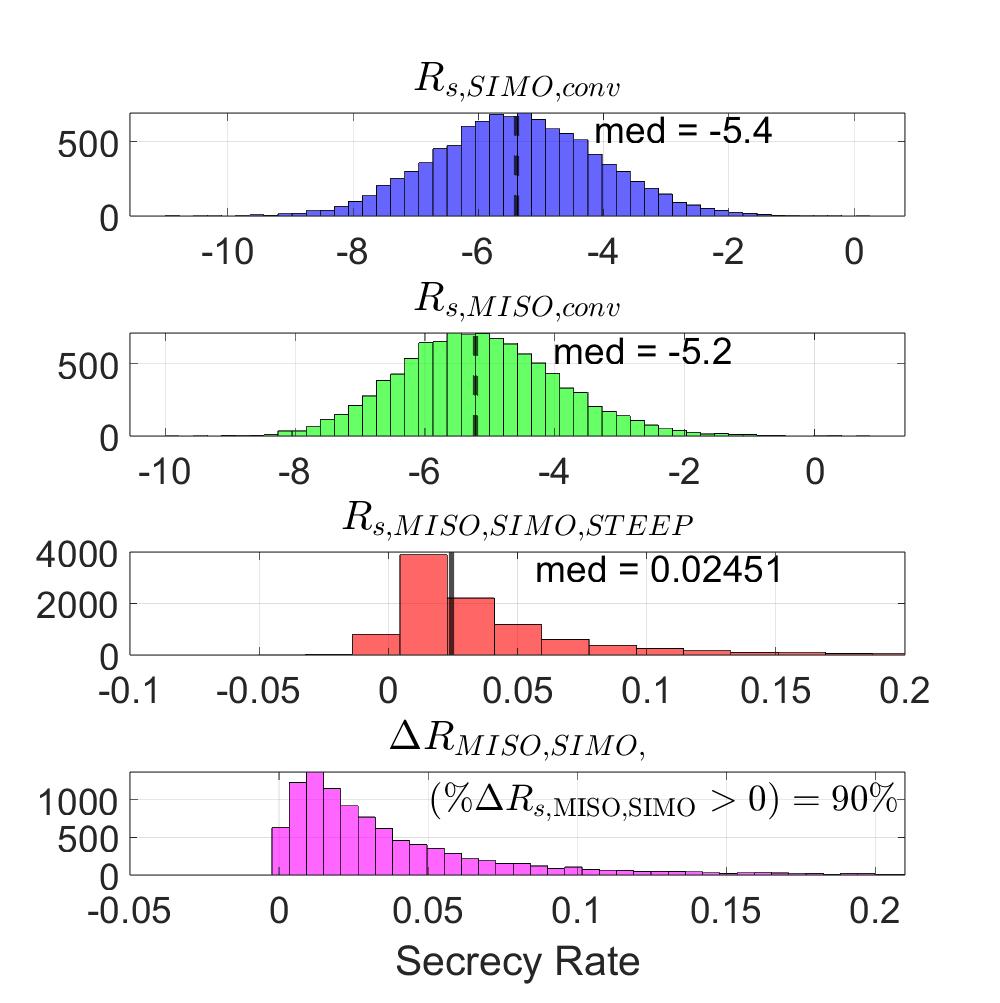}
\subcaption{$p_B=40$dB, $p_A=\bar p_{A,opt}=22.4$dB. }
\label{fig:secrecy_steep_pb_40}
\end{minipage}\\
\caption{Comparison of MISO-SIMO-STEEP with the conventional using $p_E=10$dB, $n_A=4$ and $n_E=6$.}
\label{fig:comp_optimized_miso_simo_steep_n_E_6}
\end{figure}
\begin{figure}[ht]
\begin{minipage}[b]{0.48\linewidth}
\centering
\includegraphics[width=5cm,height=4.5cm]{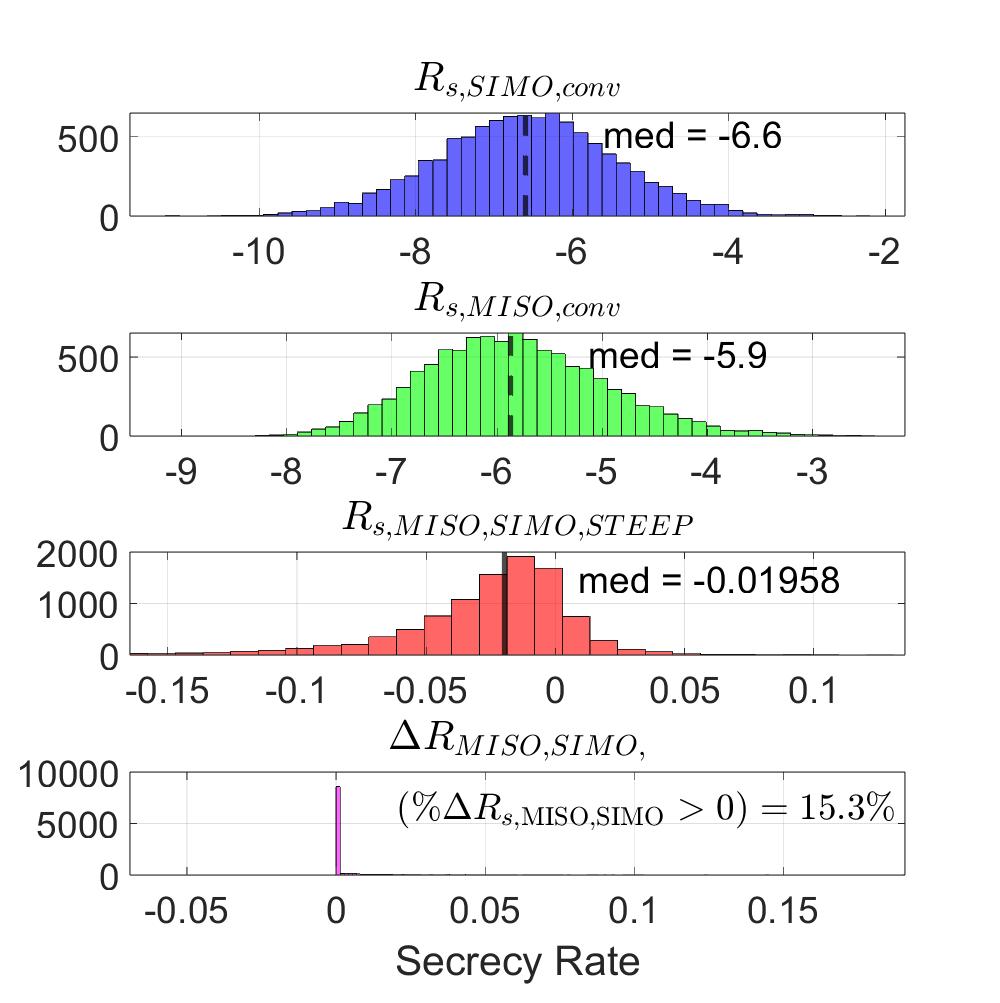}
\subcaption{$p_B=30$dB, $p_A=\bar p_{A,opt}=16.4$dB.}
\label{fig:secrecy_steep_pb_30}
\end{minipage}
\hspace{0.1cm}
\begin{minipage}[b]{0.48\linewidth}
\centering
\includegraphics[width=5cm,height=4.5cm]{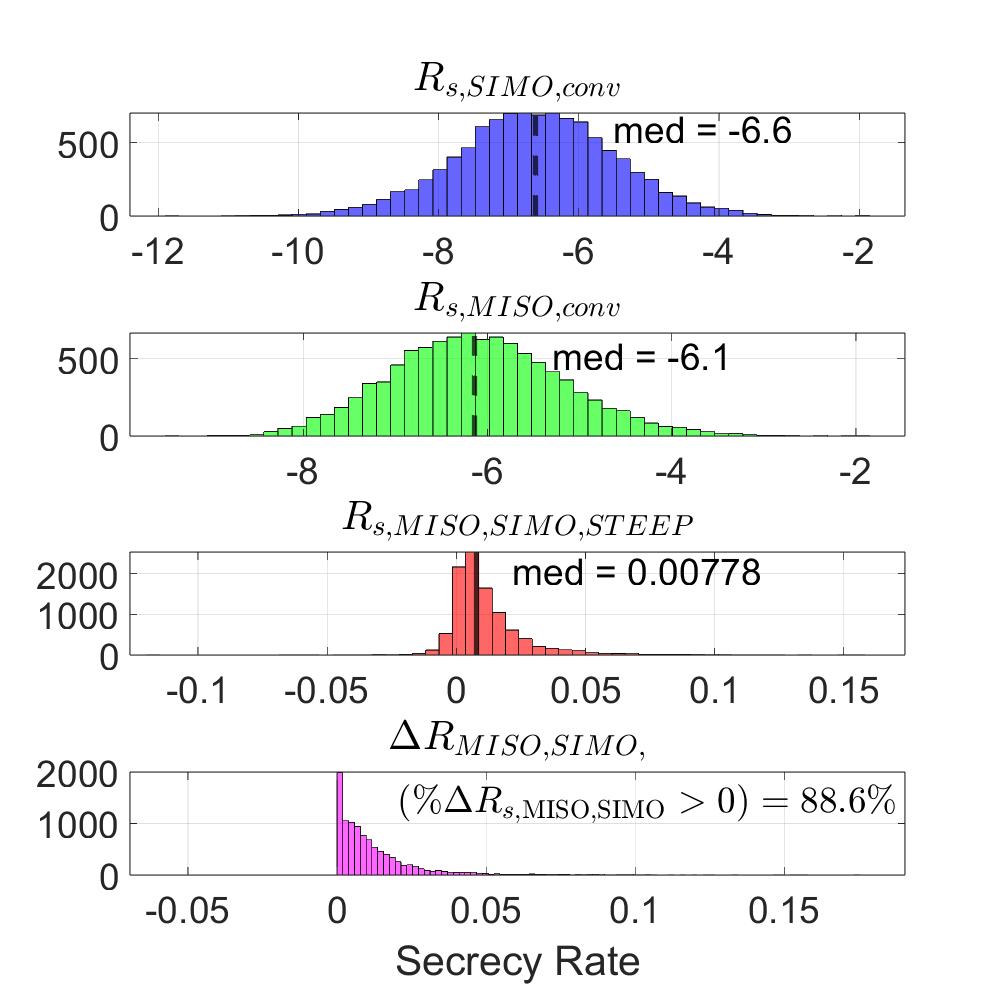}
\subcaption{$p_B=40$dB, $p_A=\bar p_{A,opt}=18.9$dB.}
\label{fig:secrecy_steep_pb_40}
\end{minipage}\\
\caption{Comparison of MISO-SIMO-STEEP with the conventional using $p_E=10$dB, $n_A=4$ and $n_E=8$.}
\label{fig:comp_optimized_miso_simo_steep_n_E_8}
\end{figure}

Shown in Figs. \ref{fig:comp_optimized_miso_simo_steep_n_E_4}, \ref{fig:comp_optimized_miso_simo_steep_n_E_6} and \ref{fig:comp_optimized_miso_simo_steep_n_E_8} are the histograms of the individual terms in \eqref{eq:Delta_M_S_STEEP} for $n_E=4$, $n_E=6$ and $n_E=8$ respectively. The key numbers from these figures are summarized in Table \ref{Table_2}.

Although most of the patterns in Figs  \ref{fig:comp_optimized_miso_simo_steep_n_E_4}-\ref{fig:comp_optimized_miso_simo_steep_n_E_8} are similar to those in Figs. \ref{fig:comp_optimized_simo_miso_steep_n_E_4}-\ref{fig:comp_optimized_simo_miso_steep_n_E_8}, we see that the secrecy rate advantage of MISO-SIMO-STEEP over the conventional SIMO and MISO schemes tends to be more than that of SIMO-MISO-STEEP. This is because subject to a dominant phase-2 power, the secrecy rate of STEEP is mostly governed by the phase-1 channels. See Proposition 2 in \cite{Hua_STEEP_2025}. Here the phase-1 channel for the users relative to Eve's phase-1 receive channel for MISO-SIMO-STEEP tends to be stronger than that for SIMO-MISO-STEEP. Also note that when Alice serves as the transmitter in phase 1, she injects artificial noise to make Eve's receive channel weaker. But when Bob serves as the transmitter in phase 1, he cannot inject any artificial noise (other than the probing signal) due to his single antenna.

It should be stressed that STEEP is meant for strongly eavesdropped channels where Eve's receive channel is stronger than the legitimate receiver's channel. If this is not the case, the conventional schemes should be used. In fact, for a given phase 2 power, if Eve's channel in phase 2 is weaker than user's channel in phase 2, the optimal phase 1 power of STEEP is typically zero, which reduces STEEP (in phase 2) to a conventional scheme.  Tables \ref{Table_1}-\ref{Table_2} also show that for a given phase 2 power (e.g., 30dB or 40dB), the median of the optimal power in phase 1 for SIMO-MISO-STEEP is much smaller than that for MISO-SIMO-STEEP. This is because it is more likely that the MISO user's channel in phase 2 (after optimal beamforming) is stronger than the MIMO Eve's channel in phase 2 than that the SIMO user's channel in phase 2 is stronger than the SIMO Eve's channel in phase 2.

 In the above examples, we assumed $p_E=10$dB. If $p_E$ increases, $p_A$ and $p_B$ need to increase accordingly in order for STEEP in either form to maintain a positive secrecy rate. As discussed earlier in section \ref{sec:complexity}, the received jamming powers from Eve (for which $p_E$ is a factor) could be measured online in real time by users before the values of $p_A$ and $p_B$ are chosen. We also chose $\gamma_A=0.5$ for all relevant schemes. The optimal $\gamma_A$ for SIMO-MISO-STEEP has been found to be one typically while the optimal $\gamma_A$ for MISO-SIMO-STEEP has been found to be typically around $1/n_A$. A good explanation of this remains open. But MISO-SIMO-STEEP with $\gamma_A=0.5$ still typically outperforms SIMO-MISO-STEEP with $\gamma_A=1$ in terms of secrecy rates.

\section{Conclusion}\label{sec:conclusion}
 A detailed investigation of STEEP for secure communications over SIMO and MISO channels between a multi-antenna Alice and a single-antenna Bob has been presented where a multi-antenna Eve capable for jamming and eavesdropping in full-duplex is assumed to be present. Subject to optimal jamming and eavesdropping by Eve against both Alice and Bob (also assuming that Eve knows all channels), the achievable secrecy rate of STEEP in either the SIMO-MISO-STEEP form or the MISO-SIMO-STEEP form has been analyzed, which is shown to be positive subject to a sufficiently large power from user in phase 2, regardless of the channel condition between every pair of nodes and finite jamming power from Eve. The connections between STEEP in either form and the conventional SIMO and MISO schemes have been highlighted. Assuming random fading channels, the secrecy rate of STEEP in either form has been compared via an extensive simulation with that of the conventional SIMO and MISO schemes. Our results have shown that, provided the phase-2 power from user dominates the jamming power from Eve, STEEP with a correspondingly proper phase-1 power (typically much smaller than the phase-2 power) outperforms the conventional schemes in achieving a positive secrecy rate especially when Eve has more antennas than Alice. This is consistent with the general theory shown in \cite{Hua_STEEP_2025} since the jamming noise from Eve can be viewed as an additional noise at user. The comparison was done in terms of the secrecy rate of STEEP in bits per round-trip channel use and the sum secrecy rate of the conventional SIMO and MISO schemes in bits per channel use subject to the same set of transmission powers from all nodes (Alice, Bob and full-duplex Eve). While STEEP beats in secrecy rate the classic wiretap channel scheme exploited in numerous prior works such as \cite{Thien2022}-\cite{Abdalla2023} subject to an aggressive and/or strong Eve, more tailored applications of STEEP for the specific settings in those prior works remain unexplored. More attention in the future should be paid to heavily eavesdropped channels with  aggressive eavesdroppers (who cannot be further weakened by other affordable means), for which this paper has provided an example study. This is because Eve's channel and its number of antennas in many situations are largely unknown and hence have to be assumed to be as aggressive as practically possible.

 For readers interested in applying STEEP to some specific network architecture, it is important to note that the channel model adopted in this paper is quite generic. For example, the channel vector $\mathbf{h}_{A,B}$ in \eqref{eq:yA1k} represents a composite channel response vector from the baseband symbols transmitted by Bob to the baseband symbols received by Alice, which include all effects of radio frequency chains, multipath propagation, passive and/or active reconfigurable intelligence surfaces \cite{Sun2024}, metasurfaces \cite{Sun2025}, and/or  the steering vectors on high altitude platforms \cite{Lin2025}. Clearly, for different network architecture (with or without any optimization to maximize user's channel gain and/or minimize Eve's channel gain), the statistical behaviour of the generic channel vectors and matrices defined in this paper would be different, and hence so would be the statistical behavior of the secrecy rate of STEEP. Many new problems of STEEP, including user's channel estimation errors and fast-fading channels, still await to be investigated.


\end{document}